\documentclass[11pt]{article}
\usepackage[utf8]{inputenc}

\usepackage{amsmath,amsfonts,amssymb,amsthm}
\usepackage[margin=1in]{geometry}
\usepackage{graphicx}
\usepackage{float}
\usepackage[colorlinks]{hyperref}
\usepackage{mathrsfs}
\usepackage{mathtools}
\mathtoolsset{centercolon}
\usepackage[numbers,comma,sort&compress]{natbib}
\usepackage[linesnumbered,ruled]{algorithm2e}
\usepackage{bm}
\usepackage{bbm}
\usepackage{subcaption}
\usepackage[qm]{qcircuit}

\newcommand{\C}{{\mathbb{C}}}

\newcommand{\R}{{\mathbb{R}}}
\newcommand{\Z}{{\mathbb{Z}}}

\newcommand{\abs}[1]{\left\lvert#1\right\rvert}

\def\C{{\mathbb C}}
\def\R{{\mathbb R}}

\def\Z{{\mathbb Z}}

\def\le{\leqslant}
\def\ge{\geqslant}

\newcommand{\rd}{\mathrm{d}}

\newcommand{\Or}{\mathcal{O}}

\numberwithin{equation}{section}

\DeclareMathOperator{\poly}{poly}

\newcommand{\D}{{\mathscr{D}}}

\renewcommand{\d}{\mathrm{d}}

\newcommand{\range}[1]{[{#1}]}
\newcommand{\rangez}[1]{[{#1}]_0}

\newtheorem{problem}{Problem}
\newtheorem{theorem}{Theorem}[section]
\newtheorem{lemma}[theorem]{Lemma}
\newtheorem{corollary}[theorem]{Corollary}

\theoremstyle{definition}

\newtheorem{remark}[theorem]{Remark}

\numberwithin{equation}{section}
\allowdisplaybreaks

\newcommand\norm[1]{\left\lVert#1\right\rVert}

\numberwithin{equation}{section}

\newcommand{\eq}[1]{(\ref{eq:#1})}
\renewcommand{\sec}[1]{\hyperref[sec:#1]{Section~\ref*{sec:#1}}}
\newcommand{\app}[1]{\hyperref[app:#1]{Appendix~\ref*{app:#1}}}
\newcommand{\thm}[1]{\hyperref[thm:#1]{Theorem~\ref*{thm:#1}}}
\newcommand{\lem}[1]{\hyperref[lem:#1]{Lemma~\ref*{lem:#1}}}
\newcommand{\cor}[1]{\hyperref[cor:#1]{Corollary~\ref*{cor:#1}}}
\newcommand{\prb}[1]{\hyperref[prb:#1]{Problem~\ref*{prb:#1}}}
\newcommand{\fgr}[1]{\hyperref[fgr:#1]{Figure~\ref*{fgr:#1}}}
\newcommand{\tab}[1]{\hyperref[tab:#1]{Table~\ref*{tab:#1}}}

\newcommand{\beq}{\begin{equation}}
\newcommand{\eeq}{\end{equation}}
\newcommand{\beqa}{\begin{eqnarray}}
\newcommand{\eeqa}{\end{eqnarray}}
\newcommand{\bra}[1]{\ensuremath{{\langle{#1}|}}}
\newcommand{\ket}[1]{\ensuremath{{|{#1}\rangle}}}

\newcommand{\Idnd}{\ensuremath{\mathbb{I}_{n_d\times n_d}}}

\usepackage{soul}

\makeatletter
\let\@@magyar@captionfix\relax
\makeatother
\begin{document}

\title{Efficient quantum algorithm for nonlinear \\reaction-diffusion equations and energy estimation}

\author{
Jin-Peng Liu$^{1,2,3,4}$,\quad Dong An$^{1}$,\quad Di Fang$^{5}$,\quad Jiasu Wang$^{3}$,\quad
Guang Hao Low$^{6}$, \quad Stephen Jordan$^{6, 7}$\\
\footnotesize $^{1}$ Joint Center for Quantum Information and Computer Science, University of Maryland, MD\\
\footnotesize $^{2}$ Simons Institute for the Theory of Computing, Berkeley, CA\\
\footnotesize $^{3}$ Department of Mathematics, University of California, Berkeley, CA\\
\footnotesize $^{4}$ Center for Theoretical Physics, Massachusetts Institute of Technology, Cambridge, MA\\
\footnotesize $^{5}$ Department of Mathematics and Duke Quantum Center, Duke University, Durham, NC\\
\footnotesize $^{6}$ Microsoft Quantum, Redmond, WA\\
\footnotesize $^{7}$ Google Quantum AI, Santa Barbara, CA\\
}
\date{}
\maketitle

\begin{abstract}
Nonlinear differential equations exhibit rich phenomena in many fields but are notoriously challenging to solve. Recently, Liu et al. \cite{LKK20} demonstrated the first efficient quantum algorithm for dissipative quadratic differential equations under the condition $R < 1$, where $R$ measures the ratio of nonlinearity to dissipation using the $\ell_2$ norm. Here we develop an efficient quantum algorithm based on \cite{LKK20} for reaction-diffusion equations, a class of nonlinear partial differential equations (PDEs). To achieve this, we improve upon the Carleman linearization approach introduced in \cite{LKK20} to obtain a faster convergence rate under the condition $R_D < 1$, where $R_D$ measures the ratio of nonlinearity to dissipation using the $\ell_{\infty}$ norm. Since $R_D$ is independent of the number of spatial grid points $n$ while $R$ increases with $n$, the criterion $R_D<1$ is significantly milder than $R<1$ for high-dimensional systems and can stay convergent under grid refinement for approximating PDEs. As applications of our quantum algorithm we consider the Fisher-KPP and Allen-Cahn equations, which have interpretations in classical physics. In particular, we show how to estimate the mean square kinetic energy in the solution by postprocessing the quantum state that encodes it to extract derivative information.
\end{abstract}

\section{Introduction}

Nonlinear partial differential equations (PDEs) of reaction-diffusion type are widespread and have many applications, ranging from biology and ecology to data science. Exhibiting rich phenomena, reaction-diffusion equations have been applied to describe biological transport networks such as leaf venations and blood flow \cite{LuHu2022,CaiHu2013,caihupde1,caihupde2,caihupde3,caihupde4,caihupde5,caihupde6,caihupdebook,FJMP19}, predator-prey dynamics in interacting populations \cite{Garvie2007, Malchow2007, Petrovskii1999}, prediction of brain functions and tumor growth \cite{Lefevre2010, Habib2003}, the formation of the Turing patterns in tissues and organs \cite{Murray2001, Murray2002, Genieys2006, Meinhardt1982}, dendritic colony growth \cite{Golding1998, Mimura2000}, complex chemical processes such as combustion \cite{Berestycki1985, Zeldovich1985, Poinsot2005, Perthame2012} and calcium dynamics \cite{Means2006}. Reaction-diffusion equations have also been applied to data classification \cite{BF12,BF16,MKB13}, and image segmentation and inpainting \cite{BEG06,DB08,EM03,ET06}. In many cases, the underlying model can be viewed as an energy optimization procedure, with the reaction-diffusion equations as the gradient flow. Such reaction-diffusion equations inherit the property that energy decays with time. Moreover, the maximum principle is satisfied, which states that if the initial and boundary conditions are bounded by a certain constant, then the entire solution remains bounded (in the $L^\infty$ sense) for all time. When designing numerical approximation schemes it is often of great interest to maintain these properties exactly. 

Attempts to solve such PDEs on classical computers are hampered by the so-called \textit{curse of dimensionality}, in which computational complexity grows exponentially with spatial dimension \cite{Bel57}. For example, in $d$ dimensions, if each coordinate is discretized by $n$ grid points, the grid will have $\Omega(n^d)$ grid points.

Recent advances in quantum computing offer a fresh approach to the efficient solution of such high-dimensional problems. Quantum algorithms have been developed to prepare a quantum state encoding the solution to an $n^d$-dimensional linear system, while in some cases only requiring quantum circuits of complexity $\poly(d,\log n)$ \cite{Amb12, AL19, CKS15, GSLW18, HHL08, LT19, SSO18, TAWL20,CAS21}. Such quantum algorithms have been applied to address high-dimensional problems governed by linear ODEs \cite{Ber14,BCOW17,CL19,FLT23} and PDEs \cite{CJS13,CPPTK13,MP16,CJO19,CLO20,ESP19,LMS20}. 

It has been a longstanding open problem to understand the capability of quantum computers to solve nonlinear differential equations. An early work proposed a quantum algorithm for ODEs that simulates polynomial nonlinearities by employing multiple copies of the solution \cite{LO08}. For $n$-dimensional systems of polynomial ordinary differential equations, this quantum algorithm scales as $\poly(\log n, 1/\epsilon^{T})$. Finding quantum algorithms with polynomial scaling in $T$ for solving nonlinear differential equations remained an open problem. Furthermore, complexity-theoretic arguments indicate that this should not be achievable in the most general case \cite{AL98,Aar05,CY16}, but rather will require exploiting specific properties of restricted classes of nonlinear differential equations.

Recently, in \cite{LKK20}, a quantum algorithm based on Carlemann linearization \cite{Car32,KS91,FP17} was proposed for solving a class of nonlinear differential equations 
\begin{equation}
    \frac{dU}{dt} = F_1 U + F_2 U^{\otimes 2} + F_0(t).
\label{eq:equation}
\end{equation}
Here we assume $F_1$ is dissipative, i.e. all the eigenvalues of $F_1$ are negative. We are given an initial condition $U(0) = U_{\mathrm{in}}$, an error tolerance $\epsilon$, and a time-duration $T$. The efficiency of the algorithm depends on $R$, defined as 
\begin{equation}
R = \frac{\|F_2\|}{|\lambda_1|}\|U_{\mathrm{in}}\|,
\end{equation}
where $\lambda_1$ is the largest eigenvalue of $F_1$. The quantity $R$ is used to quantify the relative strength of the nonlinearity and forcing to the linear dissipation according to the $\ell_2$ norm. A more general definition of $R$ for polynomial differential equations is given in \eq{A}. Under the condition $R < 1$, the algorithm has complexity $\Or \left( \frac{T^2 q}{g \epsilon} \poly(\log T, \log n, \log 1/\epsilon) \right)$, where $q = \|U_{\mathrm{in}}\|/\|U(T)\|$, and $g = \|U(T)\|$. This quadratic scaling with $T$ was an exponential improvement over prior quantum algorithms for solving nonlinear differential equations. The error dependence of quantum Carleman linearization was subsequently improved from $\poly(1/\epsilon)$ to $\poly(\log 1/\epsilon)$ in \cite{krovi2022improved}, by assuming the log-norm of the dissipation matrix is negative rather than exploiting a diagonalizability condition.

Various quantum algorithms for nonlinear differential equations have also been investigated recently based on Koopman-von Neumann linearization \cite{DS20,Jos20,LPG20,ESP21,TJ21} and the related level set formalism \cite{Jin2022quantum}. Others have been proposed based on the non-Hermitian Hamiltonian approach \cite{DS20,LPG20,DS21} and the homotopy perturbation approach \cite{XWG21}. Carleman linearization can be treated as a particular Koopman-von Neumann linearization, while the non-Hermitian Hamiltonian approach is inspired by quantum simulation. The quantum algorithm of \cite{XWG21} combines the homotopy perturbation method with the high-precision quantum linear ODE solver of \cite{BCOW17}, achieving complexity that scales linearly in $T$ and polylogarithmically in $1/\epsilon$. This is shown under the condition $K = 4R < 1$, which is stricter than the condition $R < 1$ used in \cite{LKK20}. The quantum algorithm of \cite{Jin2022quantum} is based on the level set method, which maps a nonlinear differential equation into a linear differential equation describing the dynamics of the level sets of the solution to the original nonlinear differential equation. Given a specific construction for the encoding of initial data, the quantum algorithm of \cite{Jin2022quantum} encodes the level set function from which physical observables can be estimated corresponding to multiple initial conditions.

Many of the quantum algorithms proposed for solving differential equations depend on solving a high dimensional linear system, and their complexities are thus determined by the condition number of this system. Deriving bounds on this condition number based on the properties of the original differential equations is challenging and largely unsolved. Quantum complexity lower bounds on simulating nonlinear quantum dynamics \cite{CY16} or classical dynamics \cite{LKK20} show that this condition number becomes exponential in the worst case.

In this paper, we extend upon the quantum algorithm of \cite{LKK20} by adapting the Carlemann linearization approach to the context of reaction-diffusion PDEs. 
We also show that reaction-diffusion equations are still tractable on quantum computers even for larger $R$ under the Maximum Principle, ruling out the worst-case exponential time complexity in \cite{LKK20}.
Finally, we conduct several numerical experiments for Fisher-KPP equations and Allen-Cahn equations to verify the convergence rate and efficiency of the improved Carleman linearization.

We compare our improved quantum Carleman linearization algorithm to the original one \cite{LKK20} in \tab{compare}. Both quantum Carleman linearization methods solve an $n_d$-dimensional system of ordinary differential equations with initial condition $U_{\mathrm{in}}$, for a given evolution time $T$ and normalized $\ell_2$ error tolerance $\epsilon$. (In the present work, we consider this set of ODEs as arising from the discretization of a PDE.)
The quantum algorithm in \cite{LKK20} solves nonlinear dissipative differential equations of the form in \eq{equation}.
Our new quantum algorithm solves reaction-diffusion equations, where $F_1$ is Laplacian and $F_0 = 0$, but the $F_2 U^{\otimes 2}$ term is instead allowed to be a high-degree polynomial. Thus, the class of problems we consider here is neither a strict generalization nor a strict special case of that considered in \cite{LKK20}.

Our new quantum algorithm produces a Feynman-Kitaev history state encoding the full time-evolution of the solution $U(t): t \in[0,T]$, whereas the algorithm of \cite{LKK20} produces a quantum state encoding the final value $U(T)$. The history state we produce corresponds to the gradient flow of the energy functional. We can post-process this state to extract derivative information which can be interpreted as kinetic energy in classical physical systems modeled by the nonlinear PDE.

The quantum algorithm of \cite{LKK20} has time-complexity proportional to $q/g$ where $q = \|U_{\mathrm{in}}\|/\|U(t)\|$ and $g = \|U(T)\|$ is the final norm of the solution. Our new algorithm instead has complexity proportional to $\|U_{\mathrm{in}}\|/G$, where $G = \frac{1}{T}\int_0^T \|U(t)\|$ is the time-averaged norm of the solution. In some cases, $1/G$ can be much smaller than $1/g$. For example, the solution $u(t) = e^{-t}$ arises in many homogeneous dissipative differential equations. In this case $1/g = e^T$ and $1/G = \Omega (T)$.

Polynomial complexity is here shown under the assumption $R_D < 1$, where $R_D$ is a ratio of nonlinearity to dissipation in $\ell_\infty$ norm, whereas the algorithm of \cite{LKK20} requires $R < 1$, where $R$ is a ratio of nonlinearity to dissipation in $\ell_2$ norm. The latter is a stronger assumption, not well suited for the solution of high-dimensional PDEs because it grows under grid refinement, whereas $R_D$ converges to a constant. Specifically, in the limit where each spatial dimension is discretized into $n \to \infty$ steps, the number of lattice sites scales as $n^d$ and the $\ell_2$ norm of the discretized solution vector scales as $n^{d/2}$, which leads to divergent $R$ (see \eq{trapezoidal} and \eq{rescaled-norm} for detailed discussion).

The solutions to general nonlinear differential equations can have exponentially growing norms, which would result in an exponential complexity for the algorithm introduced here. However, we rule this out for reaction-diffusion equations by establishing an upper bound on the $\ell_{\infty}$ norm of the solution independent of $T$ as shown in \thm{lower}. 

We also study the extraction of classical information of practical interest from the history state produced by our algorithm. First, from the quantum state, we can directly estimate the mean square amplitude over a specific sub-domain, which can be understood as the portion of a physical observable on this sub-domain. Our approach is a direct application of amplitude estimate technique~\cite{BHM02} and can achieve a quadratic speedup in precision over standard classical Monte Carlo sampling. Secondly, we show how to estimate the portion of the kinetic energy on a specific sub-domain by developing a quantum algorithm that can transfer a quantum state with function values to a quantum state encoding its partial derivatives. This algorithm is based on the discrete Fourier transform. It only requires $\Or(1)$ uses of quantum Fourier transform (QFT) and input oracle of a diagonal matrix, and can potentially be of independent interest in other problems such as quantum optimization algorithms. Our main results are summarized in \tab{application}. Second, we briefly discuss the potential advantages of the history state compared to the final state. In particular, the history state structure allows us to estimate the time when the system reaches equilibrium and run a pre-diagnosis procedure to avoid possible exponential overhead brought by the fast decay of the solution. 

\begin{table}[ht]
    \centering
    \scriptsize{
    \renewcommand{\arraystretch}{1.25}
    \begin{tabular}{|c|c|c|c|c|c|c|}
       \hline
       \textbf{Algorithm} & \textbf{Model}  & \textbf{Output} & \textbf{Condition} & \textbf{Complexity} & \textbf{Grid refinement} 
       \\
       \hline
       \cite[Theorem 1]{LKK20}  & dissipative quadratic ODEs & final state & $R < 1$ \eq{A} & $T^2q/(g\epsilon)$ & $R = \Or(n_d^{1/2})$ 
       \\
       \hline
       \thm{main}  & polynomial R-D systems \eq{RDODE} & history state & $R_D < 1$ \eq{A1} & $T^2\|U_{\mathrm{in}}\|^{2N}/(G\epsilon)$ & $R_D = \Or(1)$ 
       \\
       \hline
    \end{tabular}
    }
    \caption{\small
    Comparison between original and improved results of quantum Carleman linearization (QCL). We consider an $n_d$-dimensional system of ODEs, given the initial condition $U_{\mathrm{in}}$, the evolution time $T$ and normalized $\ell_2$ error tolerance $\epsilon$. We denote $g$, $q$, and $G$ as $\ell_2$ norm of the solution, $\ell_2$ relation between initial and final solutions, and time-average $\ell_2$ norm of the solution. R-D refers to reaction-diffusion. The condition $R<1$ implies the weaker condition $R_D<1$. When $R_D<1\le R$, we usually choose a suitably small value of the truncation number of Carleman linearization $N$ to reduce the cost resulting from the prefactor $\|U_{\mathrm{in}}\|^{2N}$. Grid refinement refers to the case that the system of ODEs is discretized from a PDE using an increasing number $n_d$ of grid points. Logarithmic factors in the complexities are omitted. 
    }
    \label{tab:compare}
\end{table}

\begin{table}[ht]
    \centering
    \scriptsize{
    \renewcommand{\arraystretch}{1.25}
    \begin{tabular}{|c|c|c|}
       \hline
       \textbf{Quantum algorithm} & \textbf{Output} & \textbf{Query complexity} 
       \\
       \hline
        \thm{app_deriv_state_smooth} & history state of derivatives & $\mathcal{\widetilde{O}}\left(\sqrt{d}\|\vec{f}\|/\|\vec{\nabla  f}\|\right)$ \\\hline
        \cor{app_deriv} & kinetic energy ratio & $\mathcal{\widetilde{O}}\left(\sqrt{d}\|\vec{f}\|/(\|\vec{\nabla f}\|\epsilon)\right)$ \\\hline
    \end{tabular}
    }
    \caption{\small Summary of potential applications of a quantum state encoding the information of a smooth function in its amplitudes. Query complexity is the number of queries to the circuit preparing the quantum state encoding the function $f$. Here $d$ is the spatial dimension of the function, $\epsilon$ is the tolerated level of error. We use $\vec{f}$ and $\vec{\nabla f}$ to denote the unnormalized vectors of $f$ and $\nabla f$ evaluated at discrete grid points, respectively. }
    \label{tab:application}
\end{table}

The paper is organized as follows. \sec{pde} introduces the background of reaction-diffusion equations. \sec{Carleman} develops the Carleman linearization with $\ell_2$ and $\ell_{\infty}$ convergence analysis. \sec{algorithm} presents the problem model and gives the quantum algorithm with a detailed complexity analysis. \sec{lower} establishes lower bound results. \sec{application} describes how our approach could be applied to kinetic energy estimation problems. Finally, we conclude with a discussion of the results and some possible future directions in \sec{discussion}.

\subsection{Preliminaries}

Here we denote the domain, boundaries, functions, and norms as follows.

We consider a $d$-dimensional hypercube as the spatial domain, denoting as $\D\coloneqq [0,1]^d$. We denote the spatial and time domain as $\D_T\coloneqq[0,1]^d\times(0,T]$. We also denote $\partial\D$ and $\partial\D_T$ as boundary domains of $\D$ and $\D_T$, respectively.

We consider a uniform spatial discretization on $\D$ and introduce $n$ discretization nodes for each coordinate. To represent it, we denote $\rangez{n} \coloneqq \{0,1,\ldots,n-1\}$ and a set of multi-indices as
\begin{equation}
\mathcal{I} \coloneqq \rangez{n}^d = \Bigl\{l=(l_1,\ldots,l_d) ~\Big|~ l_j\in \rangez{n}\Bigr\}.
\label{eq:indices}
\end{equation}
We then denote the set of uniform nodes as 
\begin{equation}
\chi \coloneqq \Bigl\{\chi_l ~\Big|~ l\in\mathcal{I}\Bigr\},
\label{eq:nodes}
\end{equation}
where $\chi_l$ maps index $l$ to the discretization node. The exact expression for $\chi_l$ depends on the boundary condition. For periodic boundary condition, $\chi_l$ is defined as
\begin{equation}
\chi_l \coloneqq \left(\frac{l_1}{n}, \ldots, \frac{l_d}{n}\right),
\label{eq:node}
\end{equation}
while for Dirichlet boundary condition, it is given by
\begin{equation}
\chi_l \coloneqq \left(\frac{l_1+1}{n+1}, \ldots, \frac{l_d+1}{n+1}\right).
\end{equation}
For convenience, we also introduce the set of boundary indices, which is defined as
\begin{equation}
\mathcal{B} \coloneqq \Bigl\{l=(l_1,\ldots,l_d) ~\Big|~ \chi_l \in \partial \D\Bigr\}.
\end{equation}

Let $u:\D_T \to \mathbb{R}$ be the solution to a PDE. We can discretize $u(x,t)$ on the set of uniform nodes $\chi$ to obtain an $n_d$-dimensional vector $U(t)$, where $n_d = n^d$. 
The vector's entries $U_1(t), U_2(t), \ldots, U_{n_d}(t)$ are the elements of $\{u(\chi_l,t)|l \in \mathcal{I}\}$, arranged according to the lexicographic order on $\mathcal{I}$.

We now discuss our notations for norms. For a vector $a = [a_1, a_2, \ldots, a_n]\in\R^n$, we denote the vector $\ell_p$ norm as
\begin{equation}
\|a\|_p: = \left({\sum_{k=0}^{n-1}|a_k|^p} \right)^{1/p}.
\end{equation}
For a matrix $A\in\R^{n\times n}$, we denote the operator norm $\norm{\cdot}_{p,q}$ induced by the vector $\ell_p$ and $\ell_q$ norms as
\begin{equation}
\|A\|_{p, q} : =\sup_{x \neq 0}  \frac{\norm{A x}_q }{\norm{x}_p} , \quad \|A\|_{p }: = \|A\|_{p, p}. 
\end{equation}
For a continuous scalar function $f: [0,T]\rightarrow\R$, we denote the $L^\infty$ norm as
\begin{equation}
\|f\|_\infty: = \max_{t \in [0,T]} |f(t)|.
\end{equation}
For a continuous scalar function $u: \overline{\D}_T\rightarrow\R$, for a fixed $t$, the $L^p$ norm of $u(\cdot,t)$ is given by
\begin{equation}
    \|u(\cdot,t)\|_{L^p(\D)} := \left( \int_\D  |u(x,t)|^p \, dx \right)^{1/p}.
\end{equation}
In particular, when no subscript is used, we mean $\norm{\cdot} = \norm{\cdot}_2$ for vector and matrix norms by default, and $\norm{\cdot} = \norm{\cdot}_{L^2}$ for function norm by default. 

For a continuous scalar function $u: [0,1]\times[0,T]\rightarrow\R$, which is discretized in space using uniform interpolation nodes, obtain an estimate of its $L^p$ norm as in a Riemann sum:
\begin{equation}
    \|u(\cdot,t)\|_{L^p(\D)}^p = \int_\D  |u(x,t)|^p \, dx \approx \sum_{k=0}^{n-1} \left( \left|u\left(\frac{k}{n},t\right)\right|^p \frac{1}{n} \right).
    \label{eq:trapezoidal}
\end{equation}
If we denote $U(t) = [u(\frac{0}{n},t), u(\frac{1}{n},t), \ldots, u(\frac{n-1}{n},t)]$, the RHS of \eq{trapezoidal} is $\frac{1}{n}\|U(t)\|_{p}^p$. This indicates that
\begin{equation}
    \|u(\cdot,t)\|_{L^p(\D)} \approx \frac{1}{n^{1/p}}\|U(t)\|_{p}.
    \label{eq:rescaled-norm}
\end{equation}
Thus, if $u(\cdot,t)$ is a given function in continuous one-dimensional space, then the $\ell_p$ norm of its spatially discretized function vector $U(t)$ increases under grid refinement as $n^{1/p}$. Similarly, for a $u: \overline{\D}_T\rightarrow\R$ with a general spatial dimension $d$, the $\ell_p$ norm of the spatially discretized function vector $U(t)$ increases as $n^{d/p}$. 

Note that when $p=\infty$, 
\begin{equation}
    \|u(\cdot,t)\|_{L^{\infty}(\D)} \approx \|U(t)\|_{\infty}.
    \label{eq:rescaled-norm-inf}
\end{equation}
That means the $\ell_{\infty}$ norm of the spatially discretized function vector $U(t)$ stays convergent in the continuum limit of $n \to \infty$.

For real functions $f,g: \R\rightarrow\R$, we write $f=\Or (g)$ if there exists $c>0$, such that $|f(\tau)|\le c|g(\tau)|$ for all $\tau\in\R$. We write $f=\Omega(g)$ if $g=\Or (f)$, and $f=\Theta(g)$ if both $f=\Or (g)$ and $g=\Or (f)$. We use $\widetilde \Or$ to suppress logarithmic factors in the asymptotic expression, i.e., $f=\widetilde \Or (g)$ if $f=\Or\Bigl(g\poly(\log g)\Bigr)$. We write $f=o(g)$ if $\limsup_{\tau\to\infty} \frac{|f(\tau)|}{|g(\tau)|}=0$, where $g$ is nonzero.

\section{Problem Settings} \label{sec:pde}

In this section, we introduce the class of nonlinear PDEs that we focus on and discuss its spatial discretization, as well as \textit{a priori} bounds on its solutions. We then introduce the problem statement with input and output settings.

\subsection{Reaction-diffusion equation} 
We focus on a class of nonlinear PDEs -- the reaction-diffusion equations
\begin{equation}
\frac{\partial u}{\partial t}(x,t) = D \Delta u(x,t) + f(u(x,t)),  
\label{eq:NPDE}
\end{equation}
where $u$ is a real-valued scalar function at position $x \in \D = [0,1]^d$ and time $t\in\R^+$, $f(u)$ is the nonlinear term and $D$ is a positive number.
Without loss of generality, denoting $\D = \D_1\times\D_2$, we are given homogeneous Dirichlet boundary conditions imposed on $d_1$ dimensions ($x\in\D_1$) and periodic boundary conditions imposed on $d_2$ dimensions ($x\in\D_2$, $d_1 + d_2 = d$)
\begin{align}
u(x,t) &= 0, \qquad\qquad\qquad x \in \partial \D_1, \\
u(x,t) &= u(x+v,t), \qquad v \in 0^{d_1}\times\Z^{d_2}, ~ x \in \partial \D_2.
\label{eq:boundary}
\end{align}
The solution $u(x,t)$ in \eq{NPDE} is the $L^2$ gradient flow of the free energy functional
\begin{equation}
E(u) = \frac{D}{2} \int |\nabla u|^2\d x + \int F(u) \d x 
\label{eq:energy-functional}
\end{equation}
with a potential satisfying
\begin{equation}
\frac{\partial F}{\partial u}(u) = - f(u).
\label{eq:potential}
\end{equation}
In other words, $f$ is a field driven by the potential $F$.

In this paper, we focus on the following specific reaction-diffusion equations 
\begin{equation}
\frac{\partial u}{\partial t}(x,t) = D \Delta u(x,t) + au(x,t) + bu^{M}(x,t),
\label{eq:RDPDE}
\end{equation}
with the integer $M\ge 2$. Without loss of generality, we assume $|a|, |b| = o(d)$. The motivation to consider this type of PDEs is two-fold: first, in physical and biological applications, a nonlinearity of this form is frequently encountered. For example, in the phase transition model (the so-called Allen-Cahn equation), $M = 3$ \cite{AllenCahn1979}, while $M = 2$ corresponds to the Fisher-KPP equation \cite{Fisher1937, Murray2001}. Furthermore, on quantum computers, it is a reasonable task to construct tensor powers of a quantum state, such as $u^{\otimes M}$, which exactly corresponds to the polynomial nonlinearity in \eq{RDPDE}. Although we do not do so in this paper, one might also consider an input model in which a more general nonlinearity $f$ is specified by an oracle $U_f$. 

\subsection{Spatial Discretization} \label{sec:spatial_dis}
Our approach to solving reaction-diffusion PDEs \eq{RDPDE} on quantum computers starts by performing spatial discretization to reduce to a problem of solving a system of nonlinear ODEs. Specifically, we apply the central difference discretization on \eqref{eq:NPDE} to obtain the $n^d$-dimensional polynomial ODE
\begin{equation}
\frac{\d U_j}{\d t} = D \sum_k (\Delta_h)_{jk} U_k +  f(U_j), \quad j \in \mathcal{I},
\label{eq:GAC}
\end{equation}
where $\Delta_h$ stands for the central difference of the Laplacian with homogeneous Dirichlet boundary condition or periodic boundary condition, defined as
\begin{equation}  \label{eqn:dis_lap_highD_tensor}
    \Delta_h = D_h \otimes \underbrace{I  \otimes \cdots \otimes I}_{d-1 \text{ terms}}+ I \otimes D_h \otimes \underbrace{I \otimes \cdots \otimes I}_{d-2 \text{ terms}} + \cdots +  \underbrace{I \otimes \cdots \otimes I}_{d-1 \text{ terms}} \otimes D_h.
\end{equation}
Here $D_h$ is the one-dimensional discrete Laplacian operator. For homogeneous Dirichlet boundary conditions, $D_h$ is
\begin{equation}\label{eqn:Dh_D1}
         D_h  = D_h^\text{Dir}:= (n+1)^2 \left(\begin{array}{ccccc}
        -2 & 1 & & &  \\
         1& -2 & 1 & & \\
          & \ddots& \ddots& \ddots& \\
           & & 1& -2 & 1\\
        & & & 1 & -2 \\
    \end{array}\right)_{n\times n}.
\end{equation}
We denote the eigenvalues of $D_h$ as $\mu_1\geq \mu_2\geq \cdots\geq \mu_n$. Specifically, $\mu_1=-4(n+1)^2\sin^2\left(\frac{\pi}{2n+2}\right)$.
For periodic boundary conditions, the one-dimensional discrete Laplacian operator is 
\[
    D_h=D_h^\text{per}  := n^2 \left(\begin{array}{ccccc}
        -2 & 1 & & & 1 \\
         1& -2 & 1 & & \\
          & \ddots& \ddots& \ddots& \\
           & & 1& -2 & 1\\
        1 & & & 1 & -2 \\
    \end{array}\right)_{n\times n}.
\]
In this case, the largest eigenvalue of $D_h^\text{per}$ is $\mu_1^\text{per} = 0$.

It is worth pointing out that besides serving as a numerical discretization of the corresponding PDE, the discrete reaction-diffusion equation is of interest unto itself. For example, the discrete Allen-Cahn equation has been applied to unsupervised and semi-supervised graph classification, graph cut minimization, social network segmentation, and image inpainting \cite{BF12,BF16,GB12,GGOB14,LB17,MKB13}.

We also introduce bounds on the solution to the discrete reaction-diffusion equations \eqref{eq:GAC}, which are used in the proof of later theorems.

\begin{lemma}[$\ell_\infty$ A priori Bounds on the Solution] \label{lem:apriori} Assume $f \in C^\infty(\mathbb{R})$ and has at least two distinct real-valued roots. Denote any two distinct roots of $f$ as $\gamma_1, \gamma_2$ with $\gamma_1 < \gamma_2$. Consider the solution $U(t) = (U_j)_{j \in \mathcal{I}}$ to \eqref{eq:GAC} with initial condition $U_j(0) = u_0(x_j)$ for all $j \in \mathcal{I}$. \\
(i) \textbf{Comparison Principle.} If the initial condition satisfies
\begin{equation}
\gamma_1 \le U_j(0) \le \gamma_2, \quad \text{for all $j \in \mathcal{I}$,}
\end{equation}
and so does the solution on the boundary indices $\mathcal{B}$, then the solution $U_j(t)$ remains bounded, that is,
\begin{equation}
    \gamma_1 \le U_j(t) \le \gamma_2, \quad \text{for all $t\geq 0$ and all $j \in \mathcal{I}$.}
\end{equation}
(ii) \textbf{Maximum Principle.} In particular, we denote $\gamma$ as the largest absolute value of roots of $f$. If the initial condition satisfies
\begin{equation} \label{eq:A0}
\norm{U(0)}_\infty \le \gamma, 
\end{equation}
then
\begin{equation} \label{eq:U_infty_bdd}
    \norm{U(t)}_\infty \le \gamma, \quad \text{for all $t\geq 0$.}
\end{equation}
\end{lemma}

We present the proof of \lem{apriori} in \sec{proof_apriori}. It is worth remarking that this result asserts that the solution $U(t)$ stays in the invariant set $[\gamma_1, \gamma_2]$, which is different from the type of estimate where $\norm{U(t)}_\infty \le \norm{U(0)}_\infty$. In fact, the solution $\norm{U(t)}_\infty$ can increase in time as depicted in \fgr{increase_u}. We also point out that the formation of this invariant region is precisely due to the nonlinear terms in $f$, and hence the nonlinear parts of the differential equation cannot be neglected, even if the solution remains small.

For the case of the specific reaction-diffusion equation \eq{RDPDE}, the roots of $f(u) = au + bu^{M} = 0$ gives the explicit expression $\gamma = (\frac{|a|}{|b|})^{\frac{1}{M-1}}$.

\begin{figure}
    \centering
    \begin{subfigure}{.5\textwidth}
    \centering
    \includegraphics[scale=0.3]{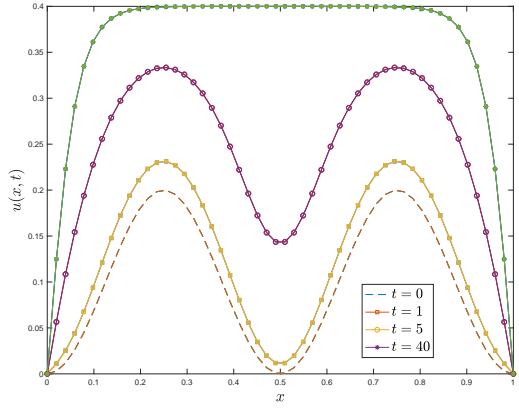}
    \caption{$f(u)=0.4u-u^2$ and $u_0(x)=0.1-0.1\cos(4\pi x)$.}
    \label{fig:sub1}
    \end{subfigure}
    \begin{subfigure}{.5\textwidth}
    \centering
    \includegraphics[scale=0.3]{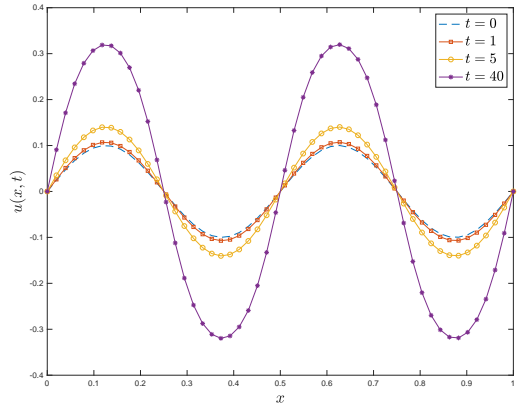}
    \caption{$f(u)=0.16u-u^3$ and $u_0(x)=0.1\sin(4\pi x)$.}
    \label{fig:sub2}
    \end{subfigure}
    \caption{Solutions $U = (U_j(t))$ to \eqref{eq:GAC} with homogeneous Dirichlet boundary condition for different nonlinear terms $f(u)$ and initial conditions $U_j(0) = u_0(x_j)$, where $D=0.005$. }
    \label{fgr:increase_u}
\end{figure}

\subsection{Problem statement} 
We are interested in solving high-dimensional reaction-diffusion equations with quantum computers. Given the initial condition described by a quantum state, we aim to provide a quantum state description of the solution given the evolution time $T$. 

The main computational problem we consider is as follows.
\begin{problem}\label{prb:node}
We consider an initial value problem of an $n_d$-dimensional polynomial ODE on $[0,T]$ as in \eq{NODE} 
\begin{equation}
\frac{\d{U}}{\d{t}} = F_1U + F_MU^{\otimes M}, 
\qquad
U(0) = U_{\mathrm{in}},
\label{eq:NODE}
\end{equation} 
Here $U=[U_1, \ldots, U_{n_d}]^{T}\in\R^{n_d}$, $U^{\otimes M}=[U_1^M, U_1^{M-1}U_2, \ldots, U_{n_d}^{M-1}U_{n_d-1}, U_{n_d}^M]^{T}\in\R^{n_d^M}$.
We assume $F_1$, $F_M$ are $s$-sparse\footnote{have at most $s$ nonzero entries in each column and row.}, $F_1$ is symmetric diagonalizable and eigenvalues of $F_1$ are negative, and $\|F_M\| \le |\lambda_1|$ by rescaling\footnote{given any nonlinear ODE, we can rescale $u\to\alpha u$ with a proper $\alpha$ to ensure $\|F_M\| \le |\lambda_1|$.}. We also know an a priori bound on the solution $\max_t\|U(t))\|_{\infty}\le\gamma$ as in \eq{U_infty_bdd}. 
We have oracles $O_{F_1}$, $O_{F_M}$ that provide the locations and values of the nonzero entries of $F_1$, $F_M$. We also know $\|U_{\mathrm{in}}\|=\|U(0)\|$ and have an oracle $O_x$ that maps $|00\ldots0\rangle \in\C^n$ to a quantum state proportional to $U_{\mathrm{in}}$.
Our goal is to produce a quantum state as a superposition of the solution at different timesteps
\begin{equation}
|\hat Y_{\mathrm{evo}}\rangle = \frac{1}{\|\hat Y_{\mathrm{evo}}\|}\hat Y_{\mathrm{evo}} = \frac{1}{\|\hat Y_{\mathrm{evo}}\|}\sum_{k=0}^{m}\hat y_1(kh)|k\rangle
\label{eq:history}
\end{equation}
with a sufficiently large $m$, where $\hat y_1(t)\in\R^{n_d}$ is a vector function that approximates $U(t)$, $\hat Y_{\mathrm{evo}} = \sum_{k=0}^{m}\hat y_1(kh)|k\rangle$ is a superposition of $\hat y_1(t)$ at different timesteps, and $\|\hat Y_{\mathrm{evo}}\| = \sqrt{\sum_{k=0}^{m}\|\hat y_1(kh)\|^2}$ is the normalization factor. 
\end{problem}

For the reaction-diffusion equation, we have the specific form of \eq{NODE} as
\begin{equation}
\frac{\d{U}}{\d{t}} = (D \Delta_h + aI)U + bU^{.M}, \qquad
U(0) = U_{\mathrm{in}}.
\label{eq:RDODE}
\end{equation}
The corresponding $F_1$ and $F_M$ in \eq{NODE} satisfy $F_1 = D \Delta_h + aI \in \R^{n_d\times n_d}$, and $F_M \in \R^{n_d\times n_d^M}$ maps $u^{\otimes M}$ to $bu^{M}$, with $U^{.M}=[U_1^M, U_2^M, \ldots, U_{n_d}^M]^{T}\in\R^{n^M}$, and henceforth $F_1$ and $F_M$ are $s$-sparse with $s=O(1)$. The representation of the Laplacian matrix $\Delta_h$ with mixed boundary conditions refers to \eq{boundary}, with $d_1$-dimensional Dirichlet boundary conditions and $d_2$-dimensional periodic boundary conditions ($d_1+d_2 = d$). We require the eigenvalues $\lambda_j$ of $F_1$ are negative, i.e., $ 4 D d_1 (n+1)^2\sin^2\left(\frac{\pi}{2n+2}\right) > a$.

The quantum state \eq{history} corresponding with $[U(0), U(h), \ldots, U(mh)]$, also known as the history state in \cite{KPE20}, describes a quantum-state evolution for the gradient flow of the free energy \eq{energy-functional}.

\section{Carleman linearization}
\label{sec:Carleman}

We aim to perform Carleman linearization on discretized nonlinear PDEs \eqref{eq:NODE} and then use quantum linear system solvers to obtain quantum states proportion to the solutions. In this section, we revisit the Carleman linearization procedure and then introduce the improved convergence result.

Defining $\widetilde y_j : = U^{\otimes j}$ for $j = 1, 2 \cdots$, one has
\begin{equation} \label{eqn:y_A}
\frac{d}{dt} \widetilde y_j = A_{j}^j \widetilde y_j  + A_{j+M-1}^j \widetilde y_{j+M-1},
\end{equation}
where
\begin{align}\label{eq:A_block}
A_{j}^j = \sum_{\nu=1}^j   \overset{\text{j factors}}{\overbrace{\Idnd \otimes \cdots \otimes \underset{\underset{\nu-\text{th position}}{\uparrow}}{F_1} \otimes  \cdots  \otimes \Idnd }}, \\
A_{j+M-1}^j = \sum_{\nu=1}^j   \overset{\text{j factors}}{\overbrace{\Idnd \otimes \cdots \otimes \underset{\underset{\nu-\text{th position}}{\uparrow}}{F_M} \otimes  \cdots  \otimes \Idnd }}.
\end{align}
Therefore, the Carleman linearization procedure gives rise to the following infinite-dimensional system
$\widetilde y'(t) = \widetilde A \widetilde y(t),$ where $\widetilde A$ is an infinite-dimensional block upper-triangular matrix
\begin{equation*}
\widetilde A := \begin{pmatrix}
A_1^1 & 0 & \cdots & 0 & A_M^1 & 0& \ldots & 0 & 0 & 0 & \ldots \\ 
0& A_2^2 & 0 & \cdots & 0 & A_{M+1}^2 &\ldots & 0 & 0&0 & \ldots \\
0 & 0 & A_3^3 & 0 & \cdots & 0 & A_{M+2}^3 \ldots  & 0 & 0 & \ldots \\   
\vdots &  \vdots &  \vdots &  \vdots  & \vdots &\vdots & \vdots & \vdots & \vdots & \vdots &\\
\end{pmatrix}. 
\end{equation*}

It follows from the definition of $A_k^j$ that the following inequalities are satisfied. 
\begin{lemma} \label{lem:ineq_Ajk}
For all $j\geq 1$, 
\begin{equation} \label{eqn:ineq_Ajk_Ajj}
    \norm{A^j_{j+M-1}}_2, \norm{A^j_{j+M-1}}_\infty \le j |b|,
\end{equation}
\begin{equation}
        \norm{e^{tA_j^j}}_2 \le e^{j\lambda_1 t} , 
\end{equation}
where $\lambda_1 \coloneqq D d_1 \mu_1 + a$, $\mu_1=-4(n+1)^2\sin^2\left(\frac{\pi}{2n+2}\right)$.
\end{lemma}

We then truncate the above infinite-dimensional system of linear ODEs at order $N$, thereby obtaining a finite system
\begin{equation}
  \frac{\d{\hat y}}{\d{t}} = A \hat y, \qquad
  \hat y(0) = \hat y_{\mathrm{in}}
\label{eq:LODE}
\end{equation}
with the upper triangular block structure
\begin{equation}
\frac{\d{}}{\d{t}}
  \begin{pmatrix}
    \hat y_1 \\
    \hat y_2 \\
    \vdots \\
    \vdots \\
    \vdots \\
    \hat y_{N-1} \\
    \hat y_N \\
  \end{pmatrix}
=
  \begin{pmatrix}
    A_1^1 & 0 & \cdots & 0 & A_M^1 &  &  \\
     & A_2^2 & \ddots &  & \ddots & \ddots &  \\
     &  & \ddots & \ddots &  & \ddots & A_N^{N-M+1} \\
     &  &  & \ddots & \ddots &  & 0 \\
     &  &  &  & \ddots & \ddots & \vdots \\
     &  &  &  &  & A_{N-1}^{N-1} & 0 \\
     &  &  &  &  &  & A_N^N \\
  \end{pmatrix}
  \begin{pmatrix}
    \hat y_1 \\
    \hat y_2 \\
    \vdots \\
    \vdots \\
    \vdots \\
    \hat y_{N-1} \\
    \hat y_N \\
  \end{pmatrix}.
\label{eq:UODE}
\end{equation}
Here, $\hat y_j=U^{\otimes j}\in\R^{n_d^j}$, $\hat y_{\mathrm{in}}=[U_{\mathrm{in}}; U_{\mathrm{in}}^{\otimes 2}; \ldots; U_{\mathrm{in}}^{\otimes N}]$, and $A_j^j \in \R^{n_d^j\times n_d^j}$, $A_{j+1}^j \in \R^{n_d^j\times n_d^{j+1}}$ for $j\in\range{N}$ are defined as \eq{A_block}.
Note that $A$ is an $(Ns)$-sparse matrix, where $s$ is the largest nonzero number of each column and row in $F_1$ and $F_M$. The dimension of \eq{LODE} is denoted as
\begin{equation}
  \mathcal{N}_{d,N} \coloneqq n_d+n_d^2+\cdots+n_d^N=\frac{n_d^{N+1}-n_d}{n_d-1}=\Or(n^{Nd}).
\end{equation}

Denote the solution to the truncated system as $\hat y_j$ ($j = 0, 1, \cdots$) and define the error resulting from the truncation as 
\begin{equation}
\eta_j(t) \coloneqq \widetilde y_j(t) - \hat{y}_j(t) = U^{\otimes j}(t) - \hat{y}_j(t).
\end{equation} 
In particular, $\eta_1(t) = \widetilde y_1(t) - \hat{y}_1(t)  = U(t) - \hat{y}_1(t)$ is the error due to the Carleman linearization procedure.

\begin{theorem}[$\ell_2$ Convergence of the Carleman Linearization] \label{thm:Carleman_0}
For the discrete reaction-diffusion equation \eq{RDODE} with mixed boundary conditions in \eq{boundary}, as originally proposed in \cite{LKK20}, we define
\begin{equation}
\begin{aligned}
R &=  \frac{\|F_M\|}{|\lambda_1|}\|U_{\mathrm{in}}\|^{M-1},\\
\overline{R} &= \frac{\|F_M\|}{|\lambda_1|}\max_t\|U(t)\|^{M-1}.
\end{aligned}
\label{eq:A}
\end{equation}
Suppose that the largest eigenvalue of $D\Delta_h+a I$, denoted by $\lambda_1$, is negative. Then the approximation error of the Carleman linearization satisfies
\begin{equation}
\norm{\eta_j(t)} \le \max_t\|U(t)\|^{j}\overline{R}^{\lceil \frac{N-j+1}{M-1} \rceil} \left(1 - e^{j\lambda_1 t} \right) 
\end{equation}
for $t\geq 0$. In particular, if $N$ is some integer multiple of $M-1$, then the error of the solution $\eta_1(t) = y_1(t) - \hat{y}_1(t)  = U(t) - \hat U(t)$ satisfies
\begin{equation} \label{eq:carleman_0_eta1}
\norm{\eta_1(t)} \le \max_t\|U(t)\| \overline{R}^{ \frac{N}{M-1} } \left(1 - e^{\lambda_1 t} \right) 
\end{equation}
for $t\geq 0$. Furthermore, if $R\leq1$, then $R = \overline{R}$.
\end{theorem}

The detailed proof is presented in \sec{proof_Carleman_0}. This theorem for polynomial differential equations is a straightforward extension of the quadratic case in \cite[Corollary 1]{LKK20} and implies exponential convergence in the order of truncation $N$ in terms of $\ell_2$ norm as long as $R \leq 1$, that is
\begin{equation}
    \frac{\|F_M\|}{|\lambda_1|}\|U_{\mathrm{in}}\|^{M-1} = \frac{|b|}{|\lambda_1|}\|U_{\mathrm{in}}\|^{M-1} < 1,
\end{equation}
as shown in \eq{A}. However, if $R>1$, we only have $R \le \overline{R}$, and $\overline{R}$ is the exact convergence radius. 

\begin{theorem}[$\ell_{\infty}$ Convergence of the Carleman Linearization] \label{thm:Carleman}
For the discrete reaction-diffusion equation \eq{RDODE} with mixed boundary conditions as proposed in \eq{boundary}, we define
\begin{equation}
\begin{aligned}
R_{D} \coloneqq \frac{|b|}{|\lambda_1|}\gamma^{M-1}C(\lambda).
\label{eq:A1}
\end{aligned}
\end{equation}
Here $\gamma = (\frac{|a|}{|b|})^{\frac{1}{M-1}}$, and $\gamma$ upper bounds $\|U(t)\|_{\infty}$ for all $t\ge 0$, as stated in \lem{apriori}; the constant $C(\lambda)$ has the form
\begin{equation}\label{eq:E}
\begin{aligned}
C(\lambda) \coloneqq
    \begin{cases}
        \frac{|\lambda_1|}{a} (e^{\frac{\ln(3) d_1}{2(\lambda-\lambda_1)}a }-1) + \frac{|\lambda_1|}{|\lambda|}, & a \neq 0, \\
        \frac{\ln(3)d_1}{2(\lambda-\lambda_1)} |\lambda_1| + \frac{|\lambda_1|}{|\lambda|}, &  a = 0,
    \end{cases}
\end{aligned}
\end{equation}
where $\mu_1=-4(n+1)^2\sin^2\left(\frac{\pi}{2n+2}\right)$, $\lambda_1 = D d_1 \mu_1 + a < 0$, and $\lambda$ is an arbitrary value satisfying $\lambda_1<\lambda < 0$.
Then the approximation error of the Carleman linearization satisfies
\begin{equation}
\norm{\eta_j(t)}_{\infty} \le \gamma^{j}R_{D}^{\lceil \frac{N-j+1}{M-1} \rceil}
\end{equation}
for $t\geq 0$, with $R_{D}$ defined as \eq{A1}. In particular, if $N$ is some integer multiple of $M-1$, then the error of the solution $\eta_1(t) = y_1(t) - \hat{y}_1(t)  = U(t) - \hat U(t)$ satisfies
\begin{equation} \label{eq:carleman_eta1}
\norm{\eta_1(t)}_{\infty} \le \gamma R_{D}^{ \frac{N}{M-1} }
\end{equation}
for $t\geq 0$. 
\end{theorem}

This theorem implies an alternative exponential convergence in the order of truncation $N$ in terms of $\ell_{\infty}$ norm as long as $R_{D} <1$, that is, 
\begin{equation}
    \frac{|b|}{|\lambda_1|}\gamma^{M-1}C(\lambda)< 1,
\end{equation}
as shown in \eq{A1}.

\begin{proof}
The truncation error $\eta_j$ satisfies the equation
\begin{equation}
\label{eqn:eta_A}
\eta'_j(t) = A^j_j \eta_j(t) + A^j_{j+M-1}\left( y_{j+M-1}(t)-\hat{y}_{j+M-1}(t)\delta_{j+M-1 \le N} \right), \qquad 1 \le j \le N.
\end{equation}
Applying the variation of constants formula \cite{BrauerNohel2012} to \eqref{eqn:eta_A}, one has
\begin{align}
\eta_j(t) = \int_0^t e^{A^j_j (t-s)} A^j_{j+M-1} y_{j+M-1}(s) \,d s, \quad  N-M+2 \le j \le N.
\end{align}
Note that it follows from \lem{apriori} that
\[
\norm{y_{j+M-1}(s)}_\infty = \norm{U^{\otimes(j+M-1)}(s)}_{\infty}  
\le \norm{U(s)}^{j+M-1}_\infty
\le \gamma^{j+M-1}.
\]
Therefore, we have for $N-M+2 \le j \le N$,
\begin{equation}
\begin{aligned}
    \norm{\eta_j(t)}_\infty 
    & \le  \int_0^t \norm{e^{A^j_j (t-s)}}_\infty \norm{A^j_{j+M-1}}_\infty \norm{y_{j+M-1}(s)}_\infty \,d s  \\
    & \le  j|b| \gamma^{j+M-1} \int_0^t \norm{e^{A^j_j (t-s)}}_{\infty}  \,d s  \\
    & \le  j|b| \gamma^{j+M-1} \int_0^t \|e^{j(t-s) (D \Delta_h + a)}\|_{\infty}  \,d s.
\end{aligned}
\end{equation}
For simplicity, we denote
\begin{equation}
    C_j(t) \coloneqq j|\lambda_1|\int_0^t \|e^{j(t-s) (D \Delta_h + a)}\|_{\infty}  \,d s
\end{equation}
such that for $N-M+2 \le j \le N$,
\begin{equation}
\begin{aligned}
    \norm{\eta_j(t)}_\infty  \le  \frac{|b|}{|\lambda_1|}\gamma^{j+M-1}C_j(t).
\end{aligned}
\end{equation}

Next, for $N-2M+3 \le j \le N-M+1$,
\begin{equation}
\begin{aligned}
    \norm{\eta_j(t)}_\infty 
    & \le  \int_0^t \norm{e^{A^j_j (t-s)}}_\infty \norm{A^j_{j+M-1}}_\infty \norm{\eta_{j+M-1}(s)}_\infty \,d s  \\
    & \le  \frac{|b|}{|\lambda_1|}\gamma^{j+2M-2}C_{j+M-1}(t)\int_0^t \norm{e^{A^j_j (t-s)}}_\infty \norm{A^j_{j+M-1}}_\infty  \,d s  \\
    & \le \left( \frac{|b|}{|\lambda_1|}\right)^2 \gamma^{j+2M-2}C_{j+M-1}(t)C_j(t).
\end{aligned}
\end{equation}
One can continue by mathematical induction for every group of $M-1$ terms and obtain
\begin{equation}
\begin{aligned}
    \norm{\eta_j(t)}_\infty 
  & \le  \int_0^t \norm{e^{A^j_j (t-s)}}_\infty \norm{A^j_{j+M-1}}_\infty \norm{\eta_{j+M-1}(s)}_\infty \,d s  \\
  & \le \left(\frac{|b|}{|\lambda_1|}\right)^{\lceil \frac{N-j+1}{M-1} \rceil}\gamma^{j+(M-1)\lceil \frac{N-j+1}{M-1} \rceil} \prod_{k=1}^{\lceil \frac{N-j+1}{M-1} \rceil}C_{j+(M-1)(k-1)}(t). 
  \label{eq:periodic-error}
\end{aligned}
\end{equation}

We now consider an upper bound on $C_j(t)$. By computing the integration of $\|e^{j(t-s) (D \Delta_h + a)}\|_{\infty}$ in \lem{exp_decay_RD}, we have
\begin{equation}
\begin{aligned}
    \int_0^t \|e^{j(t-s) (D \Delta_h + a)}\|_{\infty}  \,d s \le \begin{cases}
        \frac{1}{ja} \left(e^{\frac{\ln(3)d_1}{2(\lambda-\lambda_1)}a}-1\right) + \frac{1}{j|\lambda|}, & a \neq 0, \\
        \frac{\ln(3)d_1}{2j(\lambda-\lambda_1)} + \frac{1}{j|\lambda|}, &  a = 0.
        \end{cases}
\end{aligned}
\end{equation}
Therefore, 
\begin{equation}
\begin{aligned}
    C_j(t) 
    & = j|\lambda_1|\int_0^t \|e^{j(t-s) (D \Delta_h + a)}\|_{\infty}  \,d s \le
    \begin{cases}
        \frac{|\lambda_1|}{a} \left(e^{\frac{\ln(3)d_1}{2(\lambda-\lambda_1)}a}-1\right) + \frac{|\lambda_1|}{|\lambda|}, & a \neq 0, \\
        \frac{\ln(3)d_1}{2(\lambda-\lambda_1)}|\lambda_1| + \frac{|\lambda_1|}{|\lambda|}, &  a = 0,
    \end{cases}
\end{aligned}
\end{equation}
where the right-hand side is exactly $C(\lambda)$ in \eq{E} and is independent of both $j$ and $t$.

Finally, substituting $C_j(t)$ with its upper bound $C(\lambda)$ in \eq{periodic-error} gives
\begin{equation}
\begin{aligned}
    \norm{\eta_j(t)}_\infty \le \left(\frac{|b|}{|\lambda_1|}\right)^{\lceil \frac{N-j+1}{M-1} \rceil}\gamma^{j+(M-1)\lceil \frac{N-j+1}{M-1} \rceil} C(\lambda)^{\lceil \frac{N-j+1}{M-1} \rceil } \le \gamma^{j}R_{D}^{\lceil \frac{N-j+1}{M-1} \rceil} ,
\end{aligned}
\end{equation}
where we use $R_D = \frac{|b|}{|\lambda_1|}\gamma^{M-1}C(\lambda)$ as defined in \eq{A1}. Therefore, a sufficient condition for the convergence of $\norm{\eta_j(t)}_\infty$ in $N$ is $R_D<1$. 

In practice, we can set $N$ as some integer multiple of $M-1$, and thus
\begin{equation}
\norm{\eta_1}_{\infty} \le \gamma R_{D}^{ \frac{N}{M-1} }.
\end{equation}
This completes the proof of the desired result.
\end{proof}

\textbf{Remark.}  
According to \thm{Carleman_0} and \thm{Carleman}, the Carleman linearized solution $\hat y_1(t)$ approximates the exact solution $U(kh)$ of the original reaction-diffusion equations with exponential convergence rate in terms of the convergence radius $R$ or $R_D$. 

The quantities $R$ and $R_{D}$ are used to characterize the ratios of \textit{reaction} and \textit{diffusion} strengths in terms of $\ell_{\infty}$ and $\ell_2$ norms. Here we briefly discuss the relationship between these. In particular, we are interested in the case where the convergence radius $R_D$ has an advantage over $R$. Note that $R_D\leq R $ is equivalent to $ C(\lambda) \leq \left(\frac{\|U_{\mathrm{in}}\|}{\gamma}\right)^{M-1}$ and that $\lambda$ can be an arbitrary number between $\lambda_1$ and 0. Hence the regime of our interest turns out to be
\begin{equation}\label{eq:optRD}
    \min_{\lambda_1<\lambda<0}C(\lambda) \leq \left(\frac{\|U_{\mathrm{in}}\|}{\gamma}\right)^{M-1}, 
\end{equation}
which holds true for a large regime of parameters in the high-dimensional or finely discretized scenarios because the right-hand side is likely to grow rapidly in $n$ and $d$ while the left-hand side only has a weak dependence. 

Specifically, according to \sec{preconstant}, solving the optimization problem on the left-hand side of \eq{optRD} helps to obtain a sharper estimate of $R_D$. When $a=0$, the minimum of $C(\lambda)$ has an explicit expression. As for $a\ne 0$,  we advise tuning $\lambda$ for a sharper estimate of $R_D$ in real applications, since the optimization problem is hard to solve explicitly. 
Nevertheless, in both cases, we can show that there exists an upper bound  of  $\min_{\lambda_1<\lambda<0} C(\lambda)$ which is also independent of $n$. For sufficiently large $n$, the quantity $\min_{\lambda_1<\lambda<0} C(\lambda)$ is $\Or(d)$, where $d$ is the spatial dimension.

For $n_d$-dimensional vectors, $\|U_{\mathrm{in}}\|$ can be significantly larger than $\gamma$ due to the inequality $\|U_{\mathrm{in}}\|_{\infty} \le \|U_{\mathrm{in}}\| \le \sqrt{n^d}\|U_{\mathrm{in}}\|_{\infty}$. First, when $d$ is fixed, $\|U_{\mathrm{in}}\|$ has a polynomial growth with $n$ in the worst case and then $R_D<R$ for large $n$. Second, when $n$ is large enough and fixed, $\|U_{\mathrm{in}}\|$ increases exponentially with $d$ in the worst case while $\min_{\lambda_1<\lambda<0}C(\lambda)$ grows at most linearly with $d$. Therefore, $R_D$ is smaller than $R$ for large $d$ as well.

For the case of grid refinement, the above results also show that $R_D$ stays bounded in the continuum limit $n_d\to\infty$, while $R$ diverges. 

\subsection{Numerical results}
In this subsection, we present some numerical results to examine the effectiveness of Carleman linearization. 
\begin{figure}[htbp]
    \centering
    \includegraphics[scale=0.4]{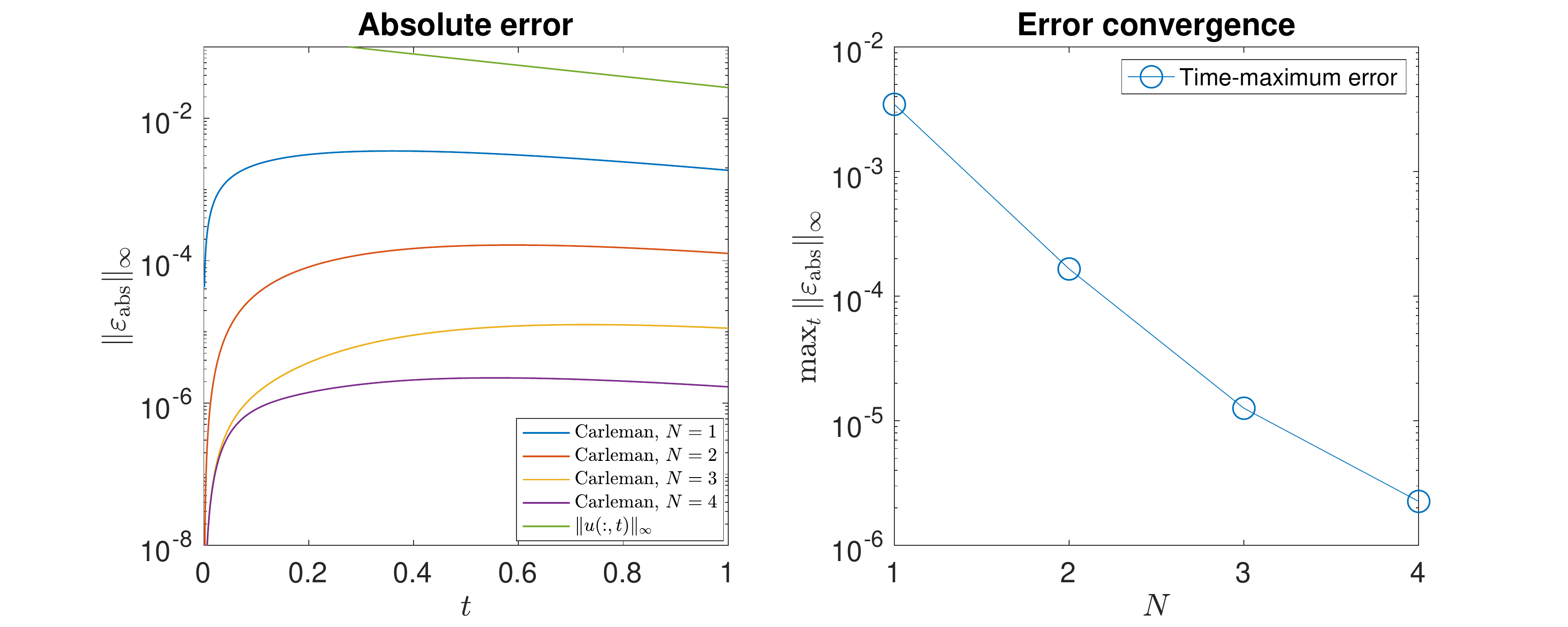}
    \caption{Applying Carleman linearization to \eq{NPDE} on a classical computer,  where $D=0.2$, $f(u)=0.2u-u^2$, the initial condition is $u(x,0)=0.1\left(1-\cos(2\pi x)\right)$ and homogeneous Dirichlet boundary conditions are imposed. Left: $l_{\infty}$ norm of the absolute error between the Carleman solutions at various truncation levels $N$. Right: the convergence of the corresponding time-maximum error.}
    \label{fgr:converge_FKPP}
\end{figure}

\begin{figure}[htbp]
    \centering
    \includegraphics[scale=0.4]{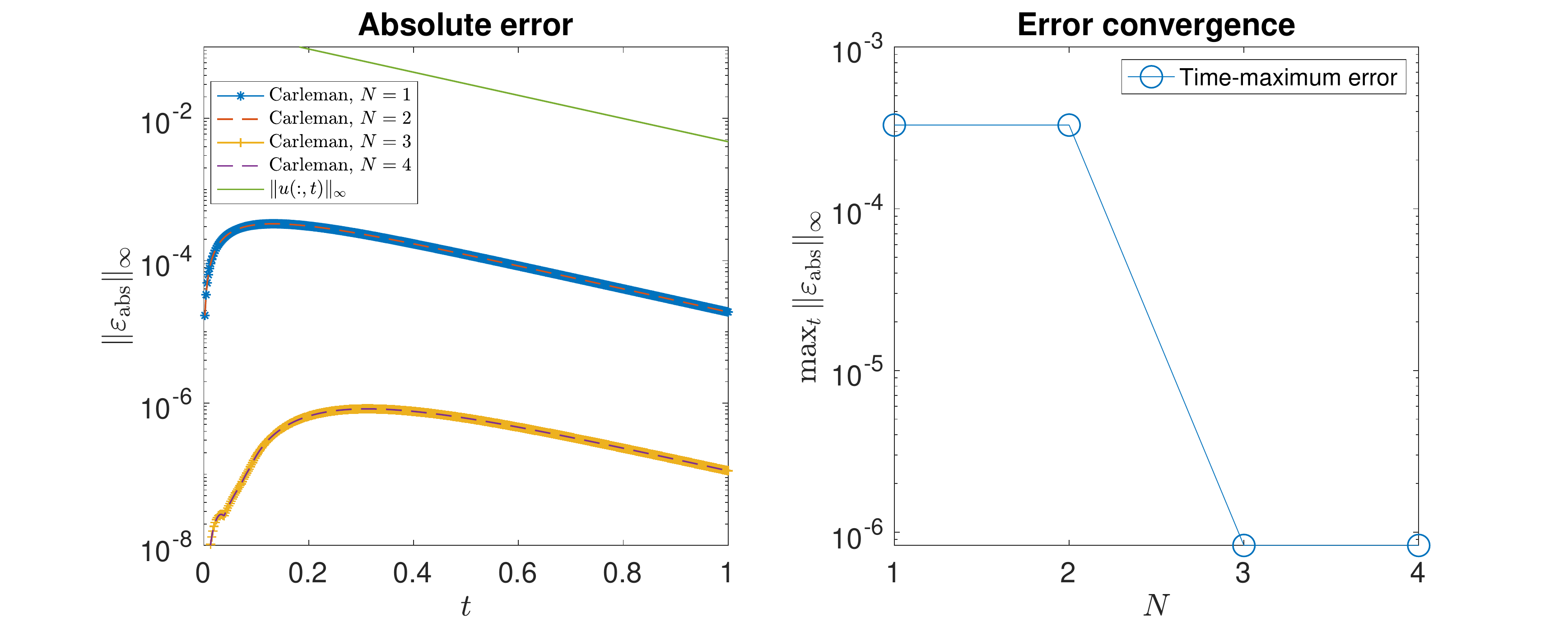}
    \caption{Applying Carleman linearization to \eq{NPDE} on a classical computer,  where $D=0.1$, $f(u)=0.16u-u^3$, the initial condition is $u(x,0)=0.2\sin(2\pi x)$ and homogeneous Dirichlet boundary conditions are imposed. Left: $l_{\infty}$ norm of the absolute error between the Carleman solutions at various truncation levels $N$. Right: the convergence of the corresponding time-maximum error.}
    \label{fgr:converge_AC}
\end{figure}

In order to demonstrate the convergence of Carleman linearization (Theorem \ref{thm:Carleman}), we apply our algorithm to \eq{NPDE} with different types of nonlinearity $f(u)$. In the first example, the nonlinear term is $f(u)=0.2u-u^2$, the Fisher-KPP type. We assume $D=0.2$, choose $u(x,0)=0.1\left(1-\cos(2\pi x)\right)$ as the initial condition, and impose homogeneous Dirichlet boundary conditions. In our second example, $f(u)=0.16u-u^3$, $D=0.1$, $u(x,0)=0.2\sin(2\pi x)$, and homogeneous Dirichlet boundary conditions are again used. The numerical results for these two examples are depicted in \fgr{converge_FKPP} and \fgr{converge_AC}, respectively. We see from the error convergence plots that the absolute error, maximized over $t \in [0,1]$, decreases exponentially as the truncation level $N$ is increased. As a function of the time $t$, the absolute error first increases and then decreases exponentially due to the decay of the exact solution. In particular, in \fgr{converge_AC}, the absolute error curve depicting the absolute error for $N=1$ agrees with that for $N=2$. That is because, according to \eq{UODE}, only $\hat y_1$ takes part in the time evolution of $\hat y_1$, no matter whether $N=1$ or $2$. A similar argument holds for the agreement of two curves for $N=3$ and $N=4$.

\begin{figure}[htbp]
    \centering
    \begin{subfigure}{.4\textwidth}
    \centering
    \captionsetup{justification=centering}
    \includegraphics[scale=0.4]{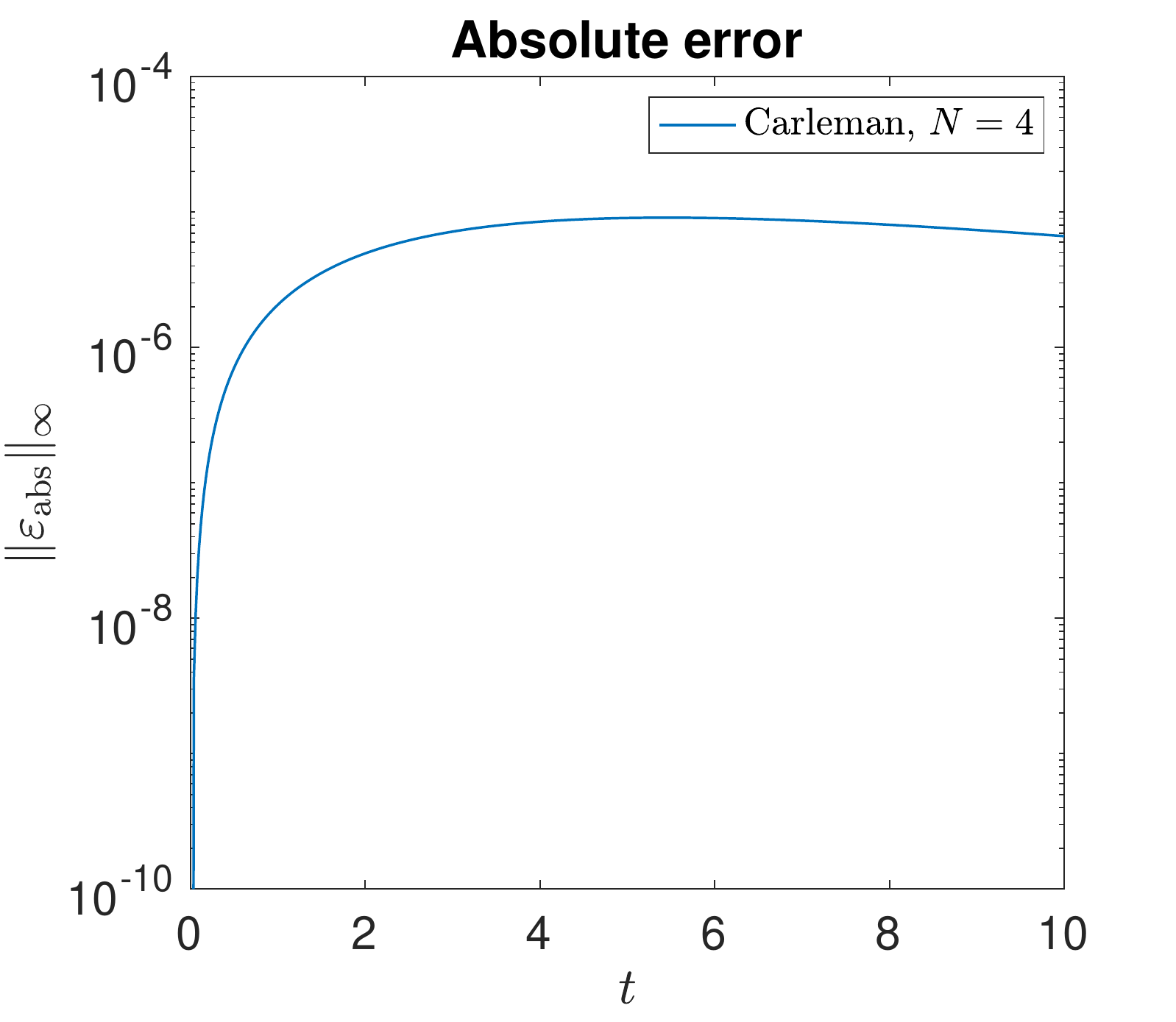}
    \caption{$D=0.01$, $f(u)=0.25u-u^3$ and $u(x,0)=0.08\sin(2\pi x)$.}
    \label{fgr:eig_pst}
    \end{subfigure}
    \begin{subfigure}{.4\textwidth}
    \centering
    \captionsetup{justification=centering}
    \includegraphics[scale=0.4]{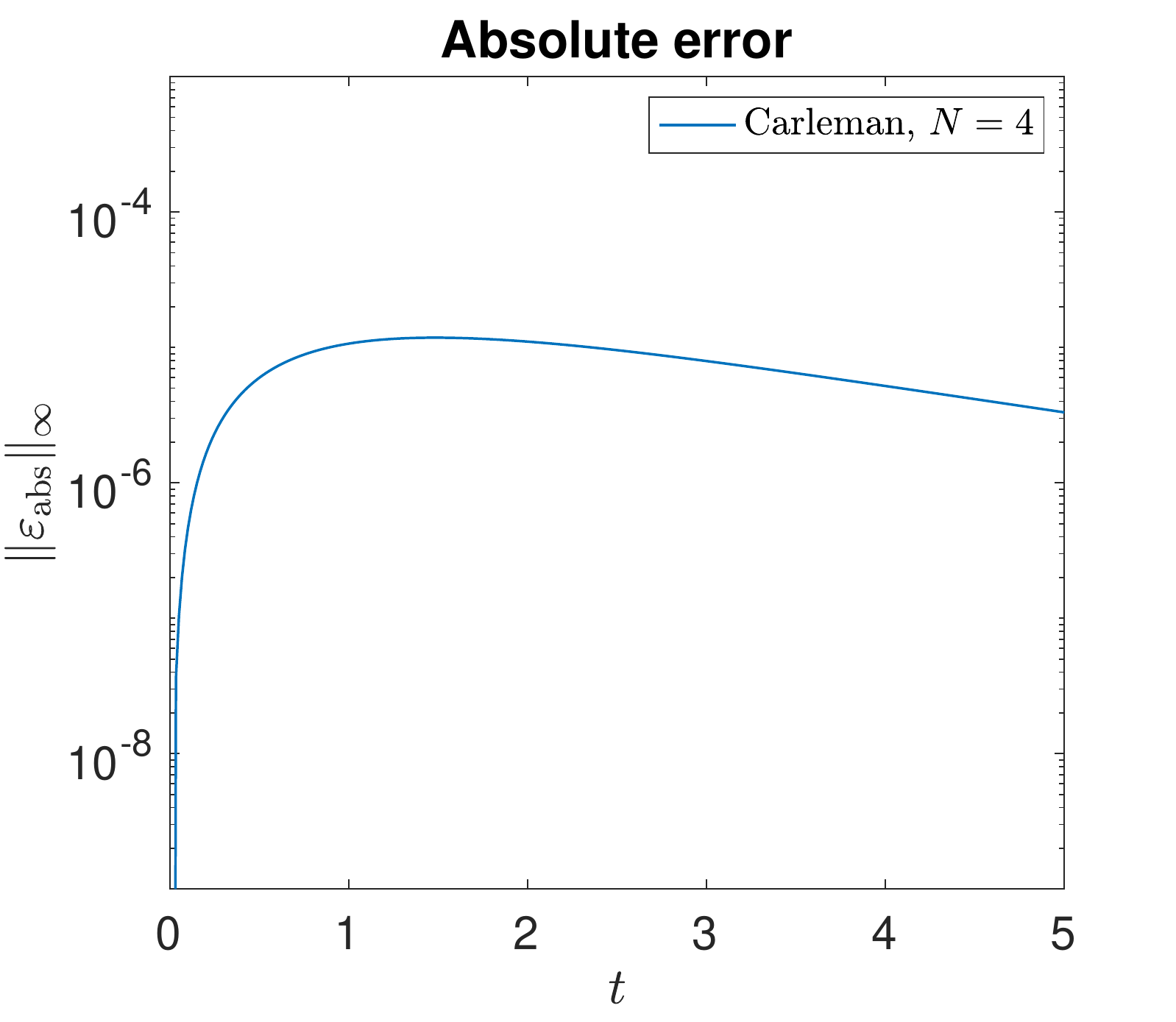}
    \caption{$D=0.012$, $f(u)=0.0196u-u^3$ and $ u(x,0)=0.14\sin(2\pi x)$ .}
    \label{fgr:rd_large}
    \end{subfigure}
    \caption{Examples of the application of Carleman linearization to \eq{NPDE} with homogeneous Dirichlet boundary conditions for different parameters and initial conditions. Left: The eigenvalues of $F_1$ are no longer all positive. Right: $R>1$ while $R_D< 1$.
    }
\end{figure}

Based on our numerical tests, we also find that Carleman linearization works for more general cases. In \fgr{eig_pst}, we relax the requirement that the eigenvalues of $F_1$ are all negative. We test \eq{NPDE} with homogeneous boundary conditions and choose $D=0.01$, $f(u)=0.25u-u^3$, and $u(x,0)=0.08\sin(2\pi x)$ as initial condition. Carleman linearization still has good numerical performance, as implied by the absolute error plot. \fgr{rd_large} illustrates the advantage of $R_D$ over $R$. In that example, we consider \eq{NPDE} with homogeneous boundary conditions and assume $D=0.012$, $f(u)=0.0196u-u^3$ and $ u(x,0)=0.14\sin(2\pi x)$.  We discretize the spatial domain into $15$ sub-intervals, i.e., the value of $n$ is 16. By computation, $R=1.4924$ and $R_D= 0.9299$ where we choose $\lambda =\lambda_1/2.3$. This illustrates that error remains well-controlled under $R_D<1$, which is a milder condition than $R<1$. 

\section{Quantum algorithm}
\label{sec:algorithm}

We now describe an efficient quantum algorithm for computing the numerical solution of the linearized ODEs \eq{LODE}. Our main algorithmic result is stated as follows. 

\begin{theorem}[Quantum Carleman linearization]\label{thm:main}
Consider an instance of the quantum 
ODE problem as defined in \prb{node}, with its $N$-th Carleman linearization as defined in \eq{LODE}. We denote a parameter
\begin{equation}
G \coloneqq \frac{1}{\sqrt{m+1}}\|\hat Y_{\mathrm{evo}}\| = \sqrt{\frac{\sum_{k=0}^{m}\|\hat y_1(kh)\|^2}{m+1}},
\label{eq:g}
\end{equation}
which parameterized the average $\ell_2$ norm of the evolution of the $N$-th Carleman linearized solution $\hat Y_{\mathrm{evo}}$. Assume $R_D<1$. Then there exists a quantum algorithm that produces a state that approximates $\hat Y_{\mathrm{evo}}$ in terms of the $\ell_2$ normalized error $\epsilon \le 1/2$ succeeding with probability $\Omega(1)$, with a flag indicating success. 
The query complexity (to the oracles $O_{F_1},O_{F_M}$, and $O_x$) is
\begin{equation}
\begin{aligned}
\frac{1}{G^2\epsilon}sT^2D^2d^2n^4N^3\|U_{\mathrm{in}}\|^{2N}\cdot\poly\biggl(\log\frac{aDdMnNsT}{G\epsilon}\biggr).
\end{aligned}
\end{equation}
The gate complexity is larger than its query complexity by logarithmic factors.
\end{theorem}

We next describe the quantum algorithm in detail, including ingredients such as state preparation, quantum linear system algorithm, and measurement. We conclude the proof of \thm{main} at the end of this section. 

\textbf{Remark.}  
We notice that a prefactor $\|U_{\mathrm{in}}\|^{2N}$ in the complexity of the $N$-th Carleman linearization. This cost is similar to prefactor $5^{2k}$ in the complexity of the $k$-th product formula \cite{BAC07}. In practice, we usually choose a suitably small value of $N$, such as $1,2,3$, to reduce the cost.

\subsection{State preparation}

We first recall a lemma used in \cite{LKK20} for preparing a quantum state corresponding to the initial vector $\hat y_{\mathrm{in}}=[U_{\mathrm{in}}; U_{\mathrm{in}}^{\otimes 2}; \ldots; U_{\mathrm{in}}^{\otimes N}]$, given the value $\|U_{\mathrm{in}}\|$ and the ability to prepare a quantum state proportional to $U_{\mathrm{in}}$.

\begin{lemma}[Lemma 5 of \cite{LKK20}]\label{lem:preparation}
Assume we are given the value $\|U_{\mathrm{in}}\|$, and let $O_x$ be an oracle that maps $|00\ldots0\rangle \in\C^n$ to a normalized quantum state $|U_{\mathrm{in}}\rangle$ proportional to $U_{\mathrm{in}}$. Then the quantum state $|\hat y_{\mathrm{in}}\rangle$ proportional to $\hat y_{\mathrm{in}}$ can be prepared using $\Or(N)$ queries to $O_x$, and the gate complexity is larger by an $\Or(\text{poly}(\log N, \log n))$ factor.
\end{lemma}

We remark that, in fact, we embed $\hat y_{\mathrm{in}}$ into a slightly larger space with a more convenient tensor product structure. Further details refer to Section 4.3 of \cite{LKK20}.

\subsection{Quantum linear system algorithm}
\label{sec:qlsa}

After the state preparation of the initial condition, we perform the forward Euler method to discretize the time interval $[0,T]$ into $m = T/h$ sub-intervals, and construct a system of linear equations as
\begin{equation}
\frac{y^{k+1}-y^k}{h} = A y^k
\label{eq:forward}
\end{equation}
where $y^k \in \R^{\mathcal{N}_{d,N}}$ approximates $\hat y(kh)$ for each $k \in \rangez{m+1} \coloneqq \{0,1,\ldots,m\}$, with $y^0 = y_{\mathrm{in}} \coloneqq \hat y(0) = \hat y_{\mathrm{in}}$.
This gives an $(m+1)\mathcal{N}_{d,N}\times(m+1)\mathcal{N}_{d,N}$ linear system
\begin{equation}
L|Y\rangle=|B\rangle, 
\label{eq:linear_system}
\end{equation}
where
\begin{equation}
L = \sum_{k=0}^{m+1}|k\rangle\langle k|\otimes I-\sum_{k=1}^{m+1}|k\rangle\langle k-1|\otimes (I+Ah).
\label{eq:matrixL}
\end{equation}
\eq{linear_system} encodes \eq{forward} and uses it to produce a numerical solution at time $T$. 
Observe that the system \eq{linear_system} has the lower triangular structure
\renewcommand{\arraystretch}{1.4}
\begin{equation}
  \begin{pmatrix}
    I &  &  &  &  \\
    -(I\!+\!Ah) & I &  &  &  \\
     & \ddots & \ddots &  &  \\
     &  & -(I\!+\!Ah) & I &  \\
     &  &  & -(I\!+\!Ah) & I \\
  \end{pmatrix}
  \begin{pmatrix}
    y^0 \\
    y^1 \\
    \vdots \\
   y^{m-1} \\
    y^m \\
  \end{pmatrix}
=
  \begin{pmatrix}
    y_{\mathrm{in}} \\
    0 \\
    \vdots \\
    0 \\
    0 \\
  \end{pmatrix}.
\end{equation}
\renewcommand{\arraystretch}{1}

For each $\mathcal{N}_{d,N}$-dim vector $y^k$ with $k\in\rangez{m+1}$, its first $n$ components (i.e., $y_1^k$) approximate the exact solution $U(kh)$, up to normalization. We apply the high-precision quantum linear system algorithm (QLSA) \cite{CKS15} to \eq{linear_system} and postselect to produce $\sum_k|y_1^k\rangle|k\rangle$ for representing the gradient flow evolution. 
We would like to note that a more advanced QLSA with block-encoding input models was recently proposed in~\cite{CostaAnSandersEtAl2022}. 
However, for technical simplicity, in this work we still employ the algorithm described in~\cite{CKS15}. 
This is because the improvements introduced in~\cite{CostaAnSandersEtAl2022} over~\cite{CKS15} are relatively minor, only affecting logarithmic factors. 
Additionally, the input model used in this work, which involves sparse input oracles, is more consistent with that employed in~\cite{CKS15}.

In \thm{main}, the solution error has two contributions: the error in the time discretization of \eq{LODE} by the forward Euler method, and the error from the QLSA. Since the QLSA produces a solution with error at most $\epsilon$ with complexity $\poly(\log(1/\epsilon))$ \cite{CKS15}, we focus on bounding the first contribution.

We provide an upper bound for the error incurred by approximating \eq{LODE} with the forward Euler method. The following proof basically follows \cite[Lemma 3]{LKK20}. 

\begin{lemma}\label{lem:global}
Consider an instance of the quantum ODE problem as defined in \prb{node}, with $R_D<1$ as defined in \eq{A1}. Choose a time step
\begin{equation}
h \le \frac{1}{N^2[4Dd(n+1)^2 + a]}.
\label{eq:h}
\end{equation}
Then the global error from the forward Euler method \eq{forward} on the interval $[0,T]$ for \eq{LODE} satisfies
\begin{equation}
\|\hat y_1(T)-y^m_1\|
\le \frac{N^2Th}{2}[4Dd(n+1)^2 + a + b]^2\max_{t\in[0,T]}\|\hat y(t)\|.
\end{equation}
\end{lemma}

\begin{proof}

First of all, we establish the following bound
\begin{equation}
\|I+Ah\| \le 1.
\end{equation}
We decompose $I+Ah$ as
\begin{equation}
I+Ah=\sum_{j=1}^N K_j
\end{equation}
where
\begin{align}
K_j &= \frac{1}{N}I+|j\rangle\langle j|\otimes A_j^jh+|j\rangle\langle j+M-1|\otimes A_{j+M-1}^jh, \quad j\in\range{N-M+1}, \\
K_{j} &= \frac{1}{N}I+|j\rangle\langle j|\otimes A_j^jh, \quad j\in\range{N-1} \backslash \range{N-M+1}.
\end{align}
All eigenvalues of $\frac{1}{N}I+|j\rangle\langle j|\otimes A_j^jh$ range from $\frac{1}{N}+j\lambda_nh$ to $\frac{1}{N}+j\lambda_1h$. Here we require that these eigenvalues lie in $[0,1]$, given by $N^2h\|F_1\| = N^2h[4Dd(n+1)^2 + a] \le 1$ in \eq{h}.  The norm of $A_{j+M-1}^j$ is bounded by $j\|F_M\|$. So we have the bound
\begin{equation}
\|K_j\| \le \frac{1}{N} - j|\lambda_1|h + j\|F_M\|h, \quad j\in\range{N-M+1},
\label{eq:norm_K}
\end{equation}
Then $\|F_M\|\le|\lambda_1|$ in \prb{node} gives
\begin{equation}
\|K_j\| \le \frac{1}{N}, \quad j\in\range{N-M+1}.
\end{equation}
It also holds for the case $j\in\range{N-1} \backslash \range{N-M+1}$. Henceforth,
\begin{equation}
\|I+Ah\| \le \sum_{j=1}^N \|K_j\| \le 1.
\label{eq:norm_equivalent}
\end{equation}

We then define the global error by
\begin{equation}
g^k \coloneqq \|\hat y(kh)-y^k\|,
\end{equation}
where $\hat y(kh)$ is the exact solution of \eq{LODE}, and $y^k$ is the numerical solution of \eq{forward}. Note that $g^m = \|\hat y(T)-y^m\|$.

The stable condition \eq{norm_equivalent} implies the local truncation error from the forward Euler method is non-increasing, and the global error increase at most linear in time. Following the standard procedure of the global error estimate (we refer it to the proof of \cite[Lemma 3]{LKK20}), the global error is bounded by
\begin{equation}
g^k \le \frac{kh^2}{2}\max_{t\in[0,T]}\|\hat y''(t)\|, \quad k\in\rangez{m+1}, 
\label{eq:global_recurrence}
\end{equation}
where we have the following estimate
\begin{equation}
\max_{t\in[0,T]}\|\hat y''(t)\| = \max_{t\in[0,T]}\|A^2\hat y(t)\| \le \|A\|^2 \max_{t\in[0,T]}\|\hat y(t)\|,
\end{equation}
\begin{equation}
\|A\| = \biggl\|\sum_{j=1}^{N}|j\rangle\langle j|\otimes A_{j}^j+\sum_{j=1}^{N-M+1}|j\rangle\langle j+M-1|\otimes A_{j+M-1}^j\biggr\| \le N(\|F_1\|+\|F_M\|).
\end{equation}
Sequentially, we conclude that
\begin{equation}
\begin{aligned}
\|\hat y^1(T)-y^m_1\| &\le g^m \le \frac{N^2Th}{2}(\|F_1\|+\|F_M\|)^2\max_{t\in[0,T]}\|\hat y(t)\| \\ 
&\le \frac{N^2Th}{2}[4Dd(n+1)^2 + a + b]^2\max_{t\in[0,T]}\|\hat y(t)\|.
\end{aligned}
\label{eq:global}
\end{equation}
\end{proof}

Given the above linear system, we can upper bound the condition number that affects the complexity of the quantum linear system algorithm. Under the same construction of the matrix $L$, we can follow the same estimate proposed by \cite[Lemma 4]{LKK20} to claim the following result. 

\begin{lemma}\label{lem:condition}
Consider an instance of the quantum ODE problem as defined in \prb{node}. Apply the forward Euler method \eq{forward} with time step size \eq{h} to the Carleman linearization \eq{LODE}. Then the condition number of the matrix $L$ defined in \eq{matrixL} satisfies
\begin{equation}
\kappa \le 2(m+1).
\end{equation}
\end{lemma}

\begin{proof}
We begin by upper bounding $\|L\|$. We write
\begin{equation}
L = L_1 + L_2,
\end{equation}
where
\begin{align}
L_1 &= \sum_{k=0}^{m}|k\rangle\langle k|\otimes I, \\
L_2 &= -\sum_{k=1}^{m}|k\rangle\langle k-1|\otimes (I+Ah).
\end{align}
Clearly $\|L_1\| = 1$. Furthermore, $\|L_2\| \le \|I+Ah\| \le 1$ by \eq{norm_equivalent}. Therefore,
\begin{equation}
\|L\| \le \|L_1\| + \|L_2\| \le 2.
\label{eq:upper}
\end{equation}
Next we upper bound $\|L^{-1}\|$. We notice that $L^{-1}$ can be directly written as
\renewcommand{\arraystretch}{1.4}
\begin{equation}
L^{-1} =
    \begin{pmatrix}
    I &  &  &  & \\
    (I\!+\!Ah) & I &  &  & \\
    (I\!+\!Ah)^2 & (I\!+\!Ah) & I &  & \\
    \ldots & \ddots & \ddots & \ddots & \\
    (I\!+\!Ah)^m & \cdots & (I\!+\!Ah)^2 & (I\!+\!Ah) & I \\
    \end{pmatrix}.
\end{equation}
\renewcommand{\arraystretch}{1}
So that
\begin{equation}
    \|L^{-1}\| \le \|I\| + \|I+Ah\| + \ldots + \|(I+Ah)^m\|.
\end{equation}
Since $\|I+Ah\| \le 1$ by \eq{norm_equivalent}, we have
\begin{equation}
    \|L^{-1}\| \le \|I\| + \|I+Ah\| + \ldots + \|(I+Ah)\|^m = m+1.
\label{eq:lower}
\end{equation}
Finally, combining \eq{upper} with \eq{lower}, we conclude
\begin{equation}
\kappa = \| L\|\| L^{-1}\| \le 2(m+1)
\end{equation}
as claimed.
\end{proof}

\subsection{Measurement probability}

After applying the QLSA to \eq{linear_system}, we perform a measurement to extract a final state of the desired form. We now consider the probability of this measurement succeeding. Differing from \cite[Lemma 6]{LKK20}, we are interested in providing a history state $\sum_{k}|y_1^k\rangle|k\rangle$ instead of a final state $|y_1^m\rangle$. Thus the measurement probability does not include the $\ell_2$ norm of the final state as well as the $\ell_2$ scaling of the initial and final states (i.e., the parameters $g$ and $q$ as in \cite[Lemma 6]{LKK20}). 

\begin{lemma}\label{lem:measure}
Consider an instance of the quantum ODE problem defined in \prb{node}, with the QLSA applied to the linear system \eq{linear_system} using the forward Euler method \eq{forward} with time step \eq{h}.
Suppose the global error from the forward Euler method as defined in \lem{global} is bounded by
\begin{equation}
\|\hat y(kh)-y^k\| \le \frac{1}{2}G.
\label{eq:g_condition_2}
\end{equation}
Then the probability of measuring a quantum state $|y^k_1\rangle$ for $k\in\rangez{m+1}$ satisfies
\begin{equation}
P_{\mathrm{measure}} \ge \frac{2G^2}{16\max_{t\in[0,T]}\|\hat y(t)\|^2 + G^2}.
\end{equation}
\end{lemma}

\begin{proof}
The quantum state produced by the QLSA applied to \eq{linear_system} has the form
\begin{equation}
|Y\rangle = \frac{1}{\||Y\rangle\|}\sum_{k=0}^{m}y^k|k\rangle = \frac{1}{\|Y\|}\sum_{k=0}^{m}\sum_{j=1}^Ny_j^k|j\rangle|k\rangle
\end{equation}
where the normalization factor satisfies $\||Y\rangle\|^2 = \sum_{k=0}^{m}\|y^k\|^2 = \sum_{k=0}^{m}\sum_{j=1}^N\|y_j^k\|^2$. 

We aim to obtain the target quantum state as the form
\begin{equation}
|Y_{\mathrm{evo}}\rangle = \frac{1}{\|Y_{\mathrm{evo}}\|}\sum_{k=0}^{m}y^k_1|k\rangle.
\label{eq:history_state}
\end{equation}
which corresponds to the gradient flow evolution state \eq{history}
We measure the register $|j\rangle$, $j\in\range{N}$ and extract $|Y_{\mathrm{target}}\rangle$ from $|Y\rangle$ when $k=1$. The success probability is lower bounded as below.

According to the Cauchy–Schwarz inequality, 
\begin{equation}
\|\hat y_1(kh)\|^2 = \|\hat y_1(kh)-y^k_1+y^k_1\|^2 \le 2\|\hat y_1(kh)-y^k_1\|^2 + 2\|y^k_1\|^2,
\end{equation}
so that
\begin{equation}
\||y^k_1\rangle\|^2
\ge \frac{1}{2}\|\hat y_1(kh)\|^2 - \|\hat y_1(kh)-y^k_1\|^2
\ge \frac{1}{2}\|\hat y_1(kh)\|^2 - \frac{1}{4}G^2.
\end{equation}
Summing $k$ from $0$ to $m$, and using the definition of $G$ in \eq{g}, we have
\begin{equation}
\begin{aligned}
\||Y_{\mathrm{evo}}\rangle\|^2 
&= \sum_{k=0}^{m}\|y_1^k\|^2 \ge \frac{1}{2}\sum_{k=0}^{m}\|\hat y_1(kh)\|^2 - \frac{m+1}{4}G^2 \\
&\ge \frac{m+1}{2}G^2 - \frac{m+1}{4}G^2 = \frac{m+1}{4}G^2.
\end{aligned}
\label{eq:measure_1}
\end{equation}

Second, we use \eq{g_condition_2} and the parallel inequality again to upper bound $\|y^k\|^2$ by
\begin{equation}
\begin{aligned}
\|y^k\|^2 
\le 2\|\hat y(kh)\|^2 + 2\|\hat y(kh)-y^k\|^2 \le 2\|\hat y(kh)\|^2 + \frac{1}{2}G^2.
\end{aligned}
\end{equation}
Therefore
\begin{equation}
\begin{aligned}
\||Y\rangle\|^2 = \sum_{k=0}^{m}\|y^k\|^2 \le 2\sum_{k=0}^{m}\|\hat y(kh)\|^2 + \frac{m+1}{2}G^2 \le 2(m+1)\max_{t\in[0,T]}\|\hat y(t)\|^2 + \frac{m+1}{2}G^2.
\label{eq:measure_2}
\end{aligned}
\end{equation}

Finally, using \eq{measure_1} and \eq{measure_2}, we have
\begin{equation}
P_{\mathrm{measure}} = \frac{\||Y_{\mathrm{evo}}\rangle\|^2}{\||Y\rangle\|^2} \ge 
\frac{G^2}{8\max_{t\in[0,T]}\|\hat y(t)\|^2 + 2G^2}
\end{equation}
as claimed.

\end{proof}

Using amplitude amplification, $\Or(\sqrt{1/P_{\mathrm{measure}}})$ iterations suffice to succeed with constant probability.

\subsection{Proof of \texorpdfstring{\thm{main}}{Theorem \ref{thm:main}}
}

\begin{proof}
We first present the quantum Carleman linearization (QCL) algorithm and then analyze its complexity.

\paragraph{The QCL algorithm.} 

We introduce the choice of parameters as follows. 
Given an error bound $\epsilon \le 1$ and $G$, we define
\begin{equation}
\delta \coloneqq 
\frac{G\epsilon}{1+\epsilon},
\label{eq:delta_1}
\end{equation}
which satisfies $\delta \le G/2$ for any $t\in[0,T]$. 

Now we discuss the choice of $h$. On the one hand, $h$ must follow \eq{h} to satisfy the conditions of \lem{global} and \lem{condition}.
On the other hand, we choose
\begin{equation}
\begin{aligned}
h \le \min \biggl\{ \frac{1}{N^2[4Dd(n+1)^2 + a]},\frac{G\epsilon}{N^2T[4Dd(n+1)^2 + a + b]^2\max_{t\in[0,T]}\|\hat y(t)\|}\biggr\}
\end{aligned}
\end{equation}
Then according to the requirement \eq{h} in \lem{global}, and for $k\in\rangez{m+1}$, 
\begin{equation}
\|\hat y_1(kh)-y^k_1\| \le \|\hat y(kh)-y^k\| \le \frac{N^2Th}{2}[4Dd(n+1)^2 + a + b]^2\max_{t\in[0,T]}\|\hat y(t)\| \le \delta.
\label{eq:error_2}
\end{equation}
It also leads to $\|\hat y(kh)-y^k\| \le G/2$ used as a condition in \lem{measure}.

We now consider the error between the exact and numerical gradient flow evolutions
\begin{equation}
|\hat Y_{\mathrm{evo}}\rangle = \frac{1}{\|\hat Y_{\mathrm{evo}}\|}\hat Y_{\mathrm{evo}} = \frac{1}{\|\hat Y_{\mathrm{evo}}\|}\sum_{k=0}^{m}\hat y_1(kh)|k\rangle
\end{equation}
and (as denoted in \eq{history_state})
\begin{equation}
|Y_{\mathrm{evo}}\rangle = \frac{1}{\|Y_{\mathrm{evo}}\|}Y_{\mathrm{evo}} = \frac{1}{\|Y_{\mathrm{evo}}\|}\sum_{k=0}^{m}y^k_1|k\rangle,
\end{equation}
where $\|\hat Y_{\mathrm{evo}}\|$ and $\|Y_{\mathrm{evo}}\|$ are normalization factors. Recall the definition of $G$ in \eq{g}
\begin{equation}
G = \frac{1}{\sqrt{m+1}}\|\hat Y_{\mathrm{evo}}\| = \sqrt{\frac{\sum_{k=0}^{m}\|\hat y_1(kh)\|^2}{m+1}},
\end{equation}
The $\ell_2$ normalized error can be controlled by
\begin{equation}
\biggl\| |\hat Y_{\mathrm{evo}}\rangle-|Y_{\mathrm{evo}}\rangle \biggr\| \le \frac{\|\hat Y_{\mathrm{evo}}-Y_{\mathrm{evo}}\|}{\min\{\|\hat Y_{\mathrm{evo}}\|,\|Y_{\mathrm{evo}}\|\}} \le \frac{\|\hat Y_{\mathrm{evo}}-Y_{\mathrm{evo}}\|}{\|\hat Y_{\mathrm{evo}}\|-\|\hat Y_{\mathrm{evo}}-Y_{\mathrm{evo}}\|}.
\label{eq:normalized_inequality}
\end{equation}
Then using \eq{error_2}, since
\begin{equation}
\|\hat Y_{\mathrm{evo}}-Y_{\mathrm{evo}}\|^2 \le \sum_{k=0}^{m} \|\hat y_1(kh)-y^k_1\|^2 \le (m+1)\delta^2,
\end{equation}
we have
\begin{equation}
\|\hat Y_{\mathrm{evo}}-Y_{\mathrm{evo}}\| \le \sqrt{m+1}\delta,
\end{equation}
which gives
\begin{equation}
\biggl\| |\hat Y_{\mathrm{evo}}\rangle-|Y_{\mathrm{evo}}\rangle \biggr\| \le \frac{\sqrt{m+1}\delta}{\|\hat Y_{\mathrm{evo}}\|-\sqrt{m+1}\delta} = \frac{\delta}{G-\delta} = \epsilon, 
\label{eq:error_4}
\end{equation}
i.e., $\epsilon$ upper bounds the $\ell_2$ normalized error between $|\hat Y_{\mathrm{evo}}\rangle$ and $|Y_{\mathrm{evo}}\rangle$.

We follow the procedure in \lem{preparation} to prepare the initial state $|\hat y_{\mathrm{in}}\rangle$.
We apply the QLSA \cite{CKS15} to the linear system \eq{linear_system} with $m=\lceil T/h \rceil$, giving a solution $|Y\rangle$. 
By \lem{measure}, the probability of obtaining a state is $\sum_{k=0}^{m}|y_1^k\rangle|k\rangle$ 
\begin{equation}
P_{\mathrm{measure}} \ge \frac{2G^2}{16\max_{t\in[0,T]}\|\hat y(t)\|^2 + G^2}.
\end{equation}
By amplitude amplification, we can achieve success probability $\Omega(1)$ with $\Or(\max_{t\in[0,T]}\|\hat y(t)\|/G)$ repetitions of the above procedure.

\paragraph{Analysis of the complexity.} By \lem{preparation}, the initial state $|\hat y_{\mathrm{in}}\rangle$ can be prepared with $\Or(N)$ queries to $O_x$, with gate complexity larger by a $\Or(\text{poly}(\log N, \log n))$ factor. The matrix $L$ is an $(m+1)\mathcal{N}_{d,N}\times(m+1)\mathcal{N}_{d,N}$ matrix with $\Or(Ns)$ nonzero entries in any row or column. By \lem{condition} and our choice of parameters, the condition number of $L$ is at most
\begin{equation}
\begin{aligned}
&2(m+1) \\
&\quad= O\biggl(\frac{1}{G\epsilon}N^2T^2[4Dd(n+1)^2 + a + b]^2\max_{t\in[0,T]}\|\hat y(t)\|+N^2T[4Dd(n+1)^2 + a] \biggr) \\
&\quad= O\Bigl(\frac{1}{G\epsilon}N^2T^2D^2d^2n^4\max_{t\in[0,T]}\|\hat y(t)\|\Bigr).
\label{eq:condition_estimate}
\end{aligned}
\end{equation}
Here we use $(\|F_1\|+\|F_M\|)^2 = [4Dd(n+1)^2 + a + b]^2$, and $|a|, |b| = o(d)$.
Consequently, by Theorem 5 of \cite{CKS15}, the QLSA produces the state $|Y\rangle$ with
\begin{equation}
\begin{aligned}
\frac{1}{G\epsilon}sT^2D^2d^2n^4N^3\max_{t\in[0,T]}\|\hat y(t)\|\cdot\poly\biggl(\log\frac{aDdnNsT}{G\epsilon}\biggr)
\end{aligned}
\end{equation}
queries to the oracles $O_{F_1}$ and $O_{F_2}$. Using $O(\max_{t\in[0,T]}\|\hat y(t)\|/G)$ steps of amplitude amplification to achieve success probability $\Omega(1)$, the overall query complexity of our algorithm is
\begin{equation}
\begin{aligned}
\frac{1}{G^2\epsilon}sT^2D^2d^2n^4N^3\max_{t\in[0,T]}\|\hat y(t)\|^2\cdot\poly\biggl(\log\frac{aDdnNsT}{G\epsilon}\biggr) 
\label{eq:query_complexity}
\end{aligned}
\end{equation}
and its gate complexity is larger than its query complexity only by logarithmic factors,
based on the gate-efficient algorithm in Theorem 5 of \cite{CKS15}.

We now estimate the quantity $\max_{t\in[0,T]}\|\hat y(t)\|$. By the definition of $\eta_j(t)$, \thm{Carleman}, and $R_D<1$, we have
\begin{equation}
\|\hat y_j(t)\| = \|U^{\otimes j}(t) - \eta_j(t) \| \le \|U^{\otimes j}(t)\| + \|\eta_j(t)\| \le \max_{t\in[0,T]}\|U(t)\|^j + \gamma^{j} \le 2\max_{t\in[0,T]}\|U(t)\|^j,
\end{equation}
so that
\begin{equation}
\max_{t\in[0,T]}\|\hat y(t)\| \le 2\sum_{j=1}^N(\max_{t\in[0,T]}\|U(t)\|)^j.
\end{equation}
Based on the $\ell_2$ estimate of the solution in \lem{l2_estimate2}, when $R_D < 1$, we have 
\begin{equation}
\max_{t\in[0,T]}\|U(t)\| \le \|U_{\mathrm{in}}\|.
\end{equation}
Therefore, the overall query complexity of our algorithm is
\begin{equation}
\begin{aligned}
\frac{1}{G^2\epsilon}sT^2D^2d^2n^4N^3\|U_{\mathrm{in}}\|^{2N}\cdot\poly\biggl(\log\frac{aDdMnNsT}{G\epsilon}\biggr)
\end{aligned}
\end{equation}
and the gate complexity is larger than its query complexity by logarithmic factors
as claimed.
\end{proof}

\section{Efficiency and lower bound results}
\label{sec:lower}

The reference \cite{LKK20} established a limitation on the ability of quantum computers to solve the quadratic ODE problem when the nonlinearity is sufficiently strong. In other words, general nonlinear differential equations are intractable on quantum computers when $R \ge \sqrt{2}$. However, we can rule out such a worst-case by assuming that the initial condition of reaction-diffusion equations fulfills the maximum principle, and thus show the problem is still tractable on quantum computers. 

In the following, we state and prove our hardness and efficiency results. 
Part (i) focuses on the hardness result when $R\geq \sqrt{2}$, which has been studies~\cite[Theorem 2]{LKK20} by leveraging the hardness result of quantum state discrimination. 
However, there is a technical flaw in the original proof in~\cite{LKK20}. 
The hardness result for quantum state discrimination used in~\cite{LKK20} only assumes multiple copies of the input states at the beginning and does not allow access to the state during the algorithm. 
But, in most quantum ODE algorithms, including the Carleman linearization method, we indeed have a stronger assumption that we assume the state preparation oracle for the input state and its \emph{inverse}, and we frequently apply those during the implementation of the algorithm. 
Therefore the existing lower bound in~\cite{LKK20} has not yet fully ruled out the possibility of efficient algorithms with strong oracle assumptions. 
We fix this gap in part (i) by applying a recent lower bound for amplifiers~\cite[Theorem 13]{AnLiuWangZhao2023}, where the state preparation oracles are assumed. 
Part (ii) shows that the worst-case scenario can be precluded by assuming the maximum principle, implying that our maximum principle analysis captures the underlying reason for the efficiency of Carleman linearization method.

\begin{theorem}\label{thm:lower}
We consider the same assumptions in \prb{node}.
\\
(i) Assume $R \ge \sqrt{2}$, and the initial condition satisfies $\|U_{\mathrm{in}}\|_{\infty}>\gamma$. Then there is an instance of the quantum quadratic ODE problem defined in \prb{node} such that any quantum algorithm for producing a quantum state approximating $u(T)/\| u(T)\|$ with bounded error must have worst-case query complexity exponential in $T$ to the input state preparation oracle.
\\
(ii) If the initial condition satisfies the maximum principle $\|U_{\mathrm{in}}\|_{\infty}\le\gamma$ as \eq{A0}, then such a worst-case example can be precluded even $R \ge \sqrt{2}$. 
\end{theorem}

\begin{proof}
We consider the lower bound result when $R \ge \sqrt{2}$ and $\|U_{\mathrm{in}}\|_{\infty}>\gamma$. The same as Theorem 2 of \cite{LKK20}, we consider a $2$-dimensional system of the form
\begin{equation}
\begin{aligned}
\frac{\d{u_1}}{\d{t}} &= -u_1 + R u_1^2, \\
\frac{\d{u_2}}{\d{t}} &= -u_2 + R u_2^2,
\end{aligned}
\label{eq:NODE2}
\end{equation}
with two single-qubit states as initial conditions
\begin{equation}
|\phi(0)\rangle = \frac{1}{\sqrt{2}} (|0\rangle + |1\rangle)
\label{eq:phi_0}
\end{equation}
and
\begin{equation}
|\psi(0)\rangle = v_0|0\rangle + w_0|1\rangle := \cos\Bigl(\theta+\frac{\pi}{4}\Bigr) |0\rangle + \sin\Bigl(\theta+\frac{\pi}{4}\Bigr) |1\rangle,
\label{eq:psi_0}
\end{equation}
where $\theta \in (0,\pi/4)$, $2\sin^2\frac{\theta}{2} = \epsilon$, with an arbitrary small $\epsilon>0$. Then the overlap between the two initial states is
\begin{equation}
\langle \phi(0) | \psi(0) \rangle := \cos\theta
= 1-\epsilon.
\label{eq:overlap_0}
\end{equation}
We let $U(t)=[v(t);w(t)]$ denote the solution evolved from $U(0)=[v_0;w_0]$. According to Lemma 8 of \cite{LKK20}, $|\phi(t)\rangle=|\phi(0)\rangle$ is the fixed state; but if $1/R\ge1/\sqrt{2}$, $w(t)$ increases with $t$ and goes to infinity after
\begin{equation}
t> t^{\ast} \coloneqq \log\biggl(\frac{R}{R-1/w_0}\biggr).
\label{eq:domain}
\end{equation}
The overlap of $|\phi(T)\rangle$ and $|\psi(T)\rangle$ is no larger than a constant (e.g., $\frac{3}{\sqrt{10}}$ used in \cite{LKK20}) after a short evolution time
\begin{equation}
T < t^{\ast} = \log\biggl(\frac{R}{R-\frac{1}{w_0}}\biggr) < \log\biggl(\frac{\sqrt{2\epsilon-\epsilon^2}+1-\epsilon}{\sqrt{2\epsilon-\epsilon^2}-\epsilon}\biggr) = \Or(\log({1}/{\epsilon})).
\label{eq:T_estimate}
\end{equation}
It was shown in~\cite[Theorem 13]{AnLiuWangZhao2023} that, if a quantum algorithm with oracle input model can amplify the infidelity of two quantum states from $\epsilon$ to a constant level, then it must use $\Omega(1/\sqrt{\epsilon})$ queries in the worst case. 
By applying this result and noticing that $1/\sqrt{\epsilon} = e^{\Omega(T)}$, we directly obtain that when $R \ge \sqrt{2}$, there is an instance of the quantum quadratic ODE problem that any quantum algorithm must have worst-case time complexity exponential in $T$.

In our paper, the ODE system \eq{NODE2} is a reduced example of reaction-diffusion equations \eq{RDPDE} with $D=0$, $a=-1$, $b=R$, $M=2$, and $R_D = 1$. Besides,  $\|U_{\mathrm{in}}\|_{\infty}
= \sin\Bigl(\theta+\frac{\pi}{4}\Bigr)$ satisfies 
\begin{equation}
1/\sqrt{2}
< \sin\Bigl(\theta+\frac{\pi}{4}\Bigr) 
= (\sqrt{2\epsilon-\epsilon^2}+1-\epsilon)/\sqrt{2}
< (1+\sqrt{2\epsilon})/\sqrt{2},
\end{equation}
where $\sin\Bigl(\theta+\frac{\pi}{4}\Bigr)$ is close to $1/\sqrt{2}$ when $\epsilon$ is close to $0$. 
Notice that this example disobeys \eq{A0}, the condition of the Maximum Principle \lem{apriori}, because
\begin{equation}
\|U_{\mathrm{in}}\|_{\infty} > \gamma = 1/\sqrt{2}.
\end{equation}

Secondly, we consider an upper bound on $\|U_{\mathrm{in}}\|$ given the maximum principle $\|U_{\mathrm{in}}\|_{\infty}\le\gamma$. Then we have
\begin{equation}
\|U_{\mathrm{in}}\| \le \sqrt{n_d}\gamma.
\end{equation}
Substituting this estimate into the complexity in \thm{main}, we can upper bound the query complexity by
\begin{equation}
\begin{aligned}
\frac{1}{G^2\epsilon}sT^2D^2d^2n^4N^3n_d^N\gamma^{2N}\cdot\poly\biggl(\log\frac{aDdMnNsT}{G\epsilon}\biggr).
\end{aligned}
\end{equation}
Notice that the upper bound of the complexity still depends exponentially on $N$. 
However, according to~\thm{Carleman}, the Carleman error converges exponentially in $N$ and can be bounded independently of $T$. 
So $N$ can be chosen independently of $T$ as well, therefore our algorithm does not have the exponential overhead in $T$ stated in part (i). 
\end{proof}

The upper bounds on the query and gate complexity demonstrate that the quantum algorithm we develop has a roughly quadratic dependence on $T$ when $R_D<1$ and $\|U_{\mathrm{in}}\|_{\infty}\le\gamma$, regardless of the value of $R$. Such a loose upper bound includes a polynomial dependence on $n$, revealing that the quantum algorithm does not have a potential exponential speedup in the dimension. 

But if we are given an additional assumption
\begin{equation}
\|U_{\mathrm{in}}\| = o(n_d),
\end{equation}
then we can upper bound the query complexity by
\begin{equation}
\begin{aligned}
\frac{1}{G^2\epsilon}sT^2D^2d^2n^4N^3\cdot\poly\biggl(\log\frac{aDdMnNsT}{G\epsilon}\biggr),
\end{aligned}
\end{equation}
and the gate complexity has an upper bound that is larger by logarithmic factors as claimed. 
In this case, our algorithm still maintains the potential exponential speedup in the dimension over classical algorithms.

\section{Applications}
\label{sec:application}

In this section, we show how the quantum state obtained by solving \prb{node} can be used to compute quantities of practical interest. 
For generality, in this section, we consider the applications of a quantum state in the form specified in \prb{node}, but not limited to the output by particular algorithms. 
More specifically, let $l=(l_1,\ldots,l_d)$ with $l_j\in\range{n}$, and let $f(t,x)$ be a function defined on $[0,T]\otimes [0,1]^d$. We assume that there exists a quantum algorithm $\mathcal{A}(\epsilon)$ which can prepare the quantum state
\begin{equation}\label{eq:app_state}
|\hat{f}\rangle \sim \sum_{k=0}^{m-1}\sum_{l_1=0}^{n-1}\cdots \sum_{l_d=0}^{n-1} \hat{f}_{k,l} \ket{k}|l_1\rangle\ldots|l_d\rangle 
\end{equation}
proportional to the vector $(f(kT/m,l_1/n,\cdots,l_d/n))$ within some prescribed error tolerance $\epsilon>0$ in $\ell^2$ norm. 
Here $\hat f_{k,l}$ represents an approximation of the function $f$ evaluated at $(kT/m, l_1/n,\cdots, l_d/n)$. 
In the context of this paper, the algorithm $\mathcal{A}$ is the quantum Carleman linearization method, but the discussion in this section works for any quantum algorithm which can encode a function evaluated at discrete grid points. 

\subsection{Mean square amplitude}

One quantity of potential practical interest is the fraction of the squared amplitude contained in a sub-domain $\D_t \times \D_x$ defined by $\D_t \subset [0,T]$ and $\D_x \subset [0,1]^d$. 
This can be described as the ratio 
\begin{equation}\label{eq:app_ratio_cont}
    \frac{\int_{\D_t}\int_{\D_x}|f(t,x)|^2 dx dt}{\int_{0}^{T}\int_{[0,1]^d}|f(t,x)|^2 dx dt}. 
\end{equation}
In the context of quantum mechanics such quantities are motivated by Born's rule, whereas in the context of classical wave mechanics such quantities are motivated by notions of energy.

In the spatial discretization context, we can approximate the integrals via numerical quadrature with equidistant nodes, and thus we are interested in computing the ratio 
\begin{equation}\label{eq:app_ratio}
    \frac{\sum_{kT/m \in \D_t} \sum_{(l_1/n,\cdots,l_d/n)\in \D_x} |f(kT/m,l_1/n,\cdots,l_d/n)|^2 }{\sum_{k=0}^{m-1} \sum_{l_1=0}^{n-1} \cdots \sum_{l_d=0}^{n-1} |f(kT/m,l_1/n,\cdots,l_d/n)|^2 }. 
\end{equation}
Note that the difference between \eq{app_ratio_cont} and \eq{app_ratio} scales as $\Or(T(T/m+d/n))$~\cite{BurdenNA2000}, so it can be reduced to the level of $\Or(\epsilon)$ with arbitrarily small $\epsilon$ by refining time and spatial discretization as $m = \Or(T^2/\epsilon), n = \Or(dT/\epsilon)$. 
For simplicity, here we fix the grids for discretization and focus on computing the discretized ratio as shown in \eq{app_ratio}. 
This quantity can be easily estimated given the quantum state in the form of \eq{app_state}. 
Let $P$ be the projector onto the space spanned by the computational basis corresponding the $\Omega$, i.e., $P = \sum_{kT/m\in\D_t}\sum_{(l_1/n,\cdots,l_d/n)\in\D_x} (\ket{k}\bra{k})\otimes (\ket{l_1}\bra{l_1})\otimes \cdots \otimes (\ket{l_d}\bra{l_d})$. 
Then 
\begin{equation}
\begin{split}
    &\quad  \frac{\sum_{kT/m\in\D_t} \sum_{(l_1/n,\cdots,l_d/n)\in \D_x} |f(kT/m,l_1/n,\cdots,l_d/n)|^2 }{\sum_{k=0}^{m-1} \sum_{l_1=0}^{n-1} \cdots \sum_{l_d=0}^{n-1} |f(kT/m,l_1/n,\cdots,l_d/n)|^2 } \\
    & \approx \frac{\sum_{kT/m\in\D_t} \sum_{(l_1/n,\cdots,l_d/n)\in \D_x} |\hat{f}_{k,l}|^2 }{\sum_{k=0}^{m-1} \sum_{l_1=0}^{n-1} \cdots \sum_{l_d=0}^{n-1} |\hat{f}_{k,l}|^2 } \\
    & = \langle \hat{f}| P | \hat{f} \rangle
\end{split}
\end{equation}
and thus can be estimated by the amplitude estimate technique~\cite{BHM02}. 

The complexity of such an algorithm is given in the following theorem, and the proof can be found in the appendix \sec{append_proof_app}. 
\begin{theorem}\label{thm:app_ratio}
    Assume that we are given an algorithm $\mathcal{A}(\epsilon')$ preparing \eq{app_state} and a unitary transform $(I-2P)$ for the projector $P$ associated with the domain $\D$. 
    Then, for any $0 < \epsilon < 1$, $0 < \delta < 1$, there exists a quantum algorithm which can output an approximation of \eq{app_ratio} within tolerated error $\epsilon$ and probability at least $(1-\delta)$, using $\mathcal{A}(\epsilon/4)$ and $(I-2P)$ for  $\Or((1/\epsilon)\log(1/\delta))$ times.  
\end{theorem}

\thm{app_ratio} proves an almost linear sampling scaling in terms of the precision for the quantum approach to estimate the ratio.
Compared with the most widely used classical approach for estimating high dimensional integral, standard Monte Carlo methods, which typically scale quadratically in the precision, we can obtain a quadratic speedup for the sampling step.

\subsection{Derivatives and kinetic energy}

As a particular example of practical interest, we would like to study the dynamical kinetic energy ratio of the system on the domain $\D_t\times\D_x \subset [0,T]\times[0,1]^d$, which is defined to be 
\begin{equation}\label{eq:app_ratio_free_energy}
    \frac{\int_{\D_t} \int_{\D_x} |\nabla_x f(t,x)|^2 dx dt}{\int_0^T \int_{[0,1]^d} |\nabla_x f(t,x)|^2 dx dt} \approx \frac{\sum_{kT/m\in\D_t} \sum_{(l_1/n,\cdots,l_d/n)\in \D_x} |\nabla_x f(kT/m,l_1/n,\cdots,l_d/n)|^2 }{\sum_{k=0}^{m-1} \sum_{l_1=0}^{n-1} \cdots \sum_{l_d=0}^{n-1} |\nabla_x f(kT/m,l_1/n,\cdots,l_d/n)|^2 }.  
\end{equation}
To apply our result in \thm{app_ratio}, we first study how to prepare a quantum state which is proportional to the partial derivative of the function $f$. 
That is, our goal is to prepare a quantum state which is a good approximation of 
\begin{equation}\label{eq:app_deriv_state}
    \ket{\nabla_x f} \sim \sum_{j=1}^{d}\sum_{k=0}^{m-1} \sum_{l_1=0}^{n-1} \cdots \sum_{l_d=0}^{n-1} \nabla_{x_j} f(kT/m,l_1/n,\cdots,l_d/n) \ket{j}\ket{k}\ket{l_1}\cdots\ket{l_d}. 
\end{equation}
Compared to \eq{app_state}, here we introduce one more register to simultaneously encode all the partial derivatives in a single quantum state. 
In this subsection, we further assume that the function $f(t,x)$ satisfies periodic boundary conditions for spatial variable $x$. 
Notice that the periodic boundary condition here is only for technical simplicity and not essential because a general function can be smoothly extended to a larger space domain with periodic boundary conditions. Our algorithm will also work by requiring access to another projector connecting the original and the extended space domains. 

We first study how to prepare a quantum state encoding the derivative information within amplitudes. 
The idea is using the discrete Fourier transform to transform the state to the frequency domain, multiplying the frequency in this domain and then transforming back to the space domain by inverse discrete Fourier transform. 
More specifically, let 
$\mathcal{F}$ denote the one-dimensional quantum Fourier transform\footnote{Note that, following the standard convention of choosing signs in~\cite{NielsenChuang2000}, the quantum Fourier transform exactly corresponds to the inverse discrete Fourier transform in the classical setting.} with $n$ nodes.
Furthermore, for any positive integer $\theta \leq n/2$, let $l = (l_1,\cdots,l_d)$ and $D_{j,\theta}$ be a diagonal matrix 
    \begin{equation}
        D_{j,\theta} = 2\pi i \theta \sum_{l_1,\cdots, l_d = 0}^{n-1} D_{j,\theta}(l) \ket{l}\bra{l} 
    \end{equation}
    where 
    \begin{equation}
        D_{j,\theta}(l) = \begin{cases}
           l_j/\theta , & \text{ if } 1 \leq l_j \leq \theta, \\
          (l_j - n)/\theta , & \text{ if } n-\theta \leq l_j \leq n-1, \\
          0 , & \text{ else}.
        \end{cases}
    \end{equation}
Notice that we define $D_{j,\theta}(l)$ such that $|D_{j,\theta}| \leq 1$. 
Then, for any smooth function $g(x)$ defined on $[0,1]^d$, we have 
\begin{equation}
\begin{split}
    & \quad (\otimes_{j-1}I \otimes \mathcal{F} \otimes_{d-j} I)D_{j,\theta}(\otimes_{j-1}I \otimes \mathcal{F}^{-1} \otimes_{d-j} I)\sum_{l_1=0}^{n-1}\cdots \sum_{l_d=0}^{n-1} g(l_1/n_1,\cdots,l_d/n_d) \ket{l_1}\cdots \ket{l_d} \\
    &\approx \sum_{l_1=0}^{n-1}\cdots \sum_{l_d=0}^{n-1} \partial_{x_j} g(l_1/n_1,\cdots,l_d/n_d) \ket{l_1}\cdots \ket{l_d}. 
\end{split}
\end{equation}
Such an approach has been widely used and is regarded as the standard way to compute derivatives in classical scientific computing, and we briefly illustrate the reasoning in the appendix \sec{append_DFT}.

Now we discuss how to implement this approach on a quantum device. 
In the general high-dimensional case, we introduce another ancilla register with $\Or(\log d)$ qubits, which we will refer to as the dimension register later, and start with the state 
\begin{equation}
    \ket{0}\ket{f} =  \frac{1}{\|\vec{f}\|} \sum_{k=0}^{m-1}\sum_{l_1,\cdots,l_d = 0 }^{n-1} f(kT/m, l_1/n,\cdots,k_d/n) \ket{0}\ket{k}\ket{l_1}\cdots\ket{l_d}. 
\end{equation}
Applying the Hadamard gates to the dimension register, we obtain 
\begin{equation}
    \frac{1}{\sqrt{d}\|\vec{f}\|} \sum_{k=0}^{m-1}\sum_{l_1,\cdots,l_d = 0 }^{n-1} \sum_{j=0}^{d-1} f(k_1/n,\cdots,k_d/n) \ket{j}\ket{k}\ket{l_1}\cdots\ket{l_d}. 
\end{equation}
The discrete Fourier transform can be efficiently implemented via quantum Fourier transform. 
Specifically, we apply the operation $(\sum_{j=0}^{d-1}\ket{j}\bra{j} \otimes I_{n_s}\otimes \mathcal{F}_j^{-1})$ and denote the resulting state as 
\begin{equation}
    \frac{1}{\sqrt{d}\|\vec{f}\|} \sum_{k=0}^{m-1}\sum_{l_1,\cdots,l_d = 0 }^{n-1} \sum_{j=0}^{d-1} \tilde{f}(j,k,l) \ket{j}\ket{k}\ket{l_1}\cdots\ket{l_d}. 
\end{equation}
For the multiplication of the matrix $D_{j,\theta}$ (we will show later that $\theta$ can be chosen as an $\Or(1)$ parameter for smooth functions), we assume that we are given an oracle of the mapping 
\begin{equation}\label{eq:app_O_D}
    O_D: \ket{0}\ket{j}\ket{l} \rightarrow \ket{D_{j,\theta}(l)}\ket{j}\ket{l}. 
\end{equation}
Then the multiplication of the matrix $D_{j,\theta}$ can be implemented as follows. 
We first add two ancilla registers, one as the rotation register on which we will perform conditional rotation later and the other as the $D$-register for encoding $D_{j,\theta}$. 
Applying $O_D$ to encode $D_{j,
\theta}(l)$ in the $D$-register gives 
\begin{equation}
    \frac{1}{\sqrt{d}\|\vec{f}\|} \sum_{k=0}^{m-1}\sum_{l_1,\cdots,l_d = 0 }^{n-1} \sum_{j=0}^{d-1} \tilde{f}(j,k,l) \ket{0}\ket{D_{j,\theta}(l)}\ket{j}\ket{k}\ket{l_1}\cdots\ket{l_d}. 
\end{equation}
Performing a rotation on the rotation register conditioned on $\ket{D_{j,\theta}(l)}$ yields 
\begin{equation}
    \frac{1}{\sqrt{d}\|\vec{f}\|} \sum_{k=0}^{m-1}\sum_{l_1,\cdots,l_d = 0 }^{n-1} \sum_{j=0}^{d-1} \tilde{f}(j,k,l)\left(D_{j,\theta}(l)\ket{0} + \sqrt{1-D_{j,\theta}(l)^2} \ket{1}\right) \ket{D_{j,\theta}(l)}\ket{j}\ket{k}\ket{l_1}\cdots\ket{l_d}. 
\end{equation}
Uncomputing the $D$-register gives the state
\begin{equation}
    \frac{1}{\sqrt{d}\|\vec{f}\|} \sum_{k=0}^{m-1}\sum_{l_1,\cdots,l_d = 0 }^{n-1} \sum_{j=0}^{d-1} D_{j,\theta}(l)\tilde{f}(j,k,l)\ket{0}\ket{0}\ket{j}\ket{k}\ket{l_1}\cdots\ket{l_d} + \ket{\perp}, 
\end{equation}
where the first part is the desired outcome after the diagonal transformation, and $\ket{\perp}$ represents a quantum state with the rotation register being $\ket{1}$. 
Finally, applying $(\sum_{j=0}^{d-1} \ket{j}\bra{j} \otimes \mathcal{F}_j)$ on the registers $(\ket{j}\otimes \ket{l})$ completes the operation for computing the partial derivatives as discussed before, which yields an approximation of 
\begin{align}
     \frac{1}{2\pi i \theta } \frac{1}{\sqrt{d}\|\vec{f}\|} \sum_{k=0}^{m-1} \sum_{l_1,\cdots,l_d = 0 }^{n-1} \sum_{j=0}^{d-1} \partial_{x_j} f(kT/m, l_1/n,\cdots,l_d/n) \ket{0}\ket{0} \ket{j}\ket{k}\ket{l_1}\cdots\ket{l_d} + \ket{\perp}. 
\end{align}
By measuring the ancilla rotation register to get $0$ and discarding the $D$-register, we get a quantum state approximately proportional to 
\begin{equation}\label{eq:application_dev_state}
    \sum_{k=0}^{m-1}\sum_{l_1,\cdots,l_d = 0 }^{n-1} \sum_{j=0}^{d-1} \partial_{x_j} f(k_1/n,\cdots,k_d/n) \ket{j}\ket{k}\ket{l_1}\cdots\ket{l_d}, 
\end{equation}
which encodes the partial derivatives in the amplitude controlled by a dimension register. 
The entire quantum circuit is summarized in Fig. \ref{fig:circuit_derivative}, and the overall complexity estimate is given in the following theorem, of which the proof can be found in the appendix \sec{append_proof_app}. 
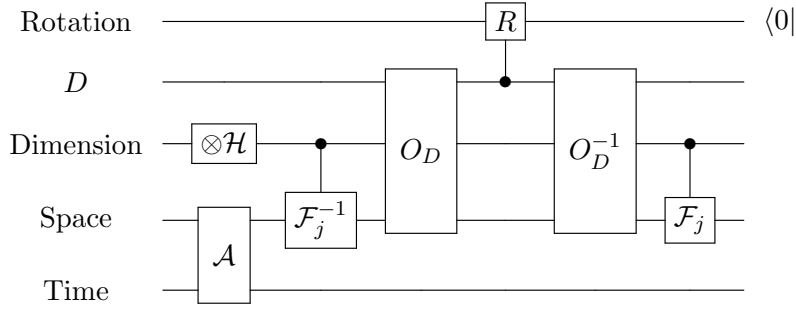
\begin{figure}
    \centerline{
    \Qcircuit @R=1em @C=1em {
    \text{Rotation}\quad\quad\quad\quad\quad\quad &  \qw & \qw & \qw & \gate{R} & \qw & \qw & \qw & \rstick{\!\!\!\!\!\bra{0}} \\
    \text{$D$}\quad\quad\quad\quad\quad\quad  &  \qw & \qw & \multigate{2}{O_D} & \ctrl{-1} & \multigate{2}{O_D^{-1}} & \qw & \qw \\
    \text{Dimension}\quad\quad\quad\quad\quad\quad  & \gate{\otimes \mathcal{H}} & \ctrl{1} & \ghost{O_D} & \qw & \ghost{O_D^{-1}} & \ctrl{1} & \qw \\
    \text{Space}\quad\quad\quad\quad\quad\quad &  \multigate{1}{\mathcal{A}} & \gate{\mathcal{F}_j^{-1}} & \ghost{O_D} & \qw & \ghost{O_D^{-1}} & \gate{\mathcal{F}_j} & \qw \\
    \text{Time}\quad\quad\quad\quad\quad\quad  &  \ghost{\mathcal{A}} & \qw & \qw & \qw & \qw & \qw & \qw \\
    }
    }
    \caption{ Quantum circuit for preparing a quantum state encoding partial derivatives of a known function in amplitudes. Here $\mathcal{A}$ is the algorithm for the state encoding the function,  $\mathcal{H}$ represents the Hadamard gate, $\mathcal{F}_j$ represents the one-dimensional quantum Fourier transform acting on the $j$-th direction, $O_D$ is the oracle specified in \eq{app_O_D}, and $R$ is the rotation operation. }
    \label{fig:circuit_derivative}
\end{figure}

\begin{theorem}\label{thm:app_deriv_state_smooth}
    Let $f(t,x)$ be a smooth function and $\vec{f}$ be a possibly unnormalized vector
    \begin{equation}
        \vec{f} = \sum_{k=0}^{m-1}\sum_{l_1=0}^{n-1}\cdots \sum_{l_d=0}^{n-1} f(kT/m,l_1/n_1,\cdots,l_d/n_d) \ket{k}\ket{l_1}\cdots \ket{l_d}. 
    \end{equation}
    Assume that $f$ satisfies periodic boundary condition for $x$ and $\sup_{j,p} (\|\partial_{x_j}^p f\|_{\infty})^{1/p} < \pi n$. 
    Then for any $0 < \epsilon < 1$, $0 < \delta < 1$, there exists a quantum algorithm which, with probability at least $(1-\delta)$, outputs an $\epsilon$-approximation of the state proportional to 
    \begin{equation}
        \vec{\nabla f} = \sum_{j=0}^{d-1}\sum_{k=0}^{m-1}\sum_{l_1,\cdots,l_d=0}^{n-1} \partial_{x_j} f(kT/m,l_1/n,\cdots,l_d/n) \ket{j}\ket{k}\ket{l_1}\cdots\ket{l_d}, 
    \end{equation}
    using queries to $\mathcal{A}(\epsilon/Q)$ and $O_D$ for $\Or(Q\log(1/\delta))$ times and additional $\Or(d (\log n)^2)$ gates, where 
    $$Q = \frac{4\sqrt{d}\|\vec{f}\| (\sup_{j,p} (\|\partial_{x_j}^p f\|_{\infty})^{1/p}+1)}{\|\vec{\nabla f}\|}. $$
\end{theorem}

We briefly compare our result with the standard classical approach of computing gradient in terms of the dimension parameters, including $n$ and $d$. 
On the one hand, the cost for computing gradient evaluated at all discrete grid points typically scale $\Or(d n^d)$, since the sizes of both the vector storing the information of the gradients and the matrices related to finite difference and discrete Fourier transform scale $\Or(d n^d)$. 
On the other hand, the corresponding scaling for our quantum approach is subtler since it depends on the scaling of the quantity $Q$. 
Notice that this quantity can be very large if the function $f$ is close to a constant function. 
However, it can also be of $\Or(1)$ if at least one of the non-trivial Fourier components in $f$ is on the order of $\Omega(1)$, since in this case $\|\vec{f}\| \sim \Or(\sqrt{mn^d})$ and $\|\vec{\nabla f}\| \sim 
\Omega(\sqrt{dmn^d})$. 
In this scenario, the overall complexity scales only polynomially in terms of $d$ and poly-logarithmically in $n$, which achieves exponential speedup in terms of the dimension parameters. 

The cost of computing the ratio of the kinetic energy can be directly estimated by combining \thm{app_ratio} and \thm{app_deriv_state_smooth}. 
As discussed before, we can get a quadratic speedup in terms of precision compared to the classical Monte Carlo method. 
The result is summarized in the following Corollary. 

\begin{corollary}\label{cor:app_deriv}
    Let $f(t,x)$ be a smooth function such that $f$ satisfies periodic boundary condition for $x$, and $\sup_{j,p} (\|\partial_{x_j}^p f\|_{\infty})^{1/p} < \pi n$. 
    Assume that we are given a unitary transform $(I-2P)$ for the projector $P$ associated with the domain $\D$. 
    Then for any $0 < \epsilon < 1$, $0 < \delta < 1$, there exists a quantum algorithm which can output an approximation of \eq{app_ratio_free_energy} within tolerated error $\epsilon$ and probability at least $(1-\delta)$, using queries to $(I-2P)$ for $\Or((1/\epsilon)\log(1/\delta))$ times, queries to $\mathcal{A}(\epsilon/(4Q))$ and $O_D$ for $\Or( Q(\log(1/\delta))^2/\epsilon)$ times and additional $\Or(d (\log n )^2 \log(1/\delta) /\epsilon)$ gates, where 
    $$Q = \frac{4\sqrt{d}\|\vec{f}\| (\sup_{j,p} (\|\partial_{x_j}^p f\|_{\infty})^{1/p}+1)}{\|\vec{\nabla f}\|}. $$
\end{corollary}

\subsection{History state and decay of kinetic energy}

Unlike the existing work on quantum differential equation solvers that typically output a final state encoding the solution at the final time, our Carleman linearization algorithm and the input model assumed in this section take a more general history state encoding the solutions at all time steps. 
In this subsection, we briefly discuss how the general history state structure may broaden the application of our algorithm. 

One potential application is to study the kinetic energy curve. 
As discussed in \cite{Perthame2015,FJMP19}, applied mathematicians are interested in the curve that describes the decay of the kinetic energy of the gradient flow \eq{energy-functional}, and particularly the time of reaching the equilibrium, \emph{i.e.}, the time when the kinetic energy almost stops changing. 
We define this time to be the equilibrium time $t^*$, and it can be easily estimated by combining the history state and our algorithm for computing gradients in Fig.~\ref{fig:circuit_derivative}. 
Specifically, we first run the algorithm in Fig.~\ref{fig:circuit_derivative} to get an approximation of the history state of gradients of the solution as in \eq{app_deriv_state}, and then measure the time register $\ket{k}$ to obtain an integer $k$. 
Notice that after the equilibrium time $t^*$ when the kinetic energy almost stops changing, the corresponding partial derivatives are very close to 0. 
This implies that we have almost no measurement outcomes after $t^*$. 
By repeating such a procedure and taking the maximum of the measure outcomes $k_{\max}$, we can use $k_{\max}h$ to estimate $t^*$ with high probability. 

The history state structure can also allow us to overcome potential exponential cost in the differential equation solvers caused by the decay of the solution. 
To the extent of our knowledge, most of the existing quantum differential equation solvers which output a final state~\cite{Ber14,BCOW17,CL19,LKK20} scale at least linearly in terms of the parameter $\|u_{\text{in}}\|/\|u_{\text{out}}\|$, where $u_{\text{in}}$ and $u_{\text{out}}$ denote the unnormalized solutions of the differential equation at the initial and final time, respectively. 
Such a linear dependence is typically caused by post-processing a quantum state obtained from solving specific linear systems of equations to get the desired final state and may introduce extra exponential time dependence if the solution of the differential equations experiences rapid decay. 
A simple example is the imaginary time evolution 
\begin{equation}
    \frac{du}{ds}  = - H u
\end{equation}
where $H$ is a positive definite matrix. 
Here $\|u(T)\| = \|\exp(-HT)u(0)\| \leq \exp(-\lambda T)\|u(0)\|$ with $\lambda$ being the smallest eigenvalue of $H$, and thus $\|u(0)\|/\|u(T)\| \geq \exp(\lambda T)$, leading to an extra exponentially large term in $T$. 
Another example is the nonlinear ordinary differential equation with no constant term, which can be explicitly written down as 
\begin{equation}
    \frac{du}{ds} = F_1 u + F_2 u^{\otimes 2}. 
\end{equation}
It is proved in~\cite{LKK20} that, if $F_1$ only has negative eigenvalues and the nonlinearity is relatively weak in certain sense (namely the parameter $R$ defined in this work is smaller than 1), then $\|u(t)\|$ decays exponentially in terms of $t$, which also leads to an exponentially large $\|u_{\text{in}}\|/\|u_{\text{out}}\|$ in $T$. 

With a history state, we can run a ``pre-diagnosis'' to first identify whether the final state is close to 0. 
The key observation is that if the state has sufficiently decayed such that the solution is very close to 0, then the success probability of getting the corresponding time step by measuring the time register is exponentially small. 
In particular, assume that we are interested in obtaining an approximation of the final solution $u(T)$ when $\|u(T)\|$ is very small or the corresponding state $u(T)/\|u(T)\|$ when $\|u(T)\|$ is reasonably away from 0. 
We can repeatedly prepare a history state, measure the time register, and obtain the output $k$. 
If for all $k$'s, we have $kh < T$, then the final solution is expected to be exponentially small with high probability, so we can stop here and directly use 0 to be the approximation of the solution. 
On the other hand, if there exists a $k$ such that $kh = T$, then this is a reasonable indication that the quantity $\|u_{\text{in}}\|/\|u_{\text{out}}\|$ is not quite large, and we can follow the standard post-processing procedure to obtain the final state $u(T)/\|u(T)\|$. 
The history state here helps us determine which scenario the differential equation is in without exponential cost in the evolution time $T$.

\section{Discussion}
\label{sec:discussion}

We have presented an efficient quantum algorithm for the gradient flow evolution of reaction-diffusion equations. We improve the Carleman linearization under a  condition $R_D < 1$. It relaxes the previous condition $R<1$ in \cite{LKK20} for high-dimensional systems of nonlinear differential equations.
Besides, we discussed estimating the mean square amplitude and ratios of the kinetic energy of the gradient flow as potential applications. 

This work raises several natural open problems. The first aspect is regarding further improvement of our algorithm. Though our work focuses on improving the convergence condition for the Carleman linearization, it is also interesting to seek further improvement of the dependence on other parameters in the complexity, such as the evolution time and the error tolerance. Another related topic is to obtain meaningful classical outputs with super-polynomial quantum speedups over the best-known classical algorithms. 
It is also an interesting question to generalize our algorithms or design new quantum algorithms dealing with other types of nonlinear differential equations. Our analysis relies heavily on the Maximum Principle and some good regularity of the solutions, which are essential properties of reaction-diffusion equations. On a high level, the maximum principle controls the norm of the solution that shares the same spirit of spectral norm preserving properties of Hamiltonian simulation. It is thus an interesting direction to consider other norm-controlled problems, such as the gradient flow structured on certain norm preserving manifold, and see whether such confinement helps conquer the nonlinearity at hand.  

However, relaxing the regularity assumptions of the solutions seems to be a fundamentally difficult problem. All linearization-based techniques require the solutions to be well-posed and regular, which is not the case for applications such as the conservation laws, fluid dynamics, and Hamilton-Jacobi equation, where the solutions blow up in finite time. Therefore, other approaches beyond the linearization framework are desired for such nonlinear problems, requiring some new insights. We also point out that the quantum Carleman linearization based approaches may suffer from an overhead sensitive to the spatial grid refinement for partial differential equations. Though our improved $\ell_\infty$ framed convergence criterion can concur the sensitivity to grid refinements, this sensitivity can be reintroduced through the $\ell_2$ norm dependence when implementing the quantum algorithm since the $\ell_2$ norm of the solution discretized by finite difference still grows as more grid points are used. In particular, it appears when getting the solution from the huge Carleman state vector. 
It is our future work to make the quantum algorithm fully insensitive to the spatial grid refinements, probably by trying other spatial discretization which can preserve the $\ell_2$ norm in grid refinements such as Fourier discretization. 

In the application section, we use the amplitude estimate technique to obtain classical information beyond the quantum state output and study the kinetic energy distribution by preparing a state with derivative information. 
The techniques we propose in this section do not rely on specific models or differential equations. Thus it is interesting to find applications of our technique to output classical information for other problems, such as phase separation and transition, chemical reactions, and self-organized biological patterns. 
We want to remark that all the classical outputs we study in this work are in terms of ratios. 
While obtaining the absolute value rather than the ratio seems to require accurate computation of the observable on the entire domain, which might incur exponential overhead, it is still very interesting to further study whether the absolute value of observables can be approximated with only a relatively small overhead. 
Our algorithm for preparing quantum states of derivatives requires the regularity of the function, and we want to understand its performance for non-smooth functions as well. 
It may also be of interest in some quantum optimization problems which require gradient information to optimize the objective function and will be our future work.

\section*{Acknowledgments}

We thank Andrew M.\ Childs and Lin Lin for valuable discussions. DA acknowledge the hospitality of the Simons Institute for the Theory of Computing in Berkeley. DA, DF and JW acknowledge the Challenge Institute for Quantum Computation funded by NSF through grant number OMA-2016245, and the Department of Energy under grant No. DE-SC0017867.
DA acknowledges the support by the Department of Defense through the Hartree Postdoctoral Fellowship at QuICS, and the NSF under Grant No. DMS-1652330. 
DF is supported by the NSF grant number DMS-2208416.
JPL acknowledges support from the Department of Energy, Office of Science, Office of Advanced Scientific Computing Research, Quantum Algorithms Teams and Accelerated Research in Quantum Computing programs, the National Science Foundation award (CCF-1813814, DMS-2008568), and from the National Science Foundation Quantum Information Science and Engineering Network (QISE-NET) triplet award (DMR-1747426).

\bibliographystyle{unsrt}
\bibliography{npde}

\appendix

\section{Proof of \texorpdfstring{\lem{apriori}}{Lemma \ref{lem:apriori}}}
\label{sec:proof_apriori}

We now discuss the proof of the a priori estimate of the solution. Before that, we introduce the comparison principle lemma for the discrete reaction-diffusion equation \eqref{eq:GAC}, which implies the a priori estimate as a direct consequence. 

\begin{lemma}[Comparison principle]\label{lem:comp_principle}
Assume $f(u) \in C^\infty(\mathbb{R})$ and $U(t) = [U_1, \cdots, U_n]^T$, $V(t)= [V_1, \cdots, V_n]^T$ are continuous functions that satisfy
\begin{equation*}\label{eqn:comparison_principle_condition}
    \frac{\d U_j}{\d t} - D \left(\Delta_h U \right)_j - f(U_j) \le  \frac{\d V_j}{\d t} - D \left(\Delta_h V \right)_j - f(V_j),
\end{equation*}
for $t \in (0, T]$ and all the multi-indices $j \in \mathcal{I}$. Furthermore, $U_j(0) \le V_j(0)$ for all $j \in \mathcal{I}$ and $U_j(t) \le V_j(t)$ for $j \in \mathcal{B}$ and $t\in [0, T]$. Then it holds that
\begin{equation*}
    U_j(t) \le V_j(t)
\end{equation*}
for all $j$ and all time $t\in (0, T]$.
\end{lemma}

We remark that Lemma \ref{lem:comp_principle} immediately implies Lemma \ref{lem:apriori}, because both $U_j(t) = \gamma_1$ and $U_j(t) = \gamma_2$ all $j$ and $t$ are solutions to \eqref{eq:GAC} and hence comparison principle can be applied to the solution of interests and these constant-valued equilibrium solutions, which yields the desired result.

\begin{proof}
Our proof can be split into the following three steps: in the first two steps, we consider a \textit{linear} operator of $U$ given by $\frac{\d U_j}{\d t} - D \left(\Delta_h U \right)_j + \tilde{C}_j U_j$ for some vector $\tilde{C} = \left( \tilde{C_{j}}\right)$, and show that a maximum principle result for this linear operator; and in the last step, we prove the comparison principle for the nonlinear problem as considered in \eqref{eqn:comparison_principle_condition}. It is worth pointing out that although a linear problem is considered at first, there is no linearization procedure introducing any extra error here.

First, we claim that if 
\begin{equation}\label{eqn:comparison_principle_pf_step1_condition}
    \frac{\d V_j}{\d t} - D \left(\Delta_h V \right)_j + \tilde{C}_j V_j < 0,
\end{equation}
for some $\tilde{C}$ with all entries positive, then it holds that
\[
\displaystyle\max_{t\in [0, T], j\in \mathcal{I}} V_j (t) 
= \displaystyle\max_{(j,t)\in (\mathcal{B}\times (0,T]) \cup (\mathcal{I} \times \{0\}) } V^+_{j}(t)  
\]
where $\mathcal{B}$ denotes the boundary indices and $V^+_{j} : = \max(V_j, 0)$.

Suppose the claim does not hold, then there exists some time $t_0 \in (0, T]$ and $j_0 \not\in \mathcal{B}$ such that $V_j(t)$ attains the positive maximum value at $(t_0, j)$. Note that the definition \eqref{eqn:dis_lap_highD_tensor} of $\Delta_h$ can be written in the following equivalent form
\begin{align} 
\Delta_h v_{j_1, \cdots, j_n} : = & \frac{v_{j_1+1, \cdots, j_n}  - 2 v_{j_1, \cdots, j_n} + v_{j_1-1, \cdots, j_n}}{h^2} + \frac{v_{j_1,  j_2+1, \cdots, j_n}  - 2 v_{j_1, j_2, \cdots, j_n} + v_{j_1, j_2-1 \cdots, j_n}}{h^2} \nonumber \\
&+ \cdots + \frac{v_{j_1, \cdots, j_n+1}  - 2 v_{j_1, \cdots, j_n} + v_{j_1, \cdots, j_n-1}}{h^2}, \label{eqn:def_lap_h_eltwise}
\end{align}
and hence one has $-\left(\Delta_h V \right)_{j_0} \geq 0$. Meanwhile, if $t_0\in (0,T)$, then 
\[
\frac{\d V_{j_0}}{\d t} (t_0) =0;
\]
otherwise, $t_0 = T$, and one has 
\[
\frac{\d V_{j_0}}{\d t} (t_0) \geq 0.
\]
Therefore, one has 
\[
\frac{\d V_j}{\d t} - D \left(\Delta_h V \right)_j + \tilde{C}_j V_j \geq 0,
\]
which is a contradiction. This completes the proof of the claim.

In the second step, we relax the condition of the claim from the following two angles: the condition on $C$ is changed from positive to bounded from below; and the equality is allowed. To be precise, we shall show that for $C$ satisfying that $C_j \geq c_{\rm min}$ for all $j \in \mathcal{I}$, if 
\begin{equation}\label{eqn:comparison_principle_pf_step2_condition}
    \frac{\d V_j}{\d t} - D \left(\Delta_h V \right)_j + \tilde{C}_j V_j \le 0,
\end{equation}
then 
\begin{equation}\label{eqn:comparison_principle_pf_step2_conclusion}
\displaystyle\max_{t\in [0, T], j \in \mathcal{I}} V_j (t) = \displaystyle\max_{(j,t)\in (\mathcal{B}\times (0,T]) \cup (\mathcal{I} \times \{0\}) } V^+_{j}(t). 
\end{equation}
To prove this, we define $\tilde{V}(t)$ by the change of variable $\tilde{V}(t) = e^{c_{\rm min} t} V(t) - \delta t$. A straightforward calculation yields
\[
 \frac{\d \tilde{V}_j}{\d t} = c_{\rm min} \tilde{V}_j + e^{c_{\rm min} t}  \frac{\d V_j}{\d t} -\delta 
 \le  D \left(\Delta_h \tilde{V} \right)_j - \left(\tilde{C}_j - c_{\rm min} \right) \tilde{V}_j  - \delta.
\]
Thus $\tilde{V}_j$ satisfies \eqref{eqn:comparison_principle_pf_step1_condition}. Letting $\delta \to 0_+$ yields the desired result.

The last part of the proof is to show the comparison principle. Let $W: = U-V$. One has $W_j(0) \le 0$ for $j \in \mathcal{I}$ and $W_j(t) \le 0$ for all $t\in (0,T]$ and $j \in \mathcal{B}$ so that the right-hand-side of \eqref{eqn:comparison_principle_pf_step2_conclusion} is non-positive. Moreover, $W$ satisfies
\[
\frac{\d W_j}{\d t} - D \left(\Delta_h W\right)_j - f(U_j) + f(V_j) \le 0,
\]
and by mean value theorem $f(U_j) - f(V_j) = f'(\xi_j)$ with $\xi_j$ in between $U_j$ and $V_j$, we arrive at 
\[
    \frac{\d W_j}{\d t} - D \left(\Delta_h W\right)_j - f'(\xi_j) W_j \le 0.
\]
Applying the result of the second step yields
\[
W_j (t) \le 0,
\]
for all $j \in \mathcal{I}$ and $t \in [0, T]$, which completes the proof of this lemma.
\end{proof}

\section{Matrix inequality}
\label{sec:matrix_inequality}

\begin{lemma}[Discrete maximal principle]\label{lem:discrete_max_principle} Let $D_h$ be the discrete $1$-dimensional Laplacian operator with homogeneous Dirichlet boundary conditions, then 
\begin{equation}
    \norm{e^{tD_h}}_{\infty}\leq 1, \quad \forall~t>0.
\end{equation}
\end{lemma}
\begin{proof}
In this proof, we will apply the results,
\begin{equation}
    e^{tD_{h}} =\lim_{k\to \infty} \left(1+\frac{t\Delta_{h}}{k}\right)^k.
\end{equation}
First, we restrict $t\leq t_0:=\frac{1}{(n+1)^2}$. Then
when $k$ is larger enough, i.e. $k\geq 2t_0 (n+1)^2$, one has
\begin{equation}
\norm{1+\frac{tD_h}{k}}_{\infty}\leq  \left(1-2 \frac{t }{k}(n+1)^2\right)+ \frac{t }{k}(n+1)^2+ \frac{t }{k}(n+1)^2\leq 1. 
\end{equation}
Hence,
\begin{equation*}
    \norm{\left(1+\frac{tD_{h}}{k}\right)^k}_{\infty}\leq \norm{1+\frac{tD_{h}}{k}}_{\infty}^k\leq 1.
\end{equation*}
By continuity, one can get that $\norm{e^{tD_h}}_{\infty}\leq 1$. As for $t>t_0$, we choose positive integer $l$ such that $t/l\leq t_0$  and get 
\begin{equation*}
    \norm{e^{tD_h}}_{\infty}=\norm{\left(e^{tD_h/l}\right)^l}_{\infty}\leq \norm{e^{tD_h/l}}_{\infty}^l\leq 1.
\end{equation*}
In this way, we obtain the desired results.
\end{proof}

\begin{lemma}\label{lem:Delta_p_estimate}
Let $D_h^\text{per}$ be the discrete $1$-dimensional Laplacian operator with periodic boundary conditions, then 
\begin{equation}
    \norm{e^{tD_h^\text{per}}}_{\infty}= 1, \quad \forall~t>0.
\end{equation}
\end{lemma}
\begin{proof}
Using the same argument in \lem{discrete_max_principle}, we can prove that 
\begin{equation*}
    \norm{e^{tD_h^\text{per}}}_{\infty}\leq 1, \quad \forall~t>0.
\end{equation*}
Now we only need to prove the inequality in the opposite direction. Denote $v$ to be the eigenvector of $D_h^\text{per}$ associated with eigenvalue $0$. Then we know that $e^{tD_h^\text{per}} v=v$, which implies $\norm{e^{tD_h^\text{per}}}_{\infty}\geq 1$. In this way, we obtain that $\norm{e^{tD_h^\text{per}}}_{\infty}= 1$.
\end{proof}

\begin{lemma}\label{lem:Delta_h_estimate}
Let $D_h^\text{Dir}$ be the discrete $1$-dimensional Laplacian operator with homogeneous Dirichlet boundary conditions, and denote $\mu_1\geq \mu_2\geq \cdots\geq \mu_n$ to be the eigenvalues of $D_h^\text{Dir}$. Then 
\begin{equation}
    \norm{e^{tD_h^\text{Dir}}}_{\infty} \leq \left(\frac{4}{\pi}+\frac{2e^{\mu_2 t}}{e^{-2\mu_1 t/\pi}-1}\right)e^{\mu_1 t}, \quad \forall~t>0.
\end{equation}
\end{lemma}

\begin{proof}
We use the decomposition of $D_h^\text{Dir}$:
\[
D_h^\text{Dir} =U\text{diag}(\mu_1,...,\mu_n)U^\top \text{ with } U_{kl}=\sqrt{\frac{2}{n+1}}\sin\left(\frac{kl \pi}{n+1}\right).
\]
Here $\mu_k=-4(n+1)^2\left(\sin\left(\frac{k\pi}{2n+2}\right)\right)^2$. Besides, we know that for any $3\leq k\leq n$,
\begin{equation}
    \begin{split}
        \mu_k-\mu_2& = -4(n+1)^2\left(\sin\left(\frac{k\pi}{2n+2}\right)\right)^2+4(n+1)^2\left(\sin\left(\frac{\pi}{n+1}\right)\right)^2\\
        &= 2(n+1)^2\left( \cos\left(\frac{k\pi}{n+1}\right)-\cos\left(\frac{2\pi}{n+1}\right)\right)\\
        &=-4(n+1)^2\sin\left(\frac{(k-2)\pi}{2n+2}\right)\sin\left(\frac{(k+2)\pi}{2n+2}\right)\\
        &\leq -4(n+1)^2\sin\left(\frac{\pi}{2n+2}\right)\sin\left(\frac{(k+1)\pi}{2n+2}\right)\\
        &\leq \mu_1 \frac{\sin\left(\frac{(k+1)\pi}{2n+2}\right)}{\sin\left(\frac{\pi}{2n+2}\right)} \leq \mu_1 \frac{\frac{k+1}{n+1}}{\frac{\pi}{2n+2}}=\frac{2\mu_1}{\pi} (k+1).
    \end{split}
\end{equation}
In the last step, we use the inequality $\frac{2}{x}\leq \sin(x)\leq x$ for $x\in[0,\frac{\pi}{2}]$. 
Note that 
\begin{equation}
    \norm{e^{tD_h^\text{Dir}}}_{\infty} =\max_{1\leq k\leq n} \sum_{l=1}^n \abs{\sum_{j=1}^n U_{kj} e^{\mu_j t} U_{lj}}. 
\end{equation}
For any $1\leq k\leq n$, we have
\begin{equation}
    \begin{split}
        &\sum_{l=1}^n \abs{\sum_{j=1}^n U_{kj} e^{\mu_j t} U_{lj}}\\
        &= \frac{2}{n+1}\sum_{l=1}^n \abs{\sum_{j=1}^n e^{\mu_j t}\sin\left(\frac{kj \pi}{n+1}\right)\sin\left(\frac{lj \pi}{n+1}\right) }\\
        &\leq \frac{2}{n+1}\sum_{l=1}^n \sum_{j=1}^n e^{\mu_j t}\abs{\sin\left(\frac{kj \pi}{n+1}\right)\sin\left(\frac{lj \pi}{n+1}\right) }\\
        &\leq \frac{2n}{n+1}\sum_{j=2}^n e^{\mu_j t} + \frac{2}{n+1} e^{\mu_1 t}\sin\left(\frac{k \pi}{n+1}\right)\sum_{l=1}^n\sin\left(\frac{l \pi}{n+1}\right). 
    \end{split}
\end{equation}
We notice that for any $k$,
\begin{equation}
\begin{split}
    \frac{2}{n+1} e^{\mu_1 t}\sin\left(\frac{k \pi}{n+1}\right)\sum_{l=1}^n\sin\left(\frac{l \pi}{n+1}\right) &\leq \frac{2}{n+1}   e^{\mu_1 t}\sin\left(\frac{k \pi}{n+1}\right)\frac{1}{\tan\left(\frac{ \pi}{2n+2}\right)}\\
    &\leq \frac{4}{\pi}e^{\mu_1 t}\sin\left(\frac{k \pi}{n+1}\right) \leq \frac{4}{\pi}e^{\mu_1 t},
\end{split}
\end{equation}
where we use the equality $\sum_{l=1}^n\sin\left(\frac{l \pi}{n+1}\right)=\frac{1}{\tan\left(\frac{ \pi}{2n+2}\right)}$ and $\tan(x)\geq x$ for $x\in[0,\frac{\pi}{2})$.

We also notice that 
\begin{equation}
\begin{split}
    \frac{2n}{n+1}\sum_{j=2}^n e^{\mu_j t}&= \frac{2n}{n+1}e^{\mu_2 t}\sum_{j=2}^n e^{(\mu_j-\mu_2) t} \leq \frac{2n}{n+1}e^{\mu_1 t}\sum_{j=2}^n e^{2(k+1)\mu_1 t/\pi}\\
    &\leq 2e^{\mu_2 t} \frac{e^{6\mu_1 t/\pi}\left(1-e^{2(n-1)\mu_1 t/\pi}\right)}{1-e^{2\mu_1 t/\pi}}\leq\frac{2}{e^{-2\mu_1 t/\pi}-1}e^{(\mu_1+\mu_2)t}.
\end{split}
\end{equation}
In this way, we obtain 
\begin{equation}
    \norm{e^{t D_h^\text{Dir}}}_{\infty} \leq \left(\frac{4}{\pi}+\frac{2e^{\mu_2 t}}{e^{-2\mu_1 t/\pi}-1}\right)e^{\mu_1 t}.
\end{equation}
\end{proof}
 Combining \lem{discrete_max_principle} and \lem{Delta_h_estimate}, we obtain the following results.
\begin{lemma}
    Let $\mu_1$ be the largest eigenvalue of $1$-dimensional $D_h^\text{Dir}$ with  homogeneous Dirichlet boundary conditions. Given $t\in\R^+$, for arbitrary $\mu>\mu_1$, we have the decay estimate
    \begin{equation}
        \norm{e^{tD_h^\text{Dir}}}_{\infty} \leq \begin{cases}
        1,&  0 < t < \frac{\ln(3)}{2(\mu-\mu_1)}, \\
        e^{t\mu} , &  t \geq \frac{\ln(3)}{2(\mu-\mu_1)}. 
        \end{cases}
    \end{equation}
\end{lemma}
\begin{proof}
    The case when $t < \frac{\ln(3)}{2(\mu-\mu_1)}$ directly follows from \lem{discrete_max_principle}. 
    When $t \geq \frac{\ln(3)}{2(\mu-\mu_1)}$, recall that
    \begin{equation}
        \mu_2 =-4(n+1)^2\left(\sin\left(\frac{2\pi}{2n+2}\right)\right)^2 =4\mu_1 \left(\cos\left(\frac{\pi}{2n+2}\right)\right)^2\leq 3\mu_1.
    \end{equation}
    Then we have 
    \begin{equation}
        \frac{4}{\pi}+\frac{2e^{\mu_2 t}}{e^{-2\mu_1 t/\pi}-1}\leq \frac{4}{\pi}+\frac{2e^{3\mu_1 t}}{e^{-2\mu_1 t/\pi}-1}\leq \sqrt{3}.
    \end{equation}
    Notice that 
    \begin{equation}
        \sqrt{3}e^{t\mu_1} \le e^{t(\mu-\mu_1)}e^{t\mu_1} = e^{t\mu}, \quad t \geq \frac{\ln(3)}{2(\mu-\mu_1)}
    \end{equation}
    for an $\mu>\mu_1$, we get $\norm{e^{tD_h^\text{Dir}}}_{\infty}\leq e^{t\mu}$.
    Together with \lem{exp_decay}, we obtain the desired bound. 
\end{proof}

We introduce a lemma about the decay estimate of the discrete heat semigroup.

\begin{lemma}[Decay of Discrete Heat Semigroup]\label{lem:exp_decay} 
Let $\Delta_h$ be the discrete $d$-dimensional Laplacian operator defined in \eqref{eqn:dis_lap_highD_tensor} with homogeneous Dirichlet boundary conditions imposed on $d_1$ dimensions and periodic boundary conditions imposed on $d_2$ dimensions ($d_1 + d_2 = d$). Given $t\in\R^+$, for arbitrary $\mu>\mu_1$, we have the decay estimate
\begin{equation}
        \|e^{t D \Delta_h}\|_{\infty} \leq \begin{cases}
        1,&  0 < t< \frac{\ln(3)}{2D(\mu-\mu_1)}, \\
        e^{t D d_1 \mu} , &  t \geq  \frac{\ln(3)}{2D(\mu-\mu_1)},
        \end{cases}
    \end{equation}
where $\mu_1=-4(n+1)^2\sin^2\left(\frac{\pi}{2n+2}\right)$.
\end{lemma}
\begin{proof}
Thanks to \eqref{eqn:dis_lap_highD_tensor}, one has
\begin{equation} \label{eqn:exp_lap_norm_highD_tensor}
    \norm{e^{tD\Delta_h}}_{\infty} 
    = \norm{e^{tD D_h} \otimes e^{tD D_h} \otimes \cdots \otimes e^{tD D_h}}_{\infty} 
    \le \norm{e^{tD D_h^\text{Dir}}}_{\infty}^{d_1}\norm{e^{tD D_h^\text{per}}}_{\infty}^{d_2}.
\end{equation}
It suffices to consider the heat kernel for each dimension. For the dimensions with homogeneous Dirichlet boundary conditions, we have
\begin{equation}\label{eqn:1d_decay_heat_kernel_dirichlet} 
    \norm{e^{t D D_h }}_\infty \leq \begin{cases}
        1,&  0 < t <  \frac{\ln(3)}{2D(\mu-\mu_1)}, \\
        e^{tD\mu} , &  t \geq  \frac{\ln(3)}{2D(\mu-\mu_1)} . 
        \end{cases}
\end{equation}
Similarly, for the dimensions with periodic boundary conditions, we have
\begin{equation}\label{eqn:1d_decay_heat_kernel_periodic}
\norm{e^{t D D_h^\text{per} }}_\infty \le 1, \quad \forall~t>0.
\end{equation}
Combining these, we obtain 
\begin{equation}
    \norm{e^{tD\Delta_h}}_{\infty}\leq \norm{e^{tD D_h}}_{\infty}^{d_1}\leq \begin{cases}
        1,&  0 < t< \frac{\ln(3)}{2D(\mu-\mu_1)} , \\
        e^{tDd_1\mu} , &  t \geq \frac{\ln(3)}{2D(\mu-\mu_1)}. 
        \end{cases}
\end{equation}
\end{proof}

\begin{lemma}\label{lem:exp_decay_RD}
Let $\Delta_h$ be the discrete $d$-dimensional Laplacian operator defined in \eqref{eqn:dis_lap_highD_tensor} with homogeneous Dirichlet boundary conditions imposed on $d_1$ dimensions and periodic boundary conditions imposed on $d_2$ dimensions ($d_1 + d_2 = d$). Given $j\in\Z^+$, $t\in\R^+$, and $a\in\R$ such that $\lambda_1 = D d_1 \mu_1 + a < 0$, for arbitrary $\lambda$ such that $\lambda_1<\lambda < 0$, we have the decay estimate
\begin{equation}
\begin{aligned}
    \int_0^t \|e^{j(t-s) (D \Delta_h + a)}\|_{\infty}  \,d s \le \begin{cases}
        \frac{1}{ja} (e^{\frac{\ln(3)d_1}{2(\lambda-\lambda_1)}a}-1) + \frac{1}{j|\lambda|}, & a \neq 0, \\
        \frac{\ln(3)d_1}{2j(\lambda-\lambda_1)} + \frac{1}{j|\lambda|}, &  a = 0,
        \end{cases}
\end{aligned}
\end{equation}
where $\mu_1=-4(n+1)^2\sin^2\left(\frac{\pi}{2n+2}\right)$.
\end{lemma}

\begin{proof}
Given any $\lambda$, define $\mu:=\frac{\lambda-a}{D d_1}$. Then $\lambda> \lambda_1$ implies that $\mu>\mu_1$.
According to \lem{exp_decay},  we have
\begin{equation}
       \|e^{jt (D \Delta_h + a)}\|_{\infty} \leq 
        \begin{cases}
        e^{jta}  , &  0 < t< \frac{\ln(3)}{2jD(\mu-\mu_1)}, \\
        e^{jt \lambda} = e^{jt (D d_1 \mu + a)} , &  t \geq  \frac{\ln(3)}{2jD(\mu-\mu_1)}. 
        \end{cases}
\end{equation}
We now integrate the upper bound of $\|e^{jt (D \Delta_h + a)}\|_{\infty}$. When $a>0$ or $a<0$, we have such an estimate
\begin{equation}
        \int_0^t \|e^{j(t-s) (D \Delta_h + a)}\|_{\infty}  \,d s \leq 
        \begin{cases}
        \frac{1}{ja} (e^{jta}-1) , &  0 < t< \frac{\ln(3)}{2jD(\mu-\mu_1)}, \\
        \frac{1}{ja} (e^{\frac{\ln(3)}{2D(\mu-\mu_1)}a}-1) + \frac{1}{j|\lambda|} e^{\frac{\ln(3)}{2D(\mu-\mu_1)} \lambda}, &  t \geq  \frac{\ln(3)}{2jD(\mu-\mu_1)}.
        \end{cases}
\end{equation}
When $a=0$, the integration is reduced to
\begin{equation}
        \int_0^t \|e^{j(t-s) (D \Delta_h + a)}\|_{\infty}  \,d s \leq 
        \begin{cases}
        t, &  0 < t< \frac{\ln(3)}{2jD(\mu-\mu_1)}, \\
        \frac{\ln(3)}{2jD(\mu-\mu_1)} + \frac{1}{j|\lambda|} e^{\frac{\ln(3)}{2D(\mu-\mu_1)} \lambda}, &  t \geq  \frac{\ln(3)}{2jD(\mu-\mu_1)}.
        \end{cases}
\end{equation}
Since $\int_0^t \|e^{j(t-s) (D \Delta_h + a)}\|_{\infty}$ is non-decreasing, we consider the upper bound as $t \geq  \frac{\ln(3)}{2jD(\mu-\mu_1)}$ to have
\begin{equation}
\begin{aligned}
    \int_0^t \|e^{j(t-s) (D \Delta_h + a)}\|_{\infty}  \,d s \le \begin{cases}
        \frac{1}{ja} (e^{\frac{\ln(3)}{2D(\mu-\mu_1)}a}-1) + \frac{1}{j|\lambda|} e^{\frac{\ln(3)}{2D(\mu-\mu_1)} \lambda}, & a \neq 0, \\
        \frac{\ln(3)}{2jD(\mu-\mu_1)} + \frac{1}{j|\lambda|} e^{\frac{\ln(3)}{2D(\mu-\mu_1)} \lambda}, &  a = 0.
        \end{cases}
\end{aligned}
\end{equation}
Then, noticing $e^{\frac{\ln(3)}{2D(\mu-\mu_1)} \lambda} < 1$, 
\begin{equation}
\begin{aligned}
    \int_0^t \|e^{j(t-s) (D \Delta_h + a)}\|_{\infty}  \,d s \le \begin{cases}
        \frac{1}{ja} (e^{\frac{\ln(3)}{2D(\mu-\mu_1)}a}-1) + \frac{1}{j|\lambda|}, & a \neq 0, \\
        \frac{\ln(3)}{2jD(\mu-\mu_1)} + \frac{1}{j|\lambda|}, &  a = 0.
        \end{cases}
\end{aligned}
\end{equation}
Exploiting the definition of $\mu$, we get 
\begin{equation}
\begin{aligned}
    \int_0^t \|e^{j(t-s) (D \Delta_h + a)}\|_{\infty}  \,d s \le \begin{cases}
        \frac{1}{ja} (e^{\frac{\ln(3)d_1}{2(\lambda-\lambda_1)}a}-1) + \frac{1}{j|\lambda|}, & a \neq 0, \\
        \frac{\ln(3)d_1}{2j(\lambda-\lambda_1)} + \frac{1}{j|\lambda|}, &  a = 0. 
        \end{cases}
\end{aligned}
\end{equation}
\end{proof}

\section{A naive estimate of $\norm{U(t)}$}
\label{sec:l2_estimate}
In this section, we present two estimates of $\norm{U(t)}$. Note that this estimate is by no means sharp, a tighter bound would depend on the specific forms of the nonlinearity and initial data.

\begin{lemma} \label{lem:l2_estimate}
Let $U(t)$ be a solution to \eqref{eq:NODE} that satisfies the assumptions of Lemma \ref{lem:apriori}, i.e., the initial condition $\norm{U(0)}_{\infty}\leq \gamma$. Then one has
\begin{equation}
    \norm{U(t)} \leq e^{\left(\lambda_1+\abs{b} \gamma^{M-1}\right) t}\norm{U(0)},
\end{equation}
where $\lambda_1$ is the largest eigenvalue of $D\Delta_h+a I$.
\end{lemma}
\begin{proof}
Consider the derivative of $\norm{U(t)}$,
\begin{equation}
    \begin{split}
        \frac{\rd \norm{U(t)}^2}{\rd t}&=2U^\dagger (D\Delta_h + a I ) U+ bU^\dagger U^{.M}+ b \left(U^{.M}\right)^\dagger U\\
        &\leq 2\lambda_1\norm{U}^2 +2\abs{b} \sum_{j\in \range{n_d}} U_j(t)^{M+1}\\
        &\leq  2\lambda_1\norm{U}^2 +2\abs{b} \norm{U(t)}_{\infty}^{M-1} \norm{U(t)}^2_2\\
        &\leq  2\lambda_1\norm{U}^2 +2\abs{b} \gamma^{M-1} \norm{U(t)}^2.
    \end{split}
\end{equation}
In the last inequality, we use the maximum principle described in  Lemma \ref{lem:apriori}, that is, $\norm{U(t)}_{\infty}\leq \gamma$ for any $t>0$. Hence, we obtain 
\begin{equation*}
    \norm{U(t)}^2 \leq e^{2\left(\lambda_1+\abs{b} \gamma^{M-1}\right) t}\norm{U(0)}^2.
\end{equation*}
\end{proof}

\begin{remark}
Given $U(t)$ satisfying the assumption of \ref{lem:l2_estimate}, we have the following estimate
\begin{equation}
    \norm{U(t)} \leq \sqrt{n^d} \norm{U(t)}_{\infty}\leq \sqrt{n^d} \gamma,
\end{equation}
which follows the maximal principle.
\end{remark}

\begin{remark}
When $R_D<1$, we notice that $C(\lambda)\geq 1$. So $\frac{|b|}{|\lambda_1|}\gamma^{M-1} \leq R_D<1$, which implies that $\lambda_1 +|b|\gamma^{M-1}<0$. Therefore, we know that the 2-norm of the solution decays.
\end{remark}

\begin{lemma} \label{lem:l2_estimate2}
Let $U(t)$ be a solution to \eqref{eq:NODE}. Suppose that the largest eigenvalue of $D\Delta_h+a I$, $\lambda_1$ is negative. If $\abs{b}\norm{U(0)}_2^{M-1}+\lambda_1\geq 0$, then one has
\begin{equation}
    \norm{U(t)} \leq \norm{U(0)},\quad \forall~t\geq 0.
\end{equation}
\end{lemma}
\begin{proof}
Similarly, we consider the derivative of $\norm{U(t)}$,
\begin{equation}
    \begin{split}
        2 \norm{U(t)} \frac{\rd \norm{U(t)}}{\rd t}&=\frac{\rd \norm{U(t)}^2}{\rd t}\\
        &=2U^\dagger (D\Delta_h + a I ) U+ bU^\dagger U^{.M}+ b \left(U^{.M}\right)^\dagger U\\
        &\leq 2\lambda_1\norm{U}^2 +2\abs{b} \norm{U(t)}_{M+1}^{M+1}\\
        &\leq  2\lambda_1\norm{U}^2 +2\abs{b} \norm{U(t)}^{M+1}.
    \end{split}
\end{equation}
In the last inequality, we use the inequality $\norm{u}_{M+1}\leq \norm{u}_2$ for any $M\geq 1$. For simplicity, we denote $\norm{U(t)}_2$ as $y$ and now we study the following equation instead,
\begin{equation}
    \frac{\rd y}{\rd t}\leq \lambda_1 y + \abs{b} y^M.
\end{equation}
This is a Bernoulli differential equation  and we only consider its positive solution. Define $z:=\frac{1}{(M-1)y^{M-1}}$ and it fulfils
\begin{equation}
\begin{split}
    \frac{\rd z}{\rd t}&=\frac{\rd }{\rd t}\left(\frac{1}{(M-1)y^{M-1}}\right)=-\frac{1}{y^M} \frac{\rd y}{\rd t} \\
    & \geq -\frac{\lambda_1}{y^{M-1}} -\abs{b} = -\lambda_1 (M-1) z-\abs{b}.
\end{split}
\end{equation}
By using the integrating factor, one gets
\begin{equation}\label{eq:z_eq}
    e^{\lambda_1(M-1) t} z \geq z(0)+\frac{\abs{b}}{\lambda_1 (M-1)} -\frac{\abs{b}}{\lambda_1 (M-1)} e^{\lambda_1(M-1) t}.
\end{equation}
If $\abs{b}\norm{U(0)}^{M-1}+\lambda_1>0$, then
\begin{equation}
    \frac{1}{(M-1)y^{M-1}} = z\geq \left(z(0)+\frac{\abs{b}}{\lambda_1 (M-1)} \right)e^{\lambda_1(M-1) t} - \frac{\abs{b}}{\lambda_1 (M-1)}.
\end{equation}
In terms of $\norm{U(t)}$, one has
\begin{equation}
    \norm{U(t)} =y \leq \left[ \left(\frac{1}{\norm{U(0)}_2^{M-1}} +\frac{\abs{b}}{\lambda_1}\right)e^{\lambda_1(M-1) t}-\frac{\abs{b}}{\lambda_1}\right]^{-\frac{1}{M-1}}.
\end{equation}
The coefficient $\frac{1}{\norm{U(0)}^{M-1}} +\frac{\abs{b}}{\lambda_1}$ is negative and by the monotonic decreasing of $e^{\lambda_1 (M-1)t}$, one obtains the desired results.

If $\abs{b}\norm{U(0)}^{M-1}+\lambda_1=0$, \eqref{eq:z_eq} implies $z \geq -\frac{\abs{b}}{\lambda_1 (M-1)}$. Hence $\norm{U(t)}\leq \left(\frac{-\lambda_1}{\abs{b}}\right)^{\frac{1}{M-1}}=\norm{U(0)}$, which completes the proof.
\end{proof}

\section{Proof of \texorpdfstring{\thm{Carleman_0}}{Theorem \ref{thm:Carleman_0}}}\label{sec:proof_Carleman_0}

\begin{proof}
The truncation error $\eta_j$ satisfies the equation
\begin{equation}
\label{eqn:eta_0_A}
\eta'_j(t) = A^j_j \eta_j(t) + A^j_{j+M-1}\left( y_{j+M-1}(t)-\hat{y}_{j+M-1}(t)\delta_{j+M-1 \le N} \right), \qquad 1 \le j \le N.
\end{equation}
Applying the variation of constant formula to \eqref{eqn:eta_0_A}, one has
\begin{align}
\eta_j(t) = \int_0^t e^{A^j_j (t-s)} A^j_{j+M-1} y_{j+M-1}(s) \,d s, \quad  N-M+2 \le j \le N.
\end{align}
Note that it follows from \lem{apriori} that
\[
\norm{y_{j+M-1}(s)} = \|U^{\otimes(j+M-1)}(s)\|
\le \|U(s)\|^{j+M-1}
\le \max_t\|U(t)\|^{j+M-1}.
\]
Therefore, according to \lem{ineq_Ajk}, we have for $N-M+2 \le j \le N$,
\begin{equation}
\begin{aligned}
    \norm{\eta_j(t)}
    & \le  \int_0^t \norm{e^{A^j_j (t-s)}} \|A^j_{j+M-1}\| \|y_{j+M-1}(s)\| \,d s  \\
    & \le  \int_0^t e^{j \lambda_1 (t-s)} j|b| \norm{y_{j+M-1}(s)}  \,d s  \\
    & \le  \max_t\|U(t)\|^{j+M-1}\int_0^t e^{j \lambda_1 (t-s)} j|b| \,d s \\
    &= \max_t\|U(t)\|^{j+M-1}\frac{|b|}{|\lambda_1|}\left(1 - e^{-j|\lambda_1| t} \right) \\
    &\le \max_t\|U(t)\|^{j+M-1}\frac{|b|}{|\lambda_1|}.
\end{aligned}
\end{equation}
For $N-2M+3 \le j \le N-M+1$,
\begin{equation}
\begin{aligned}
    \norm{\eta_j(t)} & \le  \int_0^t e^{j \lambda_1 (t-s)} j|b| \norm{\eta_{j+M-1}(s)} \,d s  \\
    & \le \max_t\|U(t)\|^{j+2M-2}\frac{|b|}{|\lambda_1|} \int_0^t e^{j \lambda_1 (t-s)} j|b|  \,d s \\
    & \le \max_t\|U(t)\|^{j+2M-2}(\frac{|b|}{|\lambda_1|})^2 \left(1 - e^{-j|\lambda_1| t} \right) \\
    & \le \max_t\|U(t)\|^{j+2M-2}(\frac{|b|}{|\lambda_1|})^2.
\end{aligned}
\end{equation}
One can continue by mathematical induction for every group of $M-1$ terms and arrive at
\begin{equation}
\begin{aligned}
    \norm{\eta_j(t)} & \le  \int_0^t e^{j \lambda_1 (t-s)} j|b| \norm{\eta_{j+M-1}(s)} \,d s \\
  & \le \max_t\|U(t)\|^{j+(M-1)\lceil \frac{N-j+1}{M-1} \rceil}(\frac{|b|}{|\lambda_1|})^{\lceil \frac{N-j+1}{M-1} \rceil} \left(1 - e^{j\lambda_1 t} \right) \\
  & \le \max_t\|U(t)\|^{j}\overline{R}^{\lceil \frac{N-j+1}{M-1} \rceil} \left(1 - e^{j\lambda_1 t} \right) \\
  &\le \max_t\|U(t)\|^{j}\overline{R}^{\lceil \frac{N-j+1}{M-1} \rceil},
\end{aligned}
\end{equation}
where we use $\overline{R}=\frac{\|F_M\|}{|\lambda_1|}\max_t\|U(t)\|^{j+M-1}$ as in \eq{A}.

In practice, we set $N-1$ as some integer multiples of $M-1$, such that 
\begin{equation}
\norm{\eta_1}_{\infty} \le \max_t\|U(t)\| \overline{R}^{ \frac{N}{M-1} } \left(1 - e^{\lambda_1 t} \right). 
\end{equation}
Finally, if $R\leq1$, according to \lem{l2_estimate2}, we have $\|U_{\mathrm{in}}\| < \|U(t)\|$ for $t>0$, and then $R = \overline{R}$. This completes the proof of the desired result.
\end{proof}
\section{Estimate of the preconstant}\label{sec:preconstant}

According to the definition of $R_D$, the value of $R_D$ depends on the choice of $\lambda$. In order to obtain a sharper estimate of the approximation error, we hope to find the optimal value of $\lambda$ such that it minimizes $C(\lambda)$. When $a=0$, this optimization problem is easy to solve. For any $\lambda_1<\lambda<0$, one has 
\begin{equation}
    \frac{\ln(3)d_1}{2(\lambda-\lambda_1)} |\lambda_1| + \frac{|\lambda_1|}{|\lambda|}=\left(\frac{\ln(3)d_1}{2(\lambda-\lambda_1)} + \frac{1}{-\lambda}\right)\left(\lambda-\lambda_1-\lambda\right)\geq \left(\sqrt{\frac{\ln(3)}{2}d_1}+1\right)^2,
\end{equation}
and the equality holds when $\lambda=\frac{\lambda_1}{\sqrt{\frac{\ln(3)}{2}d_1}+1}$. The minimum of $C(\lambda)$ is $\left(\sqrt{\frac{\ln(3)}{2}d_1}+1\right)^2$, which is $O(d)$,  and the corresponding value of $R_D$ is
\begin{equation}
    \frac{|b|}{|\lambda_1|}\gamma^{M-1}\left(\sqrt{\frac{\ln(3)}{2}d_1}+1\right)^2.
\end{equation}

When $a\ne 0$, the optimal value of $\lambda$ can be obtained by solving the following equation
\begin{equation}
    \sqrt{\frac{\ln(3)}{2} d_1} e^{\frac{\ln(3)d_1}{4(\lambda-\lambda_1)}a}\frac{1}{\lambda-\lambda_1}+ \frac{1}{\lambda}=0.
\end{equation}
In real applications, we suggest tuning the parameter $\frac{\lambda}{\lambda_1}$ to obtain a value of $R_D$ around the optimum, instead of directly solving the above equation. Besides, we have the following theoretical results regarding $\min_{\lambda_1<\lambda<0} C(\lambda)$.

\begin{lemma}\label{lem:upperbound}
    Suppose that $\lambda^*:=-\pi^2Dd_1+a<0$ and $a\ne 0$, there exists an upper bound of 
$\min_{\lambda_1<\lambda<0} C(\lambda)$, where $\lambda_1=Dd_1 \mu_1 +a$ and $\mu_1=-4(n+1)^2\sin^2\left(\frac{\pi}{2n+2}\right)$. The upper bound is independent of $n$.
\end{lemma}
\begin{proof}
    The inequality $\sin(x)<x$ implies that $-\pi^2<\mu_1$ for any $n$. Besides, $\lim_{n\to \infty} \mu_1=-\pi^2$. Then there exists an positive integer $n^*$ such that for any $n\geq n^*$, $\mu_1 \leq -0.9\pi^2 -\frac{a}{Dd_1}$, which leads to $\lambda^*\leq\lambda_1\leq0.9 \lambda^* $. 
    
    When $n\geq n^*$, one has
    \begin{equation}
    \begin{split}
        \min_{\lambda_1<\lambda<0} C(\lambda)&\leq C\left(\frac{\lambda_1}{2}\right)=\frac{|\lambda_1|}{a} (e^{-\frac{\ln(3) d_1}{\lambda_1}a }-1) + 2\\
        &\leq \max_{\lambda^*\leq\lambda_1\leq0.9 \lambda^*} \frac{|\lambda_1|}{a} (e^{-\frac{\ln(3) d_1}{\lambda_1}a }-1) + 2\\
        &\leq \frac{|\lambda^*|}{a} \left(e^{\frac{\ln(3) d_1}{0.9|\lambda^*|}a }-1\right) + 2=:C_1. 
    \end{split}
    \end{equation}
    Note that $C\left(\frac{\lambda_1}{2}\right)$ is continuous over $[\lambda^*, 0.9\lambda^*]$, one gets that $C_1$ is finite. Let $C_2$ be the maximum of $C\left(\frac{\lambda_1}{2}\right)$ for $1\leq n\leq n^*$. Then $\max(C_1, C_2)$ is the upper bound we desired.

\end{proof}

\begin{remark}
From the proof of \lem{upperbound}, we know that when $a\ne 0$ and $n>n^*$, $C_1$ is an upper bound of $\min_{\lambda_1<\lambda<0} C(\lambda)$, where
\begin{equation}
    C_1=\frac{|\lambda^*|}{a} \left(e^{\frac{\ln(3) d_1}{0.9|\lambda^*|}a }-1\right) + 2=\frac{\pi^2Dd_1-a}{a} \left(e^{\frac{\ln(3) d_1}{0.9(\pi^2Dd_1-a)}a }-1\right) + 2= \mathcal{O}(d).
\end{equation}
Here, we notice that $\frac{\ln(3) d_1}{0.9(\pi^2Dd_1-a)}a$ turns out to be $o(1)$ for large $d$. Combining the discussion when $a=0$, we know that for $n$ large enough, $\min_{\lambda_1<\lambda<0} C(\lambda)$ is $\mathcal{O}(d)$ regardless of $a$.
\end{remark}

\section{An illustration on approximating derivatives using discrete Fourier transform}\label{sec:append_DFT}

Here we briefly explain the reason why discrete Fourier transform can be applied to compute derivatives of a function in classical computing. 
For simplicity, we use a $1$-dimensional example. 
Let $f(x)$ denote a smooth function defined on the interval $[0,1]$, and our goal is to transform the vector 
$$\vec{f} = \sum_{k=0}^{n-1} f(k/n)\ket{k}$$
to the vector
$$\vec{f'} = \sum_{k=0}^{n-1} f'(k/n)\ket{k}.$$

Let $\mathcal{F}$ denote the one-dimensional quantum Fourier transform operator, then the discrete Fourier transform $\mathcal{F}^{-1}$ acts on a vector $v = (v_0,\cdots, v_{n-1})$ and maps it to another vector according to the formula 
\begin{equation}
    (\mathcal{F}^{-1} v)_l = \frac{1}{\sqrt{n}}\sum_{k=0}^{n-1} v_k \omega_n^{-kl}
\end{equation}
where $\omega_n = e^{2\pi i/n}$. 
$\mathcal{F}^{-1} v$ can be interpreted as the set of discrete Fourier coefficients of the vector $v$. 
To see this, let the function $f$ allow the following complex Fourier series expansion 
\begin{equation}\label{eqn:app_F_series_1_append}
    f(x) \approx \sum_{m=-\theta}^{\theta} c_m e^{2\pi i mx} 
\end{equation}
for a positive integer $\theta$. 
The error of this approximation is exponentially small in terms of $\theta$ for any smooth function $f$. 
Then the Fourier transform of the vector $\vec{f} = (f(k/n))_0^{n-1}$ becomes
\begin{equation}
    \begin{split}
        (\mathcal{F}^{-1} \vec{f})_l & = \frac{1}{\sqrt{n}}\sum_{k=0}^{n-1} f(k/n) \omega_n^{-kl} \\
        & \approx \frac{1}{\sqrt{n}}\sum_{k=0}^{n-1} \sum_{m=-\theta}^{\theta} c_m e^{2\pi i mk/n} \omega_n^{-kl} \\
        & = \frac{1}{\sqrt{n}} \sum_{m=-\theta}^{\theta} c_m \left(\sum_{k=0}^{n-1}\omega_n^{(m-l)k}\right). 
    \end{split}
\end{equation}
Noticing that 
\begin{equation}
    \left(\sum_{k=0}^{n-1}\omega_n^{(m-l)k}\right) = \begin{cases}
    n, & \text{if }m-l = jn \text{ for some integer } j, \\ 
    0, & \text{else}, 
    \end{cases}
\end{equation}
we have 
\begin{equation}\label{eqn:app_fourier_coeff_append}
    (\mathcal{F}^{-1} \vec{f})_l \approx \begin{cases}
    \sqrt{n} c_l , & \text{if }0 \leq l \leq \theta, \\
    \sqrt{n} c_{l-n}, & \text{if } n-\theta \leq l \leq n-1, \\
    0, & \text{ else.}
    \end{cases}
\end{equation}
This implies that, up to a normalization factor, each non-zero entry of the transformed vector matches one of the Fourier coefficients. 

Fourier transform of the derivative can be computed similarly. 
Starting from the Fourier series again, we have 
\begin{equation}
    f'(x) \approx \sum_{m=-\theta}^{\theta} 2\pi i m c_m e^{2\pi i mx}. 
\end{equation}
Replacing the coefficient $c_m$ by $2\pi i m c_m$ in \eqref{eqn:app_fourier_coeff_append}, we have 
\begin{equation}\label{eqn:app_fourier_coeff_deriv_append}
    (\mathcal{F}^{-1} \vec{f'})_l \approx \begin{cases}
     2\pi i \sqrt{n} l c_l , & \text{if }0 \leq l \leq \theta, \\
     2\pi i \sqrt{n} (l-n) c_{l-n}, & \text{if } n-\theta \leq l \leq n-1, \\
     0 , & \text{else.}
    \end{cases}
\end{equation}
This is just a multiplication of the diagonal matrix $\hat{D}$ on the vector $\mathcal{F}^{-1} \vec{f}$, where 
\begin{equation}
    \hat{D} = 2\pi i \text{~diag}(0,1,\cdots,\theta,0,\cdots,0,-\theta,-\theta+1,\cdots,-1). 
\end{equation}
Therefore, 
\begin{equation}
    \vec{f'} \approx  \mathcal{F} \hat{D} \mathcal{F}^{-1} \vec{f}, 
\end{equation}
which implies that the derivative operator $\vec{f'}$ can be numerically computed by first performing an inverse quantum Fourier transform, then multiplying from the left by a diagonal matrix $\hat{D}$, and finally performing a quantum Fourier transform.

\section{Proofs of the results in \sec{application}}\label{sec:append_proof_app}

In this section we present technical details on proving the results presented in \sec{application}. 

\subsection{Proof of \thm{app_ratio}}
\begin{proof}
    The error in estimating the ratio can be decomposed into two parts, the error due to the approximate quantum state, and the error within the amplitude estimate. 
    The first part of the error can be directly bounded such that 
    \begin{equation}\label{eq:app_ratio_error_2}
    \begin{split}
        & \quad \left|\frac{\sum_{kT/m\in\D_t} \sum_{(l_1/n,\cdots,l_d/n)\in \D_x} |f(kT/m,l_1/n,\cdots,l_d/n)|^2 }{\sum_{k=0}^{m-1} \sum_{l_1=0}^{n-1} \cdots \sum_{l_d=0}^{n-1} |f(kT/m,l_1/n,\cdots,l_d/n)|^2 } -\langle \hat{f}|P|\hat{f} \rangle \right| \\
        & = \left| \langle f|P|f \rangle - \langle \hat{f}|P|\hat{f} \rangle\right| \\
        & \leq 2 \|\ket{f}-\ket{\hat{f}}\| \\
        & \leq 2\epsilon'
    \end{split}
    \end{equation}
    given that $\ket{\hat{f}}$ is produced by $\mathcal{A}(\epsilon')$ with tolerated error $\epsilon'$, which will be determined later. 
    Let $E$ be the quantity obtained by amplitude estimate algorithm. 
    Then, according to~\cite{BHM02}, by querying $\Or(q)$ times to $\mathcal{A}$ and $(I-2P)$, we can bound the error with probability at least $8/\pi^2$ as 
    \begin{equation}\label{eq:app_ratio_error_3}
        |\langle \hat{f}|P|\hat{f} \rangle - E| \leq 2\pi \frac{\sqrt{\langle \hat{f}|P|\hat{f} \rangle (1-\langle \hat{f}|P|\hat{f} \rangle)}}{q} + \frac{\pi^2}{q^2} \leq \frac{2\pi}{q}. 
    \end{equation}
    By the powering lemma~\cite{jerrum1986random}, we can boost the success probability to at least $(1-\delta)$ by repeating the procedure $\Or(\log(1/\delta))$ times and taking the median, leading to total $\Or(q\log(1/\delta))$ queries to $\mathcal{A}$ and $(I-2P)$. 
    Combining \eq{app_ratio_error_2} and \eq{app_ratio_error_3}, we have, with probability at least $(1-\delta)$, 
    \begin{equation}
    \begin{split}
          \left| \frac{\sum_{kT/m\in\D_t} \sum_{(l_1/n,\cdots,l_d/n)\in \D_x} |f(kT/m,l_1/n,\cdots,l_d/n)|^2 }{\sum_{k=0}^{m-1} \sum_{l_1=0}^{n-1} \cdots \sum_{l_d=0}^{n-1} |f(kT/m,l_1/n,\cdots,l_d/n)|^2 }  - E \right| \leq  2\epsilon' + \frac{2\pi}{q}. 
    \end{split}
    \end{equation}
    The proof then can be completed by choosing $\epsilon' = \epsilon/4$ and $q = 4\pi/\epsilon$. 
\end{proof}

\subsection{Proof of \thm{app_deriv_state_smooth}}

The proof of \thm{app_deriv_state_smooth} can be decomposed into two main steps. 
The first step is to bound the classical error of using discrete Fourier transform to compute the derivatives, which is given in \lem{application_deriv_fourier_error}. 
Then, we can apply this error bound to estimate the overall complexity of constructing the desired quantum state within the error $\epsilon$. 
Here we present and prove a more general result in \thm{application_deriv_state_main}, which can be viewed as a generalization of \thm{app_deriv_state_smooth} with a weaker regularity assumption.

\begin{lemma}\label{lem:application_deriv_fourier_error}
    Let $f(x)$ be a $C^p$ function with $p \geq 3$ defined on $[0,1]^d$, $$\vec{f} = \sum_{l_1=0}^{n-1}\cdots \sum_{l_d=0}^{n-1} f(l_1/n_1,\cdots,l_d/n_d) \ket{l_1}\cdots \ket{l_d}$$ be the possibly unnormalized vector encoding $f(x)$ evaluated at discrete grids, and $\mathcal{F}$ denote the one-dimensional quantum Fourier transform with $n$ nodes.
    Furthermore, for any positive integer $\theta \leq n/2$, let $D_{j,\theta}$ be a diagonal matrix 
    \begin{equation}
    D_{j,\theta} = 2\pi i \sum_{l_1,\cdots,l_{j-1},l_{j+1},\cdots l_d=0}^{n-1}  \ket{l_1}\bra{l_1}\otimes\cdots\otimes\ket{l_{j-1}}\bra{l_{j-1}} \otimes\widetilde{D}_{\theta}\otimes\ket{l_{j+1}}\bra{l_{j+1}}\otimes \cdots \otimes \ket{l_d}\bra{l_d}, 
    \end{equation}
    \begin{equation}
    \widetilde{D}_{\theta} = \ket{1}\bra{1}+2\ket{2}\bra{2}+\cdots + \theta \ket{\theta}\bra{\theta} - \theta \ket{n-\theta}\bra{n-\theta} - (\theta-1)\ket{n-\theta+1}\bra{n-\theta+1} - \cdots - \ket{n-1}\bra{n-1}. 
\end{equation}
Then 
\begin{equation}
    (\otimes_{j-1}I \otimes \mathcal{F} \otimes_{d-j} I)D_{j,\theta}(\otimes_{j-1}I \otimes \mathcal{F}^{-1} \otimes_{d-j} I)\vec{f} = \sum_{l_1=0}^{n-1}\cdots \sum_{l_d=0}^{n-1} g_l \ket{l_1}\cdots \ket{l_d}
\end{equation}
where 
$$\left\|\vec{g} - \vec{\nabla_{x_j} f}\right\| \leq \frac{8\|\partial_{x_j}^p f\|_{\infty}}{\pi^{p-1}}\frac{1}{n^{p-2}} + \frac{2\sqrt{2}\|\partial_{x_j}^p f\|_{\infty}}{(2\pi)^{p-1}} \frac{n}{\theta^{p-1}}.$$
\end{lemma}
\begin{proof}
    We start with the definition of discrete Fourier transform acting on a vector $\vec{v} = \sum_l v_{l} \ket{l_1}\cdots\ket{l_d}$ that 
    \begin{equation}
        (\otimes_{j-1}I \otimes \mathcal{F}^{-1} \otimes_{d-j} I)\vec{v} = \frac{1}{\sqrt{n}} \sum_k \sum_{l_j=0}^{n-1} v_{(k_1,\cdots,k_{j-1},l_j,k_{j+1},\cdots,k_d)} \omega_n^{-k_jl_j}\ket{k_1}\cdots\ket{k_d}
    \end{equation}
    where $\omega_n = e^{2\pi i/n}$. 
    $\mathcal{F}^{-1} v$ can be interpreted as the set of discrete Fourier coefficients of the vector $v$.
    To see this, let the function $f$ allow the following complex Fourier series expansion along $j$-th direction
    \begin{equation}\label{eqn:app_F_series_1}
        f(l_1/n,\cdots,l_{j-1}/n,x,l_{j+1}/n,\cdots,l_d/n) = \sum_{m_j=-\infty}^{\infty} c_{\left\{l\right\}\backslash l_j, m_j} e^{2\pi i m_j x}. 
    \end{equation}
    Then the Fourier transform of the vector $\vec{f}$ becomes
    \begin{equation}
    \begin{split}
        (\otimes_{j-1}I \otimes \mathcal{F}^{-1} \otimes_{d-j} I)\vec{f} &= \frac{1}{\sqrt{n}} \sum_k \sum_{l_j=0}^{n-1} \sum_{m_j=-\infty}^{\infty} c_{\left\{k\right\}\backslash k_j, m_j} e^{2\pi i m_j l_j/n} \omega_n^{-k_jl_j}\ket{k_1}\cdots\ket{k_d} \\
        & = \frac{1}{\sqrt{n}} \sum_k  \sum_{m_j=-\infty}^{\infty} c_{\left\{k\right\}\backslash k_j, m_j}  \left(\sum_{l_j=0}^{n-1}\omega_n^{(m_j-k_j)l_j}\right)\ket{k_1}\cdots\ket{k_d}. 
    \end{split}
    \end{equation}
    Noticing that 
    \begin{equation}
    \left(\sum_{l_j=0}^{n-1}\omega_n^{(m_j-k_j)l_j}\right) = \begin{cases}
    n, & \text{if } m_j-k_j = qn \text{ for some integer } q, \\ 
    0, & \text{else}, 
    \end{cases}
    \end{equation}
    we have 
    \begin{equation}\label{eqn:app_fourier_coeff}
    (\otimes_{j-1}I \otimes \mathcal{F}^{-1} \otimes_{d-j} I)\vec{f} = \sqrt{n} \sum_k  \sum_{q=-\infty}^{\infty} c_{\left\{k\right\}\backslash k_j, k_j+qn}  \ket{k_1}\cdots\ket{k_d}. 
    \end{equation}
    Fourier transform of the derivative can be computed similarly. 
    Starting from the Fourier series again, we have 
    \begin{equation}
    \partial_{x_j} f(l_1/n,\cdots,l_{j-1}/n,x,l_{j+1}/n,\cdots,l_d/n) = \sum_{m_j=-\infty}^{\infty} c_{\left\{l\right\}\backslash l_j, m_j} 2\pi i m_j e^{2\pi i m_j x}. 
    \end{equation}
    Replacing the coefficient $c$ by $2\pi i m_j c$ in \eqref{eqn:app_fourier_coeff}, we have 
    \begin{equation}\label{eqn:app_fourier_coeff_deriv}
    (\otimes_{j-1}I \otimes \mathcal{F}^{-1} \otimes_{d-j} I)\vec{\partial_{x_j}f} = 2\pi i \sqrt{n} \sum_k  \sum_{q=-\infty}^{\infty} (k_j+qn)c_{\left\{k\right\}\backslash k_j, k_j+qn}  \ket{k_1}\cdots\ket{k_d}. 
    \end{equation}
    
    \eqref{eqn:app_fourier_coeff} and \eqref{eqn:app_fourier_coeff_deriv} only differ by multiplication of corresponding frequency factors, and multiplying \eqref{eqn:app_fourier_coeff} by the matrix $D_{j,\theta}$ will remove such a difference for bounded frequencies. 
    Based on this observation, for each $k = (k_1,\cdots,k_d)$, we can compute the difference between the entries as 
    \begin{equation}
        \begin{split}
            & \quad \left((\otimes_{j-1}I \otimes \mathcal{F}^{-1} \otimes_{d-j} I)\vec{\partial_{x_j}f} - D_{j,\theta} (\otimes_{j-1}I \otimes \mathcal{F}^{-1} \otimes_{d-j} I)\vec{f}\right)_k \\
            & = \begin{cases}
               2\pi i \sqrt{n} \sum_{q\neq 0} qn c_{\left\{k\right\}\backslash k_j, k_j+qn} , &\text{ if } 0 \leq k_j \leq \theta, \\
               2\pi i \sqrt{n} \sum_{q\neq -1} (q+1)n c_{\left\{k\right\}\backslash k_j, k_j+qn} , &\text{ if } n-\theta \leq k_j \leq n-1, \\
               2\pi i \sqrt{n} \sum_{q} (k_j+qn) c_{\left\{k\right\}\backslash k_j, k_j+qn} , &\text{ else}.
            \end{cases}
         \end{split}
    \end{equation}
    
    Now we study how to bound the difference here by first estimating the decay rate of the Fourier coefficients. 
    According to the definition of the Fourier coefficients, for any $m \neq 0$, 
    \begin{equation}
        c_{\left\{k\right\}\backslash k_j, m} = \int_0^1 f(k_1/n,\cdots, k_{j-1}/n,x,k_{j+1}/n,\cdots,k_d/n) e^{2\pi i m x} dx, 
    \end{equation}
    and by using integration by parts formula for $p$ times, we obtain 
    \begin{equation}
    \begin{split}
        |c_{\left\{k\right\}\backslash k_j, m}| &= \left|\left(-\frac{1}{2\pi i m}\right)^p\int_0^1 \partial_{x_j}^p f(k_1/n,\cdots, k_{j-1}/n,x,k_{j+1}/n,\cdots,k_d/n) e^{2\pi i m x} dx\right| \\
        & \leq \frac{\|\partial_{x_j}^p f\|_{\infty}}{|2\pi m|^p}. 
    \end{split}
    \end{equation}
    Then we have, for $0 \leq k_i \leq \theta$, 
    \begin{equation}
        \begin{split}
            \left|\sum_{q\neq 0} qn c_{\left\{k\right\}\backslash k_j, k_j+qn}\right| & \leq \|\partial_{x_j}^p f\|_{\infty} \left(\sum_{q = 1}^{\infty} \frac{nq}{(2\pi(qn+k_j))^p} + \sum_{q=1}^{\infty} \frac{nq}{(2\pi(qn-k_j))^p}\right)\\
            & \leq \|\partial_{x_j}^p f\|_{\infty} \left(\frac{n}{(2\pi(n+k_j))^p} + \int_1^{\infty} \frac{ny}{(2\pi(ny+k_j))^p}dy\right) \\
            & \quad + \|\partial_{x_j}^p f\|_{\infty} \left(\frac{n}{(2\pi(n-k_j))^p} + \int_1^{\infty} \frac{ny}{(2\pi(ny-k_j))^p}dy\right) \\
            & = \frac{\|\partial_{x_j}^p f\|_{\infty}}{(2\pi)^p} \left(\frac{n}{(n+k_j)^p} + \frac{1}{(p-2)n(n+k_j)^{p-2}} - \frac{k_j}{(p-1)n(n+k_j)^{p-1}}\right) \\
            & \quad + \frac{\|\partial_{x_j}^p f\|_{\infty}}{(2\pi)^p} \left(\frac{n}{(n-k_j)^p} + \frac{1}{(p-2)n(n-k_j)^{p-2}} + \frac{k_j}{(p-1)n(n-k_j)^{p-1}}\right) \\
            & \leq \frac{2\|\partial_{x_j}^p f\|_{\infty}}{\pi^p n^{p-1}},  
        \end{split}
    \end{equation}
    for $n-\theta \leq k_j \leq n-1$, 
    \begin{equation}
        \begin{split}
            \left|\sum_{q\neq -1} (q+1)n c_{\left\{k\right\}\backslash k_j, k_j+qn}\right| & \leq \frac{\|\partial_{x_j}^p f\|_{\infty}}{(2\pi)^p} \left(\sum_{q=1}^{\infty} \frac{qn}{(k_j+(q-1)n)^p} + \sum_{q=1}^{\infty} \frac{qn}{((q+1)n-k_j)^p}\right) \\
            & \leq \frac{\|\partial_{x_j}^p f\|_{\infty}}{(2\pi)^p} \left( \frac{n}{k_j^p} + \int_1^{\infty} \frac{ny}{(k_j+(y-1)n)^p} dy \right) \\
            & \quad + \frac{\|\partial_{x_j}^p f\|_{\infty}}{(2\pi)^p} \left(\frac{n}{(2n-k_j)^p} + \int_1^{\infty} \frac{ny}{((y+1)n-k_j)^p} dy\right) \\
            & = \frac{\|\partial_{x_j}^p f\|_{\infty}}{(2\pi)^p} \left( \frac{n}{k_j^p} + \frac{1}{(p-2)nk_j^{p-2}} + \frac{n-k_j}{(p-1)n k_j^{p-1}} \right) \\
            & \quad + \frac{\|\partial_{x_j}^p f\|_{\infty}}{(2\pi)^p} \left(\frac{n}{(2n-k_j)^p} + \frac{1}{(p-2)n(2n-k_j)^{p-2}} - \frac{n-k_j}{(p-1)n(2n-k_j)^{p-1}}\right) \\
            & \leq \frac{2\|\partial_{x_j}^p f\|_{\infty}}{\pi^p n^{p-1}}, 
        \end{split}
    \end{equation}
    and for $\theta+1\leq k_j \leq n-\theta-1$, 
    \begin{equation}
        \begin{split}
            \left|\sum_{q} (k_j+qn) c_{\left\{k\right\}\backslash k_j, k_j+qn}\right| & \leq \frac{\|\partial_{x_j}^p f\|_{\infty}}{(2\pi)^p} \left( \frac{1}{k_j^{p-1}} + \sum_{q=1}^{\infty} \frac{1}{(k_j+qn)^{p-1}} + \sum_{q=1}^{\infty} \frac{1}{(qn-k_j)^{p-1}} \right) \\
            & \leq \frac{\|\partial_{x_j}^p f\|_{\infty}}{(2\pi)^p} \left(\frac{1}{k_j^{p-1} } + \frac{1}{(k_j+n)^{p-1}} + \int_1^{\infty} \frac{dy}{(k_j+ny)^{p-1}}\right) \\
            & \quad + \frac{\|\partial_{x_j}^p f\|_{\infty}}{(2\pi)^p} \left( \frac{1}{(n-k_j)^{p-1}} + \int_1^{\infty} \frac{dy}{(ny-k_j)^{p-1}}\right) \\
            & = \frac{\|\partial_{x_j}^p f\|_{\infty}}{(2\pi)^p} \left(\frac{1}{k_j^{p-1} } + \frac{1}{(k_j+n)^{p-1}} + \frac{1}{(p-2)n(n+k_j)^{p-2}}\right) \\
            & \quad + \frac{\|\partial_{x_j}^p f\|_{\infty}}{(2\pi)^p} \left( \frac{1}{(n-k_j)^{p-1}} + \frac{1}{(p-2)n(n-k_j)^{p-2}}\right) \\
            & \leq \frac{\|\partial_{x_j}^p f\|_{\infty}}{(2\pi)^p} \left(\frac{2}{\theta^{p-1}} + \frac{3}{n^{p-1}}\right). 
        \end{split}
    \end{equation}
    Combining these three estimates, we can bound
    \begin{equation}
    \begin{split}
        & \quad \left\|(\otimes_{j-1}I \otimes \mathcal{F}^{-1} \otimes_{d-j} I)\vec{\partial_{x_j}f} - D_{j,\theta} (\otimes_{j-1}I \otimes \mathcal{F}^{-1} \otimes_{d-j} I)\vec{f}\right\| \\
        & \leq \sqrt{2\theta+1} \frac{4\|\partial_{x_j}^p f\|_{\infty}}{\pi^{p-1}n^{p-3/2}} + \sqrt{n-2\theta-1} \frac{\sqrt{2n}\|\partial_{x_j}^p f\|_{\infty}}{(2\pi)^{p-1}} \left(\frac{2}{\theta^{p-1}} + \frac{3}{n^{p-1}}\right) \\
        & \leq \frac{8\|\partial_{x_j}^p f\|_{\infty}}{\pi^{p-1}}\frac{1}{n^{p-2}} + \frac{2\sqrt{2}\|\partial_{x_j}^p f\|_{\infty}}{(2\pi)^{p-1}} \frac{n}{\theta^{p-1}}. 
    \end{split}
    \end{equation}
    This completes the proof by further using the fact that the quantum Fourier transform operator has unit 2-norm. 
    
\end{proof}

\begin{theorem}\label{thm:application_deriv_state_main}
    Let $f(t,x)$ be a function such that $f$ satisfies periodic boundary conditions for $x$, and its spatial partial derivatives exist and are continuous up to order $p \geq 3$. 
    Let the vector $\vec{f}$ be 
    \begin{equation}
        \vec{f} = \sum_{k=0}^{m-1}\sum_{l_1,\cdots,l_d = 0 }^{n-1} f(kT/m, l_1/n,\cdots,k_d/n) \ket{k}\ket{l_1}\cdots\ket{l_d}, 
    \end{equation}
    and the vector $\vec{g}$ be 
    \begin{equation}
    \begin{split}
        \vec{g} &= \sum_{j=0}^{d-1}\sum_{k=0}^{m-1}\sum_{l_1,\cdots,l_d=0}^{n-1} g_{j,k,l} \ket{j}\ket{k}\ket{l_1}\cdots\ket{l_d} \\
            &= \sum_{j=0}^{d-1}\sum_{k=0}^{m-1}\mathcal{F}_{j}D_{j,\theta}\mathcal{F}_{j}^{-1} \sum_{l_1,\cdots,l_d=0}^{n-1} f(kT/m,l_1/n,\cdots,l_d/n) \ket{j}\ket{k}\ket{l_1}\cdots\ket{l_d}.
    \end{split}
    \end{equation}
    Then,  
    \begin{enumerate}
        \item we have 
        \begin{equation}
        \begin{split}
        & \quad \left\|\vec{g} - \sum_{j=0}^{d-1}\sum_{k=0}^{m-1}\sum_{l_1,\cdots,l_d=0}^{n-1} \partial_{x_j} f(kT/m,l_1/n,\cdots,l_d/n) \ket{j}\ket{k}\ket{l_1}\cdots\ket{l_d}\right\| \\
        &\leq \left(\frac{8}{\pi^{p-1}}\frac{m}{n^{p-2}} + \frac{2\sqrt{2}}{(2\pi)^{p-1}} \frac{mn}{\theta^{p-1}}\right)\sum_{j=0}^{d-1} \|\partial_{x_j}^p f\|_{\infty}, 
        \end{split}
        \end{equation}
        \item for any $0 < \epsilon < 1$, $0< \delta <1$, there exists a quantum algorithm which outputs an $\epsilon$-approximation of $\vec{g}/\|\vec{g}\|$ with probability at least $(1-\delta)$, using queries to $\mathcal{A}(\epsilon/Q)$ and $O_D$ for $2Q\log(1/\delta)$ times and additional $\Or(d^3(\log n)^2)$ gates, 
        where $Q = \frac{4\pi\theta\sqrt{d}\|\vec{f}\|}{\|\vec{g}\|}$. 
    \end{enumerate}
\end{theorem}
\begin{proof}
    Let $\epsilon'$ denote the tolerated error in the algorithm $\mathcal{A}$. 
    Then, since all the operations are unitary, the obtained final quantum state (before measurement) is also an $\epsilon'$-approximation of the exact state. 
    We write the exact final state as 
    \begin{equation}
        \begin{split}
            \frac{1}{2\pi\theta \sqrt{d}\|\vec{f}\|} \sum_{j=0}^{d-1}\sum_{k=0}^{m-1}\sum_{l_1,\cdots,l_d=0}^{n-1} g_{j,k,l} \ket{0}\ket{j}\ket{k}\ket{l_1}\cdots\ket{l_d} + \ket{\perp} 
        \end{split}
    \end{equation}
    where
    \begin{equation}
             \sum_{l_1,\cdots,l_d=0}^{n-1} g_{j,k,l} \ket{l_1}\cdots\ket{l_d} \\
            = \mathcal{F}_{j}D_{j,\theta}\mathcal{F}_{j}^{-1} \sum_{l_1,\cdots,l_d=0}^{n-1} f(kT/m,l_1/n,\cdots,l_d/n) \ket{l_1}\cdots\ket{l_d}. 
    \end{equation}
    According to \lem{application_deriv_fourier_error}, we have 
    \begin{equation}
    \begin{split}
        & \quad \left\|\sum_{j=0}^{d-1}\sum_{k=0}^{m-1}\sum_{l_1,\cdots,l_d=0}^{n-1} g_{j,k,l} \ket{j}\ket{k}\ket{l_1}\cdots\ket{l_d} - \sum_{j=0}^{d-1}\sum_{k=0}^{m-1}\sum_{l_1,\cdots,l_d=0}^{n-1} \partial_{x_j} f(kT/m,l_1/n,\cdots,l_d/n) \ket{j}\ket{k}\ket{l_1}\cdots\ket{l_d}\right\| \\
        & \leq \left(\frac{8}{\pi^{p-1}}\frac{m}{n^{p-2}} + \frac{2\sqrt{2}}{(2\pi)^{p-1}} \frac{mn}{\theta^{p-1}}\right)\sum_{j=0}^{d-1} \|\partial_{x_j}^p f\|_{\infty}. 
    \end{split}
    \end{equation}
    
    It remains to estimate errors in the quantum state after successful measurement and the success probability. 
    For this purpose, we need some linear algebra results and we will state and prove here with slight off from the main proof. 
    
    \paragraph{Result} Let $\left\{e_i,f_j\right\}$ form an orthonormal basis of a Hilbert space, and let 
    $\psi = a+b, \widetilde{\psi} = \widetilde{a} + \widetilde{b}$ with $\|\psi\| = \|\widetilde{\psi}\| = 1$, $a,\widetilde{a} \in \text{span}\left\{e_i\right\}$, and $b,\widetilde{b} \in \text{span}\left\{f_j\right\}$. 
    If $\|\psi-\widetilde{\psi}\| < \epsilon$, then 
    \begin{enumerate}
        \item $\|a/\|a\| - \widetilde{a}/\|\widetilde{a}\|\| < 2\epsilon/\|a\|$, 
        \item $\|\widetilde{a}\| > \|a\| - \epsilon$. 
    \end{enumerate}
    
    This result can be straightforwardly proved by direct computations as
    \begin{equation}
        \begin{split}
            \left\|a/\|a\| - \widetilde{a}/\|\widetilde{a}\|\right\| & \leq \left\|a/\|a\| - \widetilde{a}/\|a\|\right\| + \left\|\widetilde{a}/\|a\| - \widetilde{a}/\|\widetilde{a}\|\right\| \\
            & = \frac{1}{\|a\|} \|a-\widetilde{a}\| + \|\widetilde{a}\|\left|\frac{1}{\|a\|} - \frac{1}{\|\widetilde{a}\|}\right| \\
            & = \frac{1}{\|a\|} \|a-\widetilde{a}\| + \frac{1}{\|a\|} |\|a\|-\|\widetilde{a}\|| \\
            & \leq \frac{2}{\|a\|} \|a-\widetilde{a}\| \\
            & \leq \frac{2}{\|a\|} \|\psi-\widetilde{\psi}\| \\
            & < \frac{2\epsilon}{\|a\|}, 
        \end{split}
    \end{equation}
    and 
    \begin{equation}
        \begin{split}
            \|\widetilde{a}\| \geq \|a\| - \|a-\widetilde{a}\| > \|a\| - \epsilon. 
        \end{split}
    \end{equation}
    
    The errors in the quantum state after successful measurement and the success probability can be directly bounded using this result and amplitude amplification, viewing $\psi$ as the exact state and $\widetilde{\psi}$ as the obtained state. 
    Specifically, for a single run and measurement, errors in the quantum state after successful measurement can be bounded by 
    \begin{equation}
        \frac{4\pi\theta\sqrt{d}\|\vec{f}\|\epsilon'}{\|\vec{g}\|}, 
    \end{equation}
    and can be further bounded by $\epsilon$ by choosing 
    \begin{equation}
        \epsilon' = \frac{\|\vec{g}\|\epsilon}{4\pi\theta\sqrt{d}\|\vec{f}\|}.
    \end{equation}
    The success probability for a single run, after amplitude amplification, is bounded from below by 
    \begin{equation}
        \frac{\|\vec{g}\|}{2\pi\theta\sqrt{d}\|\vec{f}\|} - \epsilon' = \frac{\|\vec{g}\|}{2\pi\theta\sqrt{d}\|\vec{f}\|}\left(1-\epsilon/2\right) \geq \frac{\|\vec{g}\|}{4\pi\theta\sqrt{d}\|\vec{f}\|}. 
    \end{equation}
    The overall probability of getting success at least once can be boosted to $(1-\delta)$ by repeating the algorithm $M$ times with 
    \begin{equation}
        M = \log(1/\delta)/\log\left(1/\left(1-\frac{\|\vec{g}\|}{4\pi\theta\sqrt{d}\|\vec{f}\|}\right)\right) \leq \frac{8\pi\theta\sqrt{d}\|\vec{f}\|}{\|\vec{g}\|}\log\left(\frac{1}{\delta}\right). 
    \end{equation}
    
\end{proof}

Finally, under the further regularity assumption specified in \thm{app_deriv_state_smooth}, we can obtain a simpler complexity estimate which shows a poly-logarithmic dependence in terms of the precision. 

\begin{proof}[Proof of \thm{app_deriv_state_smooth}]
    According to \thm{application_deriv_state_main}, the successful output of the algorithm is an $\epsilon$-approximation of $\vec{g}/\|\vec{g}\|$ where, for any $p \geq 3$, 
    \begin{equation}
        \|\vec{g} - \vec{\nabla f}\| \leq \left(\frac{8}{\pi^{p-1}}\frac{m}{n^{p-2}} + \frac{2\sqrt{2}}{(2\pi)^{p-1}} \frac{mn}{\theta^{p-1}}\right)\sum_{j=0}^{d-1} \|\partial_{x_j}^p f\|_{\infty}. 
    \end{equation}
    Let $c = \sup_{j,p} (\|\partial_{x_j}^p f\|_{\infty})^{1/p}$, and by choosing $\theta = c/\pi+1$, we have 
    \begin{equation}
    \begin{split}
        \|\vec{g} - \vec{\nabla f}\| &\leq \left(\frac{8}{\pi^{p-1}}\frac{m}{n^{p-2}} + \frac{2\sqrt{2}}{(2\pi)^{p-1}} \frac{mn}{\theta^{p-1}}\right)dc^p \\
        & \leq 8\pi dmn^2 \left(\frac{c}{\pi n}\right)^p + 4\sqrt{2} \pi d mn \theta \left(\frac{1}{2}\right)^p. 
    \end{split}
    \end{equation}
    Since $c/(\pi n) < 1$, we obtain $\vec{g} = \vec{\nabla f}$ by taking $p \rightarrow \infty$. 
    Therefore the claims in \thm{app_deriv_state_smooth} directly follow from \thm{application_deriv_state_main}. 
\end{proof}

\end{document}